\documentclass[aps,twocolumn]{revtex4-1}%
\usepackage{amsfonts}
\usepackage{amsmath}
\usepackage{amssymb}
\usepackage{graphicx}
\usepackage[colorlinks=true,linkcolor=blue,citecolor=red,plainpages=false,pdfpagelabels]%
{hyperref}
\providecommand{\U}[1]{\protect\rule{.1in}{.1in}}
\newtheorem{theorem}{Theorem}

\newtheorem{condition}{Condition}

\newtheorem{definition}{Definition}

\newtheorem{lemma}{Lemma}

\newtheorem{proposition}{Proposition}
\newtheorem{remark}{Remark}

\newenvironment{proof}[1][Proof]{\noindent\textbf{#1.} }{\ \rule{0.5em}{0.5em}}
\begin{document}
\preprint{ }
\title[ ]{Energy-constrained private and quantum capacities of quantum channels}
\author{Mark M. Wilde}
\affiliation{Hearne Institute for Theoretical Physics, Department of Physics and Astronomy,
Center for Computation and Technology, Louisiana State University, Baton
Rouge, Louisiana 70803, USA}
\author{Haoyu Qi}
\affiliation{Hearne Institute for Theoretical Physics, Department of Physics and Astronomy,
Louisiana State University, Baton Rouge, Louisiana 70803, USA}

\begin{abstract}
This paper establishes a general theory of energy-constrained quantum and
private capacities of quantum channels. We begin by defining various
energy-constrained communication tasks, including quantum communication with a
uniform energy constraint, entanglement transmission with an average energy
constraint, private communication with a uniform energy constraint, and secret
key transmission with an average energy constraint. We develop several code
conversions, which allow us to conclude non-trivial relations between the
capacities corresponding to the above tasks. We then show how the regularized,
energy-constrained coherent information is equal to the capacity for the first
two tasks and is an achievable rate for the latter two tasks, whenever the
energy observable satisfies the Gibbs condition of having a well defined
thermal state for all temperatures and the channel satisfies a finite
output-entropy condition. For degradable channels satisfying these conditions,
we find that the single-letter energy-constrained coherent information is
equal to all of the capacities. We finally apply our results to degradable
quantum Gaussian channels and recover several results already established in
the literature (in some cases, we prove new results in this domain). Contrary
to what may appear from some statements made in the literature recently,
proofs of these results do not require the solution of any kind of minimum
output entropy conjecture or entropy photon-number inequality.

\end{abstract}
\volumeyear{ }
\volumenumber{ }
\issuenumber{number}
\eid{identifier}
\date{\today}
\startpage{1}
\endpage{10}
\maketitle

\section{Introduction}

The capacity of a quantum channel to transmit quantum or private information
is a fundamental characteristic of the channel that guides the design of
practical communication protocols (see, e.g., \cite{W16}\ for a review). The
quantum capacity $Q(\mathcal{N})$\ of a quantum channel $\mathcal{N}$ is
defined as the maximum rate at which qubits can be transmitted faithfully over
many independent uses of $\mathcal{N}$, where the fidelity of transmission
tends to one in the limit as the number of channel uses tends to infinity
\cite{PhysRevA.55.1613,capacity2002shor,ieee2005dev}. Related, the private
capacity $P(\mathcal{N})$ of $\mathcal{N}$ is defined to be the maximum rate
at which classical bits can be transmitted over many independent uses of
$\mathcal{N}$ such that 1)\ the receiver can decode the classical bits
faithfully and 2)\ the environment of the channel cannot learn anything about
the classical bits being transmitted \cite{ieee2005dev,1050633}. The quantum
capacity is essential for understanding how fast we will be able to perform
distributed quantum computations between remote locations, and the private
capacity is connected to the ability to generate secret key between remote
locations, as in quantum key distribution (see, e.g., \cite{SBCDLP09}\ for a
review). Notions from classical information theory regarding wiretap channels are typically insightful for understanding private communication over quantum channels (see, e.g., \cite{W75,Tan12,H13,HES14,Hay15,TB16,YSP16,EHS16}). In general, there are connections between private capacity and
quantum capacity of quantum channels \cite{ieee2005dev} (see also \cite{PhysRevLett.80.5695}), but
the results of \cite{HHHO05,HHHO09,PhysRevLett.100.110502,HHHLO08}%
\ demonstrated that these concepts and the capacities can be very different.
In fact, the most striking examples are channels for which their quantum
capacity is equal to zero but their private capacity is strictly greater than
zero \cite{PhysRevLett.100.110502,HHHLO08}.

Bosonic Gaussian channels are some of the most important channels to consider,
as they model practical communication links in which the mediators of
information are photons (see, e.g., \cite{CEGH08,S17}\ for reviews).
Recent years have seen advances in the quantum information theory of bosonic
channels. For example, we now know the capacity for sending classical
information over all single-mode phase-insensitive quantum Gaussian channels
\cite{GHG15,GPCH13} (and even the strong converse capacity \cite{BPWW14}). The
result of this theoretical development is that coherent states \cite{GK04}\ of
the light field suffice to achieve classical capacity of phase-insensitive
bosonic Gaussian channels.
Note that the classical capacity of these channels is non-trivial only when there is an energy constraint placed on the input signaling states \cite{GHG15,GPCH13}---otherwise, it is equal infinity.

We have also seen advances related to quantum capacity of bosonic channels.
Important statements, discussions, and critical steps concerning quantum
capacity of single-mode quantum-limited attenuator and amplifier channels were
reported in \cite{HW01,WPG07}. In particular, these papers stated a formula
for the quantum capacity of these channels, whenever infinite energy is
available at the transmitter. These formulas have been supported with a proof
in \cite[Theorem~8]{PhysRevA.86.062306} and \cite{PLOB15,WTB16} (see
Remark~\ref{rem:WPG07}\ of the present paper for further discussion of this
point). However, in practice, no transmitter could ever use infinite energy to
transmit quantum information, and so the results from \cite{HW01,WPG07} have
limited applicability to realistic scenarios. Given that the notion of quantum
capacity itself is already somewhat removed from practice, as argued in
\cite{TBR15}, it seems that supplanting a sender and receiver with infinite
energy in addition to perfect quantum computers and an infinite number of
channel uses only serves to push this notion much farther away from
practice.\ One of the main aims of the present paper is to continue the effort
of bringing this notion closer to practice, by developing a general theory of
energy-constrained quantum and private communication. Considering quantum and
private capacity with a limited number of channel uses, as was done in
\cite{TBR15,WTB16}, in addition to energy constraints, is left for future developments.

In light of the above discussion, we are thus motivated to understand both
quantum and private communication over quantum channels with realistic energy
constraints. Refs.~\cite{GLMS03,GSE08} were some of the earlier works to
discuss quantum and private communication with energy constraints, in addition
to other kinds of communication tasks. The more recent efforts in
\cite{WHG11,PhysRevA.86.062306,QW16} have considered energy-constrained
communication in more general trade-off scenarios, but as special cases, they
also furnished proofs for energy-constrained quantum and private capacities of
quantum-limited attenuator and amplifier channels (see \cite[Theorem~8]%
{PhysRevA.86.062306} and \cite{QW16}). In more detail, let $Q(\mathcal{N}%
,N_{S})$ and $P(\mathcal{N},N_{S})$ denote the respective quantum and private
capacities of a quantum channel $\mathcal{N}$, such that the mean input photon
number for each channel use cannot exceed $N_{S}\in\lbrack0,\infty)$.
Ref.~\cite[Theorem~8]{PhysRevA.86.062306} established that the quantum
capacity of a pure-loss channel $\mathcal{L}_{\eta}$ with transmissivity
parameter $\eta\in\left[  0,1\right]  $ is equal to%
\begin{equation}
Q(\mathcal{L}_{\eta},N_{S})=\max\{g(\eta N_{S})-g((1-\eta)N_{S}%
),0\},\label{eq:pure-loss-capacities}%
\end{equation}
where $g(x)$ is the entropy of a thermal state with mean photon number $x$,
defined as%
\begin{equation}
g(x)\equiv(x+1)\log_{2}(x+1)-x\log_{2}x.
\end{equation}
The present paper (see \eqref{eq:q-cap-loss}) establishes the private capacity
formula for $\mathcal{L}_{\eta}$:%
\begin{equation}
P(\mathcal{L}_{\eta},N_{S})=\max\{g(\eta N_{S})-g((1-\eta)N_{S}%
),0\}.\label{eq:priv-cap-loss-new}%
\end{equation}
A special case of the results of \cite{QW16} established that the quantum and
private capacities of a quantum-limited amplifier channel $\mathcal{A}%
_{\kappa}$ with gain parameter $\kappa\in\lbrack1,\infty)$ are equal to%
\begin{align}
Q(\mathcal{A}_{\kappa},N_{S}) &  =P(\mathcal{A}_{\kappa},N_{S})\\
&  =g(\kappa N_{S}+\kappa-1)-g([\kappa-1][N_{S}+1]).\label{eq:amp-capacities}%
\end{align}
Taking the limit as $N_{S}\rightarrow\infty$, these formulas respectively
converge to%
\begin{align}
&  \max\{\log_{2}(\eta/\left[  1-\eta\right]
),0\},\label{eq:unconstrained-q-cap-loss}\\
&  \log_{2}(\kappa/\left[  \kappa-1\right]  )
\label{eq:unconstrained-q-cap-amp},
\end{align}
which were stated in \cite{HW01,WPG07} in the context of quantum capacity,
with the latter proved in \cite{PLOB15,WTB16} for both quantum and private
capacities. Figure~\ref{fig:cap-compare}\ plots the  ratios of the
unconstrained to constrained quantum capacity formulas in
\eqref{eq:unconstrained-q-cap-loss} and \eqref{eq:pure-loss-capacities},
respectively.
Figure~\ref{fig:cap-compare-amp}\ plots the ratios of the
unconstrained to constrained quantum capacity formulas in
\eqref{eq:unconstrained-q-cap-amp} and \eqref{eq:amp-capacities},
respectively.
\begin{figure}[ptb]
\begin{center}
\includegraphics[
width=3.3399in
]
{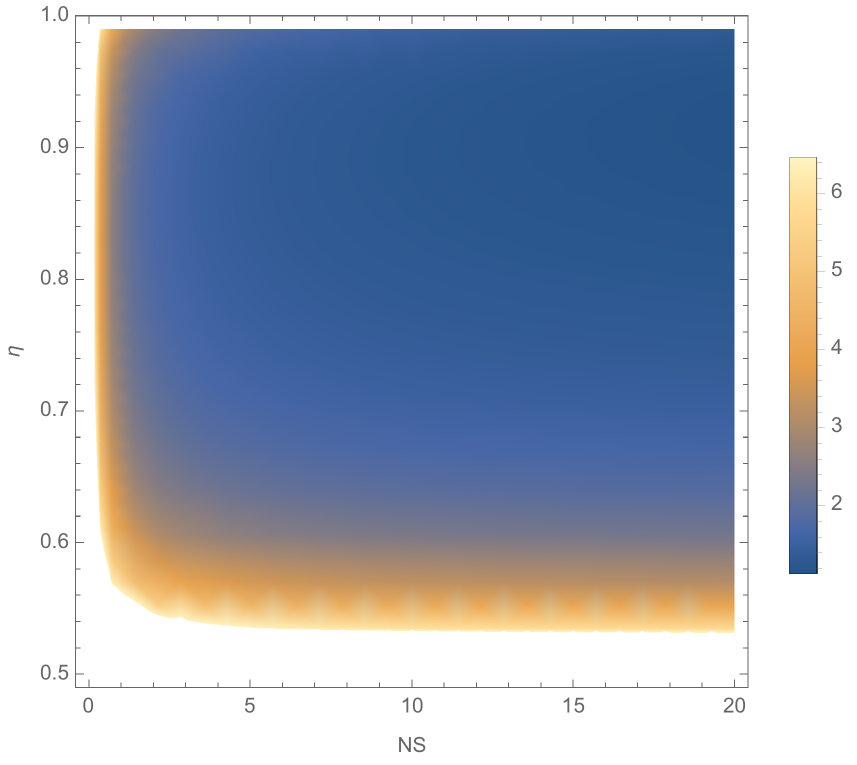}
\end{center}
\caption{Density plot of the ratio of the unconstrained  and constrained
 quantum and private capacities of the pure-loss channel for
$\eta \in [1/2, 1]$ and $N_S \in [0,20]$. For lower photon numbers and higher loss $\eta \approx 0.5$, there is a large gap between these
capacities.}%
\label{fig:cap-compare}%
\end{figure}

\begin{figure}[ptb]
\begin{center}
\includegraphics[
width=3.3399in
]
{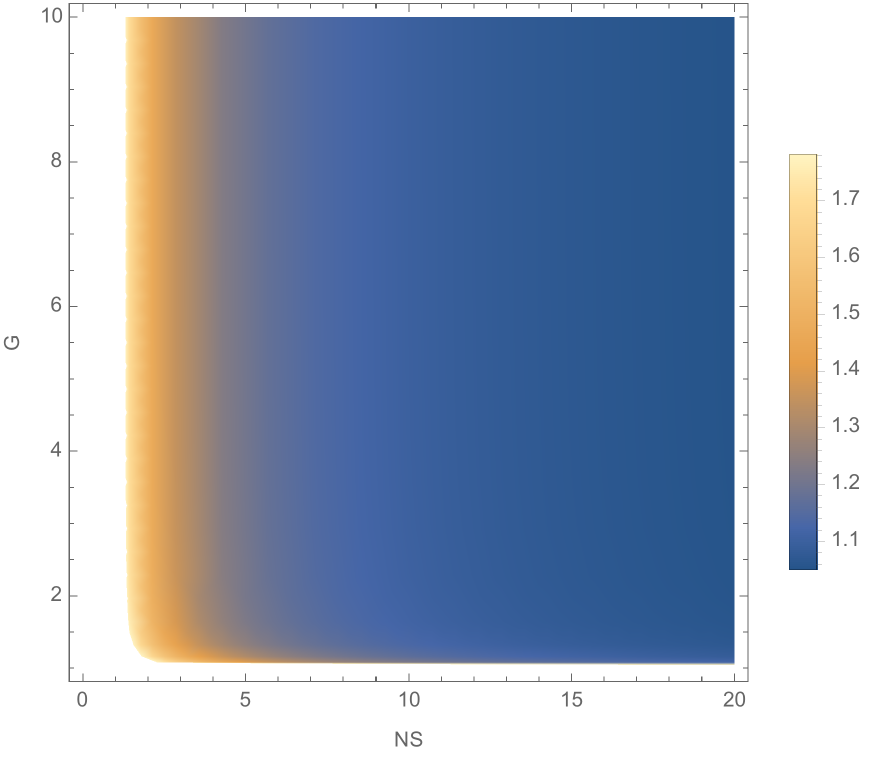}
\end{center}
\caption{Density plot of the ratio of the unconstrained  and constrained
 quantum and private capacities of the pure-amplifier channel for
$G \in [1, 10]$ and $N_S \in [0,20]$. For lower photon numbers, there is a large gap between these
capacities.}%
\label{fig:cap-compare-amp}%
\end{figure}

The main purpose of the present paper is to go beyond bosonic channels and
establish a general theory of energy-constrained quantum and private
communication over quantum channels, in a spirit similar to that developed in
\cite{H03,H04,HS12,H12} for other communication tasks. We first recall some
preliminary background on quantum information in infinite-dimensional,
separable Hilbert spaces in Section~\ref{sec:prelims}. We now summarize the main contributions of our paper:
\begin{itemize}
\item In Section~\ref{sec:energy-constrained-caps}, we define several
energy-constrained communication tasks, including quantum communication with a
uniform energy constraint, entanglement transmission with an average energy
constraint, private communication with a uniform energy constraint, and secret
key transmission with an average energy constraint.
\item In
Section~\ref{sec:code-conversions}, we develop several code conversions
between these various communication tasks, which allow us to conclude
non-trivial relations between the capacities corresponding to them, as summarized in 
Section~\ref{sec:cap-imps} and Theorem~\ref{thm:cap-relations}.

\item 
Section~\ref{sec:coh-info-ach}\ proves that the regularized,
energy-constrained coherent information is an achievable rate for all of the
tasks, whenever the energy observable satisfies the Gibbs condition of having
a well defined thermal state for all temperatures
(Definition~\ref{def:Gibbs-obs}) and the channel satisfies a finite
output-entropy condition (Condition~\ref{cond:finite-out-entropy}). This result is stated as Theorem~\ref{thm:coh-info-ach}. 

\item 
For
degradable channels satisfying the same conditions, we find in
Section~\ref{sec:degradable-channels}\ that the single-letter
energy-constrained coherent information is equal to all of the capacities (stated as Theorem~\ref{thm:-energy-constr-q-p-cap}).

\item 
Section~\ref{sec:reg-converses} establishes a regularized converse for the energy-constrained private capacity (stated as Theorem~\ref{thm:-energy-constr-p-cap-regularized}), and it also establishes that the
regularized,
energy-constrained coherent information is equal to the capacity for quantum communication with a
uniform energy constraint and entanglement transmission with an average energy
constraint, under the same conditions on the energy observable and the channel.
This latter result is stated as Theorem~\ref{thm:-energy-constr-q-cap-regularized}.

\item
We
finally apply our results to quantum Gaussian channels in
Section~\ref{sec:Gaussian-results}\ and recover several results already
established in the literature on Gaussian quantum information. In some cases,
we establish new results, like the formula for private capacity in~\eqref{eq:priv-cap-loss-new}. 

\item In Section~\ref{sec:non-Gaussian},
we discuss how our general framework, along with recent developments in \cite{SWAT17}, allow for concluding estimates for the energy-constrained private and quantum capacities of particular non--Gaussian channels. Therein, we also consider alternative energy constraints for the pure-loss and quantum-limited amplifier channels, and we bound the capacities in these settings.
\end{itemize}
\noindent We conclude in Section~\ref{sec:conclusion}%
\ with a summary and some open questions.

We would like to suggest that our contribution on this topic is timely. At the
least, we think it should be a useful resource for the community of
researchers working on related topics to have such a formalism and associated
results written down explicitly, even though a skeptic might argue that they
have been part of the folklore of quantum information theory for many years
now. To support our viewpoint, we note that some statements made in several
papers released in the past few years suggest that energy-constrained quantum
and private capacities have not been sufficiently clarified in the existing
literature. For example, in \cite{WSCW12}, one of the main results contributed
was a non-tight upper bound on the private capacity of a pure-loss bosonic
channel, in spite of the fact that \eqref{eq:priv-cap-loss-new} was already
part of the folklore of quantum information theory. In \cite{PMLG15}, it is
stated that the \textquotedblleft entropy photon-number inequality turns out
to be crucial in the determining the classical capacity regions of the quantum
bosonic broadcast and wiretap channels,\textquotedblright\ in spite of the
fact that no such argument is needed to establish the quantum or private
capacity of the pure-loss channel. Similarly, it is stated in \cite{ADO16}%
\ that the entropy photon-number inequality \textquotedblleft conjecture is of
particular significance in quantum information theory since if it were true
then it would allow one to evaluate classical capacities of various bosonic
channels, e.g. the bosonic broadcast channel and the wiretap
channel.\textquotedblright\ Thus, it seems timely and legitimate to confirm
that no such entropy photon-number inequality or minimum output-entropy
conjecture is necessary in order to establish the results regarding quantum or
private capacity of the pure-loss channel---the existing literature
(specifically, \cite[Theorem~8]{PhysRevA.86.062306} and now the previously
folklore \eqref{eq:q-cap-loss}) has established these capacities. The same is
the case for the quantum-limited amplifier channel due to the results
of~\cite{QW16}. The entropy photon-number inequality indeed implies formulas
for quantum and private capacities of the quantum-limited attenuator and
amplifier channels, but it appears to be much stronger than what is actually
necessary to accomplish this goal. The different proof of these formulas that
we give in the present paper (see Section~\ref{sec:Gaussian-results})\ is
based on the monotonicity of quantum relative entropy, concavity of coherent
information of degradable channels with respect to the input density operator,
and covariance of Gaussian channels with respect to displacement operators.

\section{Quantum information preliminaries\label{sec:prelims}}

\subsection{Quantum states and channels}

Background on quantum information in infinite-dimensional systems is available
in \cite{H12} (see also \cite{H04,SH08,HS10,HZ12,S15,S15squashed}). We review
some aspects here. We use $\mathcal{H}$ throughout the paper to denote a
separable Hilbert space, unless specified otherwise. Let $I_{\mathcal{H}}$
denote the identity operator acting on $\mathcal{H}$. Let $\mathcal{B}%
(\mathcal{H})$ denote the set of bounded linear operators acting on
$\mathcal{H}$, and let $\mathcal{P}(\mathcal{H})$ denote the subset of
$\mathcal{B}(\mathcal{H})$ that consists of positive semi-definite operators.
Let $\mathcal{T}(\mathcal{H})$ denote the set of trace-class operators, those
operators $A$ for which the trace norm is finite:$\ \left\Vert A\right\Vert
_{1}\equiv\operatorname{Tr}\{\left\vert A\right\vert \}<\infty$, where
$\left\vert A\right\vert \equiv\sqrt{A^{\dag}A}$. The Hilbert-Schmidt norm of
$A$ is defined as $\left\Vert A\right\Vert _{2}\equiv\sqrt{\operatorname{Tr}%
\{A^{\dag}A\}}$. Let $\mathcal{D}(\mathcal{H})$ denote the set of density
operators (states), which consists of the positive semi-definite, trace-class
operators with trace equal to one. A state $\rho\in\mathcal{D}(\mathcal{H})$
is pure if there exists a unit vector $|\psi\rangle\in\mathcal{H}$ such that
$\rho=|\psi\rangle\langle\psi|$. Every density operator $\rho\in
\mathcal{D}(\mathcal{H})$\ has a spectral decomposition in terms of some
countable, orthonormal basis $\{|\phi_{k}\rangle\}_{k}$ as%
\begin{equation}
\rho=\sum_{k}p(k)|\phi_{k}\rangle\langle\phi_{k}|,
\end{equation}
where $p(k)$ is a probability distribution. The tensor product of two Hilbert
spaces $\mathcal{H}_{A}$ and $\mathcal{H}_{B}$ is denoted by $\mathcal{H}%
_{A}\otimes\mathcal{H}_{B}$ or $\mathcal{H}_{AB}$.\ Given a multipartite
density operator $\rho_{AB}\in\mathcal{D}(\mathcal{H}_{A}\otimes
\mathcal{H}_{B})$, we unambiguously write $\rho_{A}=\operatorname{Tr}%
_{\mathcal{H}_{B}}\left\{  \rho_{AB}\right\}  $ for the reduced density
operator on system $A$. Every density operator $\rho$\ has a purification
$|\phi^{\rho}\rangle\in\mathcal{H}^{\prime}\otimes\mathcal{H}$, for an
auxiliary Hilbert space $\mathcal{H}^{\prime}$, where $\left\Vert |\phi^{\rho
}\rangle\right\Vert _{2}=1$ and $\operatorname{Tr}_{\mathcal{H}^{\prime}%
}\{|\phi^{\rho}\rangle\langle\phi^{\rho}|\}=\rho$. All purifications are
related by an isometry acting on the purifying system. A state $\rho_{RA}%
\in\mathcal{D}(\mathcal{H}_{R}\otimes\mathcal{H}_{A})$ extends $\rho_{A}%
\in\mathcal{D}(\mathcal{H}_{A})$ if $\operatorname{Tr}_{\mathcal{H}_{R}}%
\{\rho_{RA}\}=\rho_{A}$. We also say that $\rho_{RA}$ is an extension of
$\rho_{A}$. In what follows, we abbreviate notation like $\operatorname{Tr}%
_{\mathcal{H}_{R}}$ as $\operatorname{Tr}_{R}$.

For finite-dimensional Hilbert spaces $\mathcal{H}_{R}$ and $\mathcal{H}_{S}$
such that $\dim(\mathcal{H}_{R})=\dim(\mathcal{H}_{S})\equiv M$, we define the
maximally entangled state $\Phi_{RS}\in\mathcal{D}(\mathcal{H}_{R}%
\otimes\mathcal{H}_{S})$ of Schmidt rank $M$ as%
\begin{equation}
\Phi_{RS}\equiv\frac{1}{M}\sum_{m,m^{\prime}}|m\rangle\langle m^{\prime}%
|_{R}\otimes|m\rangle\langle m^{\prime}|_{S},
\end{equation}
where $\{|m\rangle\}_{m}$ is an orthonormal basis for $\mathcal{H}_{R}$ and
$\mathcal{H}_{S}$. We define the maximally correlated state $\overline{\Phi
}_{RS}\in\mathcal{D}(\mathcal{H}_{R}\otimes\mathcal{H}_{S})$ as%
\begin{equation}
\overline{\Phi}_{RS}\equiv\frac{1}{M}\sum_{m}|m\rangle\langle m|_{R}%
\otimes|m\rangle\langle m|_{S},
\end{equation}
which can be understood as arising by applying a completely dephasing channel
$\sum_{m}|m\rangle\langle m|(\cdot)|m\rangle\langle m|$ to either system $R$
or $S$ of the maximally entangled state $\Phi_{RS}$. We define the maximally
mixed state of system $S$ as $\pi_{S}\equiv I_{S}/M$.

A quantum channel $\mathcal{N}:\mathcal{T}(\mathcal{H}_{A})\rightarrow
\mathcal{T}(\mathcal{H}_{B})$\ is a completely positive, trace-preserving
linear map. The Stinespring dilation theorem \cite{S55} implies that there
exists another Hilbert space $\mathcal{H}_{E}$ and a linear isometry
$U:\mathcal{H}_{A}\rightarrow\mathcal{H}_{B}\otimes\mathcal{H}_{E}$ such that
for all $\tau\in\mathcal{T}(\mathcal{H}_{A})$%
\begin{equation}
\mathcal{N}(\tau)=\operatorname{Tr}_{E}\{U\tau U^{\dag}\}.
\end{equation}
The Stinespring representation theorem also implies that every quantum channel
has a Kraus representation with a countable set $\{K_{l}\}_{l}$ of bounded
Kraus operators:%
\begin{equation}
\mathcal{N}(\tau)=\sum_{l}K_{l}\tau K_{l}^{\dag},
\end{equation}
where $\sum_{l}K_{l}^{\dag}K_{l}=I_{\mathcal{H}_{A}}$. The Kraus operators are
defined by the relation%
\begin{equation}
\langle\varphi|K_{l}|\psi\rangle=\langle\varphi|\otimes\langle l|U|\psi
\rangle,
\end{equation}
for $|\varphi\rangle\in\mathcal{H}_{B}$, $|\psi\rangle\in\mathcal{H}_{A}$, and
$\{|l\rangle\}_{l}$ some orthonormal basis for $\mathcal{H}_{E}$ \cite{S13}.

A complementary channel $\mathcal{\hat{N}}:\mathcal{T}(\mathcal{H}%
_{A})\rightarrow\mathcal{T}(\mathcal{H}_{E})$ of $\mathcal{N}$ is defined for
all $\tau\in\mathcal{T}(\mathcal{H}_{A})$ as%
\begin{equation}
\mathcal{\hat{N}}(\tau)=\operatorname{Tr}_{B}\{U\tau U^{\dag}\}.
\end{equation}
Complementary channels are unique up to partial isometries acting on the
Hilbert space $\mathcal{H}_{E}$.

A quantum channel $\mathcal{N}:\mathcal{T}(\mathcal{H}_{A})\rightarrow
\mathcal{T}(\mathcal{H}_{B})$ is degradable \cite{cmp2005dev} if there exists
a quantum channel $\mathcal{D}:\mathcal{T}(\mathcal{H}_{B})\rightarrow
\mathcal{T}(\mathcal{H}_{E})$, called a degrading channel, such that for some
complementary channel $\mathcal{\hat{N}}:\mathcal{T}(\mathcal{H}%
_{A})\rightarrow\mathcal{T}(\mathcal{H}_{E})$ and all $\tau\in\mathcal{T}%
(\mathcal{H}_{A})$:%
\begin{equation}
\mathcal{\hat{N}}(\tau)=(\mathcal{D}\circ\mathcal{N})(\tau).
\end{equation}

A positive operator-valued measure (POVM) is a set $\{ \Lambda^x \}_x$ of positive semi-definite operators acting on a Hilbert space $\mathcal{H}$ such that $\sum_x \Lambda^x = I_{\mathcal{H}}$.

\subsection{Quantum fidelity and trace distance}

The fidelity of two quantum states $\rho,\sigma\in\mathcal{D}(\mathcal{H})$ is
defined as \cite{U76}%
\begin{equation}
F(\rho,\sigma)\equiv\left\Vert \sqrt{\rho}\sqrt{\sigma}\right\Vert _{1}^{2}.
\end{equation}
Uhlmann's theorem is the statement that the fidelity has the following
alternate expression as a probability overlap~\cite{U76}:%
\begin{equation}
F(\rho,\sigma)=\sup_{U}\left\vert \langle\phi^{\rho}|U\otimes I_{\mathcal{H}%
}|\phi^{\sigma}\rangle\right\vert ^{2}, \label{eq:uhlmann-fidelity}%
\end{equation}
where $|\phi^{\rho}\rangle\in\mathcal{H}^{\prime}\otimes\mathcal{H}$ and
$|\phi^{\sigma}\rangle\in\mathcal{H}^{\prime\prime}\otimes\mathcal{H}$ are
fixed purifications of $\rho$ and $\sigma$, respectively, and the optimization
is with respect to all partial isometries $U:\mathcal{H}^{\prime\prime
}\rightarrow\mathcal{H}^{\prime}$. The fidelity is non-decreasing with respect
to a quantum channel $\mathcal{N}:\mathcal{T}(\mathcal{H}_{A})\rightarrow
\mathcal{T}(\mathcal{H}_{B})$, in the sense that for all $\rho,\sigma
\in\mathcal{D}(\mathcal{H}_{A})$:%
\begin{equation}
F(\mathcal{N}(\rho),\mathcal{N}(\sigma))\geq F(\rho,\sigma).
\label{eq:fid-mono}%
\end{equation}
A simple modification of Uhlmann's theorem, found by combining
\eqref{eq:uhlmann-fidelity}\ with the monotonicity property in
\eqref{eq:fid-mono}, implies that for a given extension $\rho_{AB}$ of
$\rho_{A}$, there exists an extension $\sigma_{AB}$ of $\sigma_{A}$ such that%
\begin{equation}
F(\rho_{AB},\sigma_{AB})=F(\rho_{A},\sigma_{A}). \label{eq:Uhlmann-4-exts}%
\end{equation}

The trace distance between states $\rho$ and $\sigma$ is defined as
$\left\Vert \rho-\sigma\right\Vert _{1}$. One can normalize the trace distance
by multiplying it by $1/2$ so that the resulting quantity lies in the interval
$\left[  0,1\right]  $. The trace distance obeys a direct-sum property:\ for
an orthonormal basis $\{|x\rangle\}_{x}$ for an auxiliary Hilbert space
$\mathcal{H}_{X}$, probability distributions $p(x)$ and $q(x)$, and sets
$\left\{  \rho^{x}\right\}  _{x}$ and $\left\{  \sigma^{x}\right\}  _{x}$ of
states in $\mathcal{D}(\mathcal{H}_{B})$, which realize classical--quantum
states%
\begin{align}
\rho_{XB}  &  \equiv\sum_{x}p(x)|x\rangle\langle x|_{X}\otimes\rho_{B}%
^{x},\label{eq:cq-rho}\\
\sigma_{XB}  &  \equiv\sum_{x}q(x)|x\rangle\langle x|_{X}\otimes\sigma_{B}%
^{x}, \label{eq:cq-sigma}%
\end{align}
the following holds%
\begin{equation}
\left\Vert \rho_{XB}-\sigma_{XB}\right\Vert _{1}=\sum_{x}\left\Vert
p(x)\rho_{B}^{x}-q(x)\sigma_{B}^{x}\right\Vert _{1}. \label{eq:direct-sum-TD}%
\end{equation}
The trace distance is monotone non-increasing with respect to a quantum
channel $\mathcal{N}:\mathcal{T}(\mathcal{H}_{A})\rightarrow\mathcal{T}%
(\mathcal{H}_{B})$, in the sense that for all $\rho,\sigma\in\mathcal{D}%
(\mathcal{H}_{A})$:%
\begin{equation}
\left\Vert \mathcal{N}(\rho)-\mathcal{N}(\sigma)\right\Vert _{1}\leq\left\Vert
\rho-\sigma\right\Vert _{1}.
\end{equation}
The following equality holds for any two pure states $\phi,\psi\in
\mathcal{D}(\mathcal{H})$:%
\begin{equation}
\frac{1}{2}\left\Vert \phi-\psi\right\Vert _{1}=\sqrt{1-F(\phi,\psi)}.
\label{eq:Td-fid-pure}%
\end{equation}
For any two arbitrary states $\rho,\sigma\in\mathcal{D}(\mathcal{H})$, the
following inequalities hold%
\begin{equation}
1-\sqrt{F(\rho,\sigma)}\leq\frac{1}{2}\left\Vert \rho-\sigma\right\Vert
_{1}\leq\sqrt{1-F(\rho,\sigma)}. \label{eq:F-vd-G}%
\end{equation}
The inequality on the left is a consequence of the Powers-Stormer inequality
\cite[Lemma~4.1]{powers1970}, which states that $\left\Vert P-Q\right\Vert
_{1}\geq\left\Vert P^{1/2}-Q^{1/2}\right\Vert _{2}^{2}$ for $P,Q\in
\mathcal{P}(\mathcal{H})$. The inequality on the right follows from the
monotonicity of trace distance with respect to quantum channels, the identity
in \eqref{eq:Td-fid-pure}, and Uhlmann's theorem in
\eqref{eq:uhlmann-fidelity}. These inequalities are called Fuchs-van-de-Graaf
inequalities, as they were established in \cite{FG98} for finite-dimensional states.

\subsection{Quantum entropies and information}

The quantum entropy of a state $\rho\in\mathcal{D}(\mathcal{H})$ is defined as%
\begin{equation}
H(\rho)\equiv\operatorname{Tr}\{\eta(\rho)\},
\end{equation}
where $\eta(x)=-x\log_{2}x$ if $x>0$ and $\eta(0)=0$. The trace in the above
equation can be taken with respect to any countable orthonormal basis of
$\mathcal{H}$ \cite[Definition~2]{AL70}. The quantum entropy is a
non-negative, concave, lower semicontinuous function on $\mathcal{D}%
(\mathcal{H})$ \cite{W76}. It is also not necessarily finite (see, e.g.,
\cite{BV13}). When $\rho_{A}$ is assigned to a system $A$, we write
$H(A)_{\rho}\equiv H(\rho_{A})$.

The quantum relative entropy $D(\rho\Vert\sigma)$ of $\rho,\sigma
\in\mathcal{D}(\mathcal{H})$ is defined as \cite{F70,Lindblad1973}%
\begin{align}
&  D(\rho\Vert\sigma)\nonumber\\
&  \equiv\lbrack\ln2]^{-1}\sum_{i,j}|\langle\phi_{i}|\psi_{j}\rangle
|^{2}[p(i)\ln\!\left(  \frac{p(i)}{q(j)}\right)
+q(j)-p(i)],\label{eq:rel-ent-sep}%
\end{align}
where $\rho=\sum_{i}p(i)|\phi_{i}\rangle\langle\phi_{i}|$ and $\sigma=\sum
_{j}q(j)|\psi_{j}\rangle\langle\psi_{j}|$ are spectral decompositions of
$\rho$ and $\sigma$ with $\{|\phi_{i}\rangle\}_{i}$ and $\{|\psi_{j}%
\rangle\}_{j}$ orthonormal bases. The prefactor $[\ln2]^{-1}$ is there to
ensure that the units of the quantum relative entropy are bits. We take the
convention in \eqref{eq:rel-ent-sep} that $0\ln0=0\ln\!\left(  \frac{0}%
{0}\right)  =0$ but $\ln\!\left(  \frac{c}%
{0}\right) = +\infty$ for $c>0$. Each term in the sum in \eqref{eq:rel-ent-sep} is
non-negative due to the inequality
\begin{equation}
x\ln(x/y)+y-x\geq0
\end{equation}
holding for all $x,y\geq0$ \cite{F70}. Thus, by Tonelli's theorem, the sums in \eqref{eq:rel-ent-sep}
may be taken in either order as discussed in \cite{F70,Lindblad1973}, and it
follows that $D(\rho\Vert\sigma)\geq0$ for all $\rho,\sigma\in\mathcal{D}%
(\mathcal{H})$, with equality holding if and only if $\rho=\sigma$ \cite{F70}.
If the support of $\rho$ is not contained in the support of $\sigma$, then
$D(\rho\Vert\sigma)=+\infty$. The converse statement need not hold in general:
there exist $\rho,\sigma\in\mathcal{D}(\mathcal{H})$ with the support of
$\rho$ contained in the support of $\sigma$ such that $D(\rho\Vert
\sigma)=+\infty$. For example, take $\rho$ and $\sigma$ diagonal in the same
basis with the eigenvalues of $\rho$ as in \cite[Eq.~(7)]{BV13} and those of
$\sigma$ as $\propto1/n^{2}$ for $n\geq\lceil e\rceil$.

One of the most important properties of the quantum relative entropy
$D(\rho\Vert\sigma)$ is that it is monotone with respect to a quantum channel
$\mathcal{N}:\mathcal{T}(\mathcal{H}_{A})\rightarrow\mathcal{T}(\mathcal{H}%
_{B})$ \cite{Lindblad1975}:%
\begin{equation}
D(\rho\Vert\sigma)\geq D(\mathcal{N}(\rho)\Vert\mathcal{N}(\sigma)).
\label{eq:mono-rel-ent}%
\end{equation}

The quantum mutual information $I(A;B)_{\rho}$ of a bipartite state $\rho
_{AB}\in\mathcal{D}(\mathcal{H}_{A}\otimes\mathcal{H}_{B})$ is defined as
\cite{Lindblad1973}%
\begin{equation}
I(A;B)_{\rho}=D(\rho_{AB}\Vert\rho_{A}\otimes\rho_{B}),
\end{equation}
and obeys the bound \cite{Lindblad1973}%
\begin{equation}
I(A;B)_{\rho}\leq2\min\{H(A)_{\rho},H(B)_{\rho}\}.
\end{equation}
The coherent information $I(A\rangle B)_{\rho}$ of $\rho_{AB}$ is defined as
\cite{HS10,K11}%
\begin{equation}
I(A\rangle B)_{\rho}\equiv I(A;B)_{\rho}-H(A)_{\rho}, \label{eq:coh-info-def}%
\end{equation}
when $H(A)_{\rho}<\infty$. This expression reduces to%
\begin{equation}
I(A\rangle B)_{\rho}=H(B)_{\rho}-H(AB)_{\rho}%
\end{equation}
if $H(B)_{\rho}<\infty$ \cite{HS10,K11}.

The mutual information of a quantum channel $\mathcal{N}:\mathcal{T}%
(\mathcal{H}_{A})\rightarrow\mathcal{T}(\mathcal{H}_{B})$\ with respect to a
state $\rho\in\mathcal{D}(\mathcal{H}_{A})$ is defined as \cite{HS10}%
\begin{equation}
I(\rho,\mathcal{N})\equiv I(R;B)_{\omega},
\end{equation}
where $\omega_{RB}\equiv(\operatorname{id}_{R}\otimes\mathcal{N}_{A\rightarrow
B})(\psi_{RA}^{\rho})$ and $\psi_{RA}^{\rho}\in\mathcal{D}(\mathcal{H}%
_{R}\otimes\mathcal{H}_{A})$ is a purification of $\rho$, with $\mathcal{H}%
_{R}\simeq\mathcal{H}_{A}$. The coherent information of a quantum channel
$\mathcal{N}:\mathcal{T}(\mathcal{H}_{A})\rightarrow\mathcal{T}(\mathcal{H}%
_{B})$\ with respect to a state $\rho\in\mathcal{D}(\mathcal{H}_{A})$ is
defined as \cite{HS10}%
\begin{equation}
I_{c}(\rho,\mathcal{N})\equiv I(R\rangle B)_{\omega}%
,\label{eq:coh-info-ch-def}%
\end{equation}
with $\omega_{RB}$ defined as above. These quantities obey a data processing
inequality, which is that for a quantum channel $\mathcal{M}:\mathcal{T}%
(\mathcal{H}_{B})\rightarrow\mathcal{T}(\mathcal{H}_{C})$ and $\rho$ and
$\mathcal{N}$ as before, the following holds \cite{HS10}%
\begin{align}
I(\rho,\mathcal{N}) &  \geq I(\rho,\mathcal{M}\circ\mathcal{N}),\\
I_{c}(\rho,\mathcal{N}) &  \geq I_{c}(\rho,\mathcal{M}\circ\mathcal{N}%
).\label{eq:coh-info-DP}%
\end{align}

We require the following proposition for some of the developments in this paper:

\begin{proposition}
\label{prop:concave-degrad}Let $\mathcal{N}$ be a degradable quantum channel
and $\mathcal{\hat{N}}$ a complementary channel for it. Let $\rho_{0}$ and
$\rho_{1}$ be states and let $\rho_{\lambda}=\lambda\rho_{0}+(1-\lambda
)\rho_{1}$ for $\lambda\in\left[  0,1\right]  $. Suppose that the entropies
$H(\rho_{\lambda})$ and $H(\mathcal{N}(\rho_{\lambda}))$ are finite for all $\lambda\in\left[  0,1\right]  $. Then
the coherent information of $\mathcal{N}$\ is concave with respect to these
inputs, in the sense that%
\begin{equation}
\lambda I_{c}(\rho_{0},\mathcal{N})+(1-\lambda)I_{c}(\rho_{1},\mathcal{N})\leq
I_{c}(\rho_{\lambda},\mathcal{N}).
\end{equation}

\end{proposition}

\begin{proof}
This was established for the finite-dimensional case in \cite{YHD05MQAC}. We
follow the proof given in \cite[Theorem~13.5.2]{W16}. First note that $H(\rho_{\lambda})$ and $H(\mathcal{N}(\rho_{\lambda}))$ being finite for all $\lambda \in [0,1]$ imply that
$H(\mathcal{\hat{N}}(\rho_{\lambda}))$ is finite, by an application of the isometric invariance of the entropy, the Stinespring dilation theorem, and the entropy triangle inequality from \cite[Theorem~2]{AL70}, allowing us to conclude that
\begin{equation}
H(\mathcal{\hat{N}}(\rho_{\lambda})) \leq
H(\rho_{\lambda}) + H(\mathcal{N}(\rho_{\lambda})).
\end{equation}
Set $\overline{\lambda
}\equiv1-\lambda$. Consider that%
\begin{multline}
I_{c}(\rho_{\lambda},\mathcal{N})-\lambda I_{c}(\rho_{0},\mathcal{N}%
)-\overline{\lambda}I_{c}(\rho_{1},\mathcal{N})\\
=H(\mathcal{N}(\rho_{\lambda}))-H(\mathcal{\hat{N}}(\rho_{\lambda}))-\lambda
H(\mathcal{N}(\rho_{0}))\\
+\lambda H(\mathcal{\hat{N}}(\rho_{0}))-\overline{\lambda}H(\mathcal{N}%
(\rho_{1}))+\overline{\lambda}H(\mathcal{\hat{N}}(\rho_{1})).
\end{multline}
Defining the states%
\begin{align}
\rho_{UB}  &  =\lambda|0\rangle\langle0|_{U}\otimes\mathcal{N}(\rho
_{0})+\overline{\lambda}|1\rangle\langle1|_{U}\otimes\mathcal{N}(\rho_{1}),\\
\sigma_{UE}  &  =\lambda|0\rangle\langle0|_{U}\otimes\mathcal{\hat{N}}%
(\rho_{0})+\overline{\lambda}|1\rangle\langle1|_{U}\otimes\mathcal{\hat{N}%
}(\rho_{1}),
\end{align}
we can then rewrite the last line above as%
\begin{equation}
I(U;B)_{\rho}-I(U;E)_{\sigma}.
\end{equation}
This quantity is non-negative from data processing of mutual information
because we can apply the degrading channel $\mathcal{D}_{B\rightarrow E}$ to
system $B$ of $\rho_{UB}$ and recover $\sigma_{UE}$:%
\begin{equation}
\sigma_{UE}=\mathcal{D}_{B\rightarrow E}(\rho_{UB}).
\end{equation}
This concludes the proof.
\end{proof}

The conditional quantum mutual information (CQMI) of a finite-dimensional
tripartite state $\rho_{ABC}$ is defined as%
\begin{equation}
I(A;B|C)_{\rho}\equiv H(AC)_{\rho}+H(BC)_{\rho}-H(ABC)_{\rho}-H(C)_{\rho}.
\end{equation}
In the general case, it is defined as \cite{S15,S15squashed}%
\begin{multline}
I(A;B|C)_{\rho}\equiv\\
\sup_{P_{A}}\left\{  I(A;BC)_{Q\rho Q}-I(A;C)_{Q\rho Q}:Q=P_{A}\otimes
I_{BC}\right\}  ,
\end{multline}
where the supremum is with respect to all finite-rank projections $P_{A}%
\in\mathcal{B}(\mathcal{H}_{A})$ and we take the convention as in \cite{S15,S15squashed} that
$I(A;BC)_{Q\rho Q}=\lambda I(A;BC)_{Q\rho Q/\lambda}$ where $\lambda
=\operatorname{Tr}\{Q\rho_{ABC}Q\}$. The above definition guarantees that many
properties of CQMI\ in finite dimensions carry over to the general case
\cite{S15,S15squashed}. In particular, the following chain rule holds for a
four-party state $\rho_{ABCD}\in\mathcal{D}(\mathcal{H}_{ABCD})$:%
\begin{equation}
I(A;BC|D)_{\rho}=I(A;C|D)_{\rho}+I(A;B|CD)_{\rho}.
\end{equation}

Fano's inequality \cite{F61} is the statement that for random variables $X$ and $Y$ with
alphabets $\mathcal{X}$ and $\mathcal{Y}$, respectively, the following
inequality holds%
\begin{equation}
H(X|Y)\leq\varepsilon\log_{2}(\left\vert \mathcal{X}\right\vert -1)+h_{2}%
(\varepsilon), \label{eq:fano}%
\end{equation}
where%
\begin{align}
\varepsilon &  \equiv\Pr\{X\neq Y\},\\
h_{2}(\varepsilon)  &  \equiv-\varepsilon\log_{2}\varepsilon-(1-\varepsilon
)\log_{2}(1-\varepsilon).
\end{align}
Observe that $\lim_{\varepsilon\rightarrow0}h_{2}(\varepsilon)=0$. Let
$\rho_{AB},\sigma_{AB}\in\mathcal{D}(\mathcal{H}_{A}\otimes\mathcal{H}_{B})$
with $\dim(\mathcal{H}_{A})<\infty$, $\varepsilon\in\left[  0,1\right]  $, and
suppose that $\left\Vert \rho_{AB}-\sigma_{AB}\right\Vert _{1}/2\leq
\varepsilon$. The Alicki--Fannes--Winter (AFW) inequality is as follows
\cite{AF04,Winter15}:%
\begin{equation}
\left\vert H(A|B)_{\rho}-H(A|B)_{\sigma}\right\vert \leq2\varepsilon\log
_{2}\dim(\mathcal{H}_{A})+g(\varepsilon), \label{eq:AFW-ineq-CE}%
\end{equation}
where%
\begin{equation}
g(\varepsilon)\equiv\left(  \varepsilon+1\right)  \log_{2}\left(
\varepsilon+1\right)  -\varepsilon\log_{2}\varepsilon. \label{eq:g-function}%
\end{equation}
Observe that $\lim_{\varepsilon\rightarrow0}g(\varepsilon)=0$. If the states
are classical on the first system, as in
\eqref{eq:cq-rho}--\eqref{eq:cq-sigma}, and $\dim(\mathcal{H}_{X})<\infty$ and
$\left\Vert \rho_{XB}-\sigma_{XB}\right\Vert _{1}/2\leq\varepsilon$, then the
inequality can be strengthened to \cite[Theorem~11.10.3]{W16}%
\begin{equation}
\left\vert H(X|B)_{\rho}-H(X|B)_{\sigma}\right\vert \leq\varepsilon\log
_{2}\dim(\mathcal{H}_{X})+g(\varepsilon). \label{eq:AFW-cq}%
\end{equation}

\section{Energy-constrained quantum and private capacities}

\label{sec:energy-constrained-caps}In this section, we define various notions
of energy-constrained quantum and private capacity of quantum channels. We
start by defining an energy observable (see \cite[Definition~11.3]{H12}):

\begin{definition}
[Energy observable]\label{def:energy-obs}Let $G$ be a positive semi-definite
operator, i.e., $G\in\mathcal{P}(\mathcal{H}_{A})$. Throughout, we refer to
$G$ as an energy observable. In more detail, we define $G$ as follows: let
$\{|e_{j}\rangle\}_{j}$ be an orthonormal basis for a Hilbert space
$\mathcal{H}$, and let $\{g_{j}\}_{j}$ be a sequence of non-negative real
numbers bounded from below. Then the following formula%
\begin{equation}
G|\psi\rangle=\sum_{j=1}^{\infty}g_{j}|e_{j}\rangle\langle e_{j}|\psi\rangle
\end{equation}
defines a self-adjoint operator $G$ on the dense domain $\{|\psi\rangle
:\sum_{j=1}^{\infty}g_{j}^{2}\left\vert \left\langle e_{j}|\psi\right\rangle
\right\vert ^{2}<\infty\}$, for which $|e_{j}\rangle$ is an eigenvector with
corresponding eigenvalue $g_{j}$.
\end{definition}

\noindent For a state $\rho\in\mathcal{D}(\mathcal{H}_{A})$, we follow the
convention \cite{HS12}\ that%
\begin{equation}
\operatorname{Tr}\{G\rho\}\equiv\sup_{n}\operatorname{Tr}\{\Pi_{n}G\Pi_{n}%
\rho\},
\end{equation}
where $\Pi_{n}$ denotes the spectral projection of $G$ corresponding to the
interval $[0,n]$.

\begin{definition}
The $n$th extension $\overline{G}_{n}$\ of an energy observable $G$ is defined
as%
\begin{equation}
\overline{G}_{n}\equiv\frac{1}{n}\left[G\otimes I\otimes\cdots\otimes
I+\cdots+I\otimes\cdots\otimes I\otimes G\right],
\end{equation}
where 
$n$ is the number of factors in each tensor product above.

\end{definition}

In the subsections that follow, let $\mathcal{N}:\mathcal{T}(\mathcal{H}%
_{A})\rightarrow\mathcal{T}(\mathcal{H}_{B})$ denote a quantum channel, and
let $G$ be an energy observable. Let $n\in\mathbb{N}$ denote the number of
channel uses, $M\in\mathbb{N}$ the size of a code, $P\in\lbrack0,\infty)$ an
energy parameter, and $\varepsilon\in\left[  0,1\right]  $ an error parameter.
In what follows, we discuss four different notions of capacity:\ quantum
communication with a uniform energy constraint, entanglement transmission with
an average energy constraint, private communication with a uniform energy
constraint, and secret key transmission with an average energy constraint.
Note that it is possible to consider other combinations, such as quantum communication with an average energy constraint, or secret key transmission with a uniform energy constraint, but we have decided to focus on the above four scenarios for simplicity.

\subsection{Quantum communication with a uniform energy constraint}

An $(n,M,G,P,\varepsilon)$ code for quantum communication with uniform energy
constraint consists of an encoding channel $\mathcal{E}^{n}:\mathcal{T}%
(\mathcal{H}_{S})\rightarrow\mathcal{T}(\mathcal{H}_{A}^{\otimes n})$ and a
decoding channel $\mathcal{D}^{n}:\mathcal{T}(\mathcal{H}_{B}^{\otimes
n})\rightarrow\mathcal{T}(\mathcal{H}_{S})$, where $M=\dim(\mathcal{H}_{S})$.
The energy constraint is uniform, in the sense that the following bound is
required to hold for all states resulting from the output of the encoding
channel $\mathcal{E}^{n}$:
\begin{equation}
\operatorname{Tr}\left\{  \overline{G}_{n}\mathcal{E}^{n}(\rho_{S})\right\}
\leq P, \label{eq:q-code-unif-energy-const}%
\end{equation}
where $\rho_{S}\in\mathcal{D}(\mathcal{H}_{S})$. Note that%
\begin{equation}
\operatorname{Tr}\left\{  \overline{G}_{n}\mathcal{E}^{n}(\rho_{S})\right\}
=\operatorname{Tr}\left\{  G\overline{\rho}_{n}\right\}  ,
\end{equation}
where%
\begin{equation}
\overline{\rho}_{n}\equiv\frac{1}{n}\sum_{i=1}^{n}\operatorname{Tr}%
_{A^{n}\backslash A_{i}}\{\mathcal{E}^{n}(\rho_{S})\}.
\end{equation}
due to the i.i.d.~nature of the observable $\overline{G}_{n}$. Furthermore,
the encoding and decoding channels are good for quantum communication, in the
sense that for all pure states $\phi_{RS}\in\mathcal{D}(\mathcal{H}_{R}%
\otimes\mathcal{H}_{S})$, where $\mathcal{H}_{R}$ is isomorphic
to$~\mathcal{H}_{S}$, the following entanglement fidelity criterion holds%
\begin{equation}
F(\phi_{RS},(\operatorname{id}_{R}\otimes\lbrack\mathcal{D}^{n}\circ
\mathcal{N}^{\otimes n}\circ\mathcal{E}^{n}])(\phi_{RS}))\geq1-\varepsilon.
\label{eq:q-code-fidelity}%
\end{equation}

A rate $R$ is achievable for quantum communication over $\mathcal{N}$ subject
to the uniform energy constraint $P$\ if for all $\varepsilon\in(0,1)$,
$\delta>0$, and sufficiently large $n$, there exists an $(n,2^{n[R-\delta
]},G,P,\varepsilon)$ quantum communication code with uniform energy
constraint. The quantum capacity $Q(\mathcal{N},G,P)$\ of $\mathcal{N}$ with
uniform energy constraint is equal to the supremum of all achievable rates.

\subsection{Entanglement transmission with an average energy constraint}

\label{sec:ent-trans-avg-code}An $(n,M,G,P,\varepsilon)$ code for entanglement
transmission with average energy constraint is defined very similarly as
above, except that the requirements are less stringent. The energy constraint
holds on average, in the sense that it need only hold for the maximally mixed
state $\pi_{S}$ input to the encoding channel $\mathcal{E}^{n}$:%
\begin{equation}
\operatorname{Tr}\left\{  \overline{G}_{n}\mathcal{E}^{n}(\pi_{S})\right\}
\leq P. \label{eq:EG-avg-energy-constraint}%
\end{equation}
Furthermore, we only demand that the particular maximally entangled state
$\Phi_{RS}\in\mathcal{D}(\mathcal{H}_{R}\otimes\mathcal{H}_{S})$, defined as%
\begin{equation}
\Phi_{RS}\equiv\frac{1}{M}\sum_{m,m^{\prime}=1}^{M}|m\rangle\langle m^{\prime
}|_{R}\otimes|m\rangle\langle m^{\prime}|_{S}, \label{eq:MES-EG-AVG}%
\end{equation}
is preserved with good fidelity:%
\begin{equation}
F(\Phi_{RS},(\operatorname{id}_{R}\otimes\lbrack\mathcal{D}^{n}\circ
\mathcal{N}^{\otimes n}\circ\mathcal{E}^{n}])(\Phi_{RS}))\geq1-\varepsilon.
\label{eq:EG-good-reliable-code}%
\end{equation}

A rate $R$ is achievable for entanglement transmission over $\mathcal{N}$
subject to the average energy constraint $P$\ if for all $\varepsilon\in
(0,1)$, $\delta>0$, and sufficiently large $n$, there exists an
$(n,2^{n[R-\delta]},G,P,\varepsilon)$ entanglement transmission code with
average energy constraint. The entanglement transmission capacity
$E(\mathcal{N},G,P)$\ of $\mathcal{N}$ with average energy constraint is equal
to the supremum of all achievable rates.

From definitions, it immediately follows that quantum capacity with uniform
energy constraint can never exceed entanglement transmission capacity with
average energy constraint:%
\begin{equation}
Q(\mathcal{N},G,P)\leq E(\mathcal{N},G,P). \label{eq:Q-less-than-E}%
\end{equation}
In Section~\ref{sec:cap-imps}, we establish the opposite inequality.

\subsection{Private communication with a uniform energy constraint}

An $(n,M,G,P,\varepsilon)$ code for private communication consists of a set
$\{\rho_{A^{n}}^{m}\}_{m=1}^{M}$\ of quantum states, each in $\mathcal{D}%
(\mathcal{H}_{A}^{\otimes n})$, and a POVM\ $\{\Lambda_{B^{n}}^{m}\}_{m=1}%
^{M}$ such that%
\begin{align}
\operatorname{Tr}\left\{  \overline{G}_{n}\rho_{A^{n}}^{m}\right\}   &  \leq
P,\label{eq:energy-constraint}\\
\operatorname{Tr}\{\Lambda_{B^{n}}^{m}\mathcal{N}^{\otimes n}(\rho_{A^{n}}%
^{m})\}  &  \geq1-\varepsilon,\label{eq:private-good-comm}\\
\frac{1}{2}\left\Vert \mathcal{\hat{N}}^{\otimes n}(\rho_{A^{n}}^{m}%
)-\omega_{E^{n}}\right\Vert _{1}  &  \leq\varepsilon, \label{eq:security-cond}%
\end{align}
for all $m\in\left\{  1,\ldots,M\right\}  $, with $\omega_{E^{n}}$ some fixed
state in $\mathcal{D}(\mathcal{H}_{E}^{\otimes n})$. In the above,
$\mathcal{\hat{N}}$ is a channel complementary to $\mathcal{N}$. Observe that%
\begin{equation}
\operatorname{Tr}\left\{  \overline{G}_{n}\rho_{A^{n}}^{m}\right\}
=\operatorname{Tr}\left\{  G\overline{\rho}_{A}^{m}\right\}  ,
\end{equation}
where%
\begin{equation}
\overline{\rho}_{A}^{m}\equiv\frac{1}{n}\sum_{i=1}^{n}\operatorname{Tr}%
_{A^{n}\backslash A_{i}}\{\rho_{A^{n}}^{m}\}. \label{eq:avg-state-energy}%
\end{equation}

A rate $R$ is achievable for private communication over $\mathcal{N}$ subject
to uniform energy constraint $P$\ if for all $\varepsilon\in(0,1)$, $\delta
>0$, and sufficiently large $n$, there exists an $(n,2^{n[R-\delta
]},G,P,\varepsilon)$ private communication code. The private capacity
$P(\mathcal{N},G,P)$\ of $\mathcal{N}$ with uniform energy constraint is equal
to the supremum of all achievable rates.

\subsection{Secret key transmission with an average energy constraint}

\label{sec:SKT-AVG-code}An $(n,M,G,P,\varepsilon)$ code for secret key
transmission with average energy constraint is defined very similarly as
above, except that the requirements are less stringent. The energy constraint
holds on average, in the sense that it need only hold for the average input
state:%
\begin{equation}
\frac{1}{M}\sum_{m=1}^{M}\operatorname{Tr}\left\{  \overline{G}_{n}\rho
_{A^{n}}^{m}\right\}  \leq P. \label{eq:SKT-energy-constraint}%
\end{equation}
Furthermore, we only demand that the conditions in
\eqref{eq:private-good-comm}--\eqref{eq:security-cond}\ hold on average:%
\begin{align}
\frac{1}{M}\sum_{m=1}^{M}\operatorname{Tr}\{\Lambda_{B^{n}}^{m}\mathcal{N}%
^{\otimes n}(\rho_{A^{n}}^{m})\}  &  \geq1-\varepsilon
,\label{eq:private-good-comm-SKT}\\
\frac{1}{M}\sum_{m=1}^{M}\frac{1}{2}\left\Vert \mathcal{\hat{N}}^{\otimes
n}(\rho_{A^{n}}^{m})-\omega_{E^{n}}\right\Vert _{1}  &  \leq\varepsilon,
\label{eq:security-cond-SKT}%
\end{align}
with $\omega_{E^{n}}$ some fixed state in $\mathcal{D}(\mathcal{H}%
_{E}^{\otimes n})$.

A rate $R$ is achievable for secret key transmission over $\mathcal{N}$
subject to the average energy constraint $P$\ if for all $\varepsilon\in
(0,1)$, $\delta>0$, and sufficiently large $n$, there exists an
$(n,2^{n[R-\delta]},G,P,\varepsilon)$ secret key transmission code with
average energy constraint. The secret key transmission capacity $K(\mathcal{N}%
,G,P)$\ of $\mathcal{N}$ with average energy constraint is equal to the
supremum of all achievable rates.

From definitions, it immediately follows that private capacity with uniform
energy constraint can never exceed secret key transmission capacity with
average energy constraint%
\begin{equation}
P(\mathcal{N},G,P)\leq K(\mathcal{N},G,P). \label{eq:P-less-than-K}%
\end{equation}
In Section~\ref{sec:cap-imps}, we establish the opposite inequality.

\section{Code conversions\label{sec:code-conversions}}

In this section, we establish several code conversions, which allow for
converting one type of code into another type of code along with some loss in
the code parameters. In particular, in the forthcoming subsections, we show
how to convert

\begin{enumerate}
\item an entanglement transmission code with an average energy constraint to a
quantum communication code with a uniform energy constraint,

\item a quantum communication code with a uniform energy constraint to a
private communication code with a uniform energy constraint,

\item and a secret key transmission code with an average energy constraint to
a private communication code with a uniform energy constraint.
\end{enumerate}

\noindent These code conversions then allow us to establish several
non-trivial relations between the corresponding capacities, which we do in
Section~\ref{sec:cap-imps}.

\subsection{Entanglement transmission with an average energy constraint to
quantum communication with a uniform energy constraint}

In this subsection, we show how an entanglement transmission code with an
average energy constraint implies the existence of a quantum communication
code with a uniform energy constraint, such that there is a loss in
performance in the resulting code with respect to several code parameters.

A result like this was first established in \cite{BKN98}\ and reviewed in
\cite{KW04,K07,Wat16}, under the assumption that there is no energy
constraint. Here we follow the proof approach available in \cite{K07,Wat16},
but we make several modifications in order to deal with going from an average
energy constraint to a uniform energy constraint.

\begin{proposition}
\label{thm:EG-2-QC}For all $\delta\in(1/M,1/2)$, the existence of an
$(n,M,G,P,\varepsilon)$ entanglement transmission code with average energy
constraint implies the existence of an $(n,\left\lfloor \delta M\right\rfloor
,G,P/\left(  1-2\delta\right)  ,\min\{1,2\sqrt{\varepsilon/[\delta-1/M]}\})$
quantum communication code with uniform energy constraint.
\end{proposition}

\begin{proof}
Suppose that an $(n,M,G,P,\varepsilon)$ entanglement transmission code with
average energy constraint exists. This implies that the conditions in
\eqref{eq:EG-avg-energy-constraint}\ and
\eqref{eq:EG-good-reliable-code}\ hold. Let $\mathcal{C}^{n}:\mathcal{T}%
(\mathcal{H}_{S})\rightarrow\mathcal{T}(\mathcal{H}_{S})$ denote the
finite-dimensional channel consisting of the encoding, communication channel,
and decoding:%
\begin{equation}
\mathcal{C}^{n}\equiv\mathcal{D}^{n}\circ\mathcal{N}^{\otimes n}%
\circ\mathcal{E}^{n}.
\end{equation}
We proceed with the following algorithm:

\begin{enumerate}
\item Set $k=M$, $\mathcal{H}_{M}=\mathcal{H}_{S}$, and $\delta\in\left(
1/M,1/2\right)  $. Suppose for now that $\delta M$ is a positive integer.

\item Set $|\phi_{k}\rangle\in\mathcal{H}_{k}$ to be a state vector such that
the input-output fidelity is minimized:%
\begin{equation}
|\phi_{k}\rangle\equiv\arg\min_{|\phi\rangle\in\mathcal{H}_{k}}\langle
\phi|\mathcal{C}^{n}(|\phi\rangle\langle\phi|)|\phi\rangle,
\end{equation}
and set the fidelity $F_{k}$ and energy $E_{k}$ of $|\phi_{k}\rangle$ as
follows:%
\begin{align}
F_{k}  &  \equiv\min_{|\phi\rangle\in\mathcal{H}_{k}}\langle\phi
|\mathcal{C}^{n}(|\phi\rangle\langle\phi|)|\phi\rangle\\
&  =\langle\phi_{k}|\mathcal{C}^{n}(|\phi_{k}\rangle\langle\phi_{k}|)|\phi
_{k}\rangle,\\
E_{k}  &  \equiv\operatorname{Tr}\{\overline{G}_{n}\mathcal{E}^{n}(|\phi
_{k}\rangle\langle\phi_{k}|)\}.
\end{align}

\item Set%
\begin{equation}
\mathcal{H}_{k-1}\equiv\operatorname{span}\{|\psi\rangle\in\mathcal{H}%
_{k}:\left\vert \left\langle \psi|\phi_{k}\right\rangle \right\vert =0\}.
\end{equation}
That is, $\mathcal{H}_{k-1}$ is set to the orthogonal complement of $|\phi
_{k}\rangle$ in $\mathcal{H}_{k}$, so that $\mathcal{H}_{k}=\mathcal{H}%
_{k-1}\oplus\operatorname{span}\{|\phi_{k}\rangle\}$. Set $k:=k-1$.

\item Repeat steps 2-3 until $k=\left(  1-\delta\right)  M$ after step 3.

\item Let $|\phi_{k}\rangle\in\mathcal{H}_{k}$ be a state vector such that the
input energy is maximized:%
\begin{equation}
|\phi_{k}\rangle\equiv\arg\max_{|\phi\rangle\in\mathcal{H}_{k}}%
\operatorname{Tr}\{\overline{G}_{n}\mathcal{E}^{n}(|\phi\rangle\langle
\phi|)\},
\end{equation}
and set the fidelity $F_{k}$ and energy $E_{k}$ of $|\phi_{k}\rangle$ as
follows:%
\begin{align}
F_{k}  &  \equiv\langle\phi_{k}|\mathcal{C}^{n}(|\phi_{k}\rangle\langle
\phi_{k}|)|\phi_{k}\rangle\\
E_{k}  &  \equiv\max_{|\phi\rangle\in\mathcal{H}_{k}}\operatorname{Tr}%
\{\overline{G}_{n}\mathcal{E}^{n}(|\phi\rangle\langle\phi|)\}\\
&  =\operatorname{Tr}\{\overline{G}_{n}\mathcal{E}^{n}(|\phi_{k}\rangle
\langle\phi_{k}|)\}.
\end{align}

\item Set%
\begin{equation}
\mathcal{H}_{k-1}\equiv\operatorname{span}\{|\psi\rangle\in\mathcal{H}%
_{k}:\left\vert \left\langle \psi|\phi_{k}\right\rangle \right\vert =0\}.
\end{equation}
Set $k:=k-1$.

\item Repeat steps 5-6 until $k=0$ after step 6.
\end{enumerate}

The idea behind this algorithm is to successively remove minimum fidelity
states from $\mathcal{H}_{S}$ until $k=\left(  1-\delta\right)  M$. By the
structure of the algorithm and some analysis given below, we are then
guaranteed for this $k$ and lower that%
\begin{equation}
1-\min_{|\phi\rangle\in\mathcal{H}_{k}}\langle\phi|\mathcal{C}^{n}%
(|\phi\rangle\langle\phi|)|\phi\rangle\leq\varepsilon/\delta.
\end{equation}
That is, the subspace $\mathcal{H}_{k}$ is good for quantum communication with
fidelity at least $1-\varepsilon/\delta$. After this $k$, we then successively
remove maximum energy states from $\mathcal{H}_{k}$ until the algorithm
terminates. Furthermore, the algorithm implies that%
\begin{align}
F_{M}  &  \leq F_{M-1}\leq\cdots\leq F_{\left(  1-\delta\right)
M+1},\label{eq:fidelity-ordering}\\
E_{\left(  1-\delta\right)  M}  &  \geq E_{\left(  1-\delta\right)  M-1}%
\geq\cdots\geq E_{1},\label{eq:energy-ordering}\\
\mathcal{H}_{M}  &  \supseteq\mathcal{H}_{M-1}\supseteq\cdots\supseteq
\mathcal{H}_{1}.
\end{align}
Also, $\{|\phi_{k}\rangle\}_{k=1}^{l}$ is an orthonormal basis for
$\mathcal{H}_{l}$, where $l\in\{1,\ldots,M\}$.

We now analyze the result of this algorithm by employing Markov's inequality
and some other tools. From the condition in \eqref{eq:EG-good-reliable-code}
that the original code is good for entanglement transmission, we have that%
\begin{equation}
F(\Phi_{RS},(\operatorname{id}_{R}\otimes\mathcal{C}^{n})(\Phi_{RS}%
))\geq1-\varepsilon.
\end{equation}
Since $\{|\phi_{k}\rangle\}_{k=1}^{M}$ is an orthonormal basis for
$\mathcal{H}_{M}$, we can write%
\begin{equation}
|\Phi\rangle_{RS}=\frac{1}{\sqrt{M}}\sum_{k=1}^{M}|\phi_{k}^{\ast}\rangle
_{R}\otimes|\phi_{k}\rangle_{S},
\end{equation}
where $\ast$ denotes complex conjugate with respect to the basis in
\eqref{eq:MES-EG-AVG}, and the reduced state can be written as $\Phi_{S}%
=\frac{1}{M}\sum_{k=1}^{M}|\phi_{k}\rangle\langle\phi_{k}|_{S}$. A consequence
of \cite[Exercise~9.5.1]{W16} is that%
\begin{align}
F(\Phi_{RS},(\operatorname{id}_{R}\otimes\mathcal{C}^{n})(\Phi_{RS}))  &
\leq\frac{1}{M}\sum_{k}\langle\phi_{k}|\mathcal{C}^{n}(|\phi_{k}\rangle
\langle\phi_{k}|)|\phi_{k}\rangle\nonumber\\
&  =\frac{1}{M}\sum_{k}F_{k}.
\end{align}
So this means that%
\begin{equation}
\frac{1}{M}\sum_{k}F_{k}\geq1-\varepsilon\quad\Leftrightarrow\quad\frac{1}%
{M}\sum_{k}\left(  1-F_{k}\right)  \leq\varepsilon.
\end{equation}
Now taking $K$ as a uniform random variable with realizations $k\in\left\{
1,\ldots,M\right\}  $ and applying Markov's inequality, we find that%
\begin{equation}
\Pr_{K}\{1-F_{K}\geq\varepsilon/\delta\}\leq\frac{\mathbb{E}_{K}\{1-F_{K}%
\}}{\varepsilon/\delta}\leq\frac{\varepsilon}{\varepsilon/\delta}=\delta.
\end{equation}
So this implies that $\left(  1-\delta\right)  M$ of the $F_{k}$ values are
such that $F_{k}\geq1-\varepsilon/\delta$. Since they are ordered as given in
\eqref{eq:fidelity-ordering}, we can conclude that $\mathcal{H}_{\left(
1-\delta\right)  M}$ is a subspace good for quantum communication in the
following sense:%
\begin{equation}
\min_{|\phi\rangle\in\mathcal{H}_{\left(  1-\delta\right)  M}}\langle
\phi|\mathcal{C}^{n}(|\phi\rangle\langle\phi|)|\phi\rangle\geq1-\varepsilon
/\delta.
\end{equation}

Now consider from the average energy constraint in
\eqref{eq:EG-avg-energy-constraint} that%
\begin{align}
P  &  \geq\operatorname{Tr}\left\{  \overline{G}_{n}\mathcal{E}^{n}(\pi
_{S})\right\} \\
&  =\frac{1}{M}\sum_{k=1}^{M}\operatorname{Tr}\left\{  \overline{G}%
_{n}\mathcal{E}^{n}(|\phi_{k}\rangle\langle\phi_{k}|_{S})\right\} \\
&  =\frac{1}{M}\sum_{k=1}^{M}E_{k}\\
&  \geq\frac{1-\delta}{\left(  1-\delta\right)  M}\sum_{k=1}^{\left(
1-\delta\right)  M}E_{k},
\end{align}
which we can rewrite as%
\begin{equation}
\frac{1}{\left(  1-\delta\right)  M}\sum_{k=1}^{\left(  1-\delta\right)
M}E_{k}\leq P/\left(  1-\delta\right)  .
\end{equation}
Taking $K^{\prime}$ as a uniform random variable with realizations
$k\in\left\{  1,\ldots,\left(  1-\delta\right)  M\right\}  $ and applying
Markov's inequality, we find that%
\begin{align}
\Pr_{K^{\prime}}\left\{  E_{K^{\prime}}\geq P/\left(  1-2\delta\right)
\right\}   &  \leq\frac{P/\left(  1-\delta\right)  }{P/\left(  1-2\delta
\right)  }\\
&  =\frac{1-2\delta}{1-\delta}.
\end{align}
Rewriting this, we find that%
\begin{align}
\Pr_{K^{\prime}}\left\{  E_{K^{\prime}}\leq P/\left(  1-2\delta\right)
\right\}   &  \geq1-\frac{1-2\delta}{1-\delta}\\
&  =\frac{\delta}{1-\delta}.
\end{align}
Thus, a fraction $\delta/\left(  1-\delta\right)  $ of the remaining $\left(
1-\delta\right)  M$ state vectors $|\phi_{k}\rangle$ are such that $E_{k}\leq
P/\left(  1-2\delta\right)  $. Since they are ordered as in
\eqref{eq:energy-ordering}, this means that $\left\{  |\phi_{\delta M}%
\rangle,\ldots,|\phi_{1}\rangle\right\}  $ have this property.

We can then conclude that the subspace $\mathcal{H}_{\delta M}$ is such that%
\begin{align}
\dim(\mathcal{H}_{\delta M})  &  =\delta M,\label{eq:resulting-code-size-1}\\
\min_{|\phi\rangle\in\mathcal{H}_{\delta M}}\langle\phi|\mathcal{C}^{n}%
(|\phi\rangle\langle\phi|)|\phi\rangle &  \geq1-\varepsilon/\delta
,\label{eq:min-fid-condition-almost-done}\\
\max_{|\phi\rangle\in\mathcal{H}_{\delta M}}\operatorname{Tr}\{\overline
{G}_{n}\mathcal{E}^{n}(|\phi\rangle\langle\phi|)\}  &  \leq P/\left(
1-2\delta\right)  . \label{eq:resulting-code-power-constr}%
\end{align}

Now applying Proposition~\ref{prop:min-fid-to-min-ent-fid}\ (in the appendix)
to \eqref{eq:min-fid-condition-almost-done}, we can conclude that the minimum
entanglement fidelity obeys the following bound:%
\begin{equation}
\min_{|\psi\rangle\in\mathcal{H}_{\delta M}^{\prime}\otimes\mathcal{H}_{\delta
M}}\langle\psi|(\operatorname{id}_{\mathcal{H}_{\delta M}^{\prime}}%
\otimes\mathcal{C}^{n})(|\psi\rangle\langle\psi|)|\psi\rangle\geq
1-2\sqrt{\varepsilon/\delta}. \label{eq:resulting-code-fid-1}%
\end{equation}

To finish off the proof, suppose that $\delta M$ is not an integer. Then there
exists a $\delta^{\prime}<\delta$ such that $\delta^{\prime}M=\left\lfloor
\delta M\right\rfloor $ is a positive integer. By the above reasoning, there
exists a code with parameters as given in
\eqref{eq:resulting-code-size-1}--\eqref{eq:resulting-code-fid-1}, except with
$\delta$ replaced by $\delta^{\prime}$. Then the code dimension is equal to
$\left\lfloor \delta M\right\rfloor $. Using that $\delta^{\prime
}M=\left\lfloor \delta M\right\rfloor >\delta M-1$, we find that
$\delta^{\prime}>\delta-1/M$, which implies that $1-2\sqrt{\varepsilon
/\delta^{\prime}}>1-2\sqrt{\varepsilon/[\delta-1/M]}$. We also have that
$P/\left(  1-2\delta^{\prime}\right)  <P/\left(  1-2\delta\right)  $. This
concludes the proof.
\end{proof}

\subsection{Quantum communication with a uniform energy constraint implies
private communication with a uniform energy constraint}

This subsection establishes that a quantum communication code with uniform
energy constraint can always be converted to one for private communication
with uniform energy constraint, such that there is negligible loss with
respect to code parameters.

\begin{proposition}
\label{thm:QC-to-PC}The existence of an $(n,M,G,P,\varepsilon)$ quantum
communication code with uniform energy constraint implies the existence of an
$(n,\left\lfloor M/2\right\rfloor ,G,P,\min\{1,2\sqrt{\varepsilon}\})$ code for
private communication with uniform energy constraint.
\end{proposition}

\begin{proof}
Starting from an $(n,M,G,P,\varepsilon)$ quantum communication code with
uniform energy constraint, we can use it to transmit a maximally entangled
state%
\begin{equation}
\Phi_{RS}\equiv\frac{1}{M}\sum_{m,m^{\prime}=1}^{M}|m\rangle\langle m^{\prime
}|_{R}\otimes|m\rangle\langle m^{\prime}|_{S}%
\end{equation}
of Schmidt rank $M$ faithfully, by applying \eqref{eq:q-code-fidelity}:%
\begin{equation}
F(\Phi_{RS},(\operatorname{id}_{R}\otimes\mathcal{D}^{n}\circ\mathcal{N}%
^{\otimes n}\circ\mathcal{E}^{n})(\Phi_{RS}))\geq1-\varepsilon.
\label{eq:fid-crit-q-to-priv}%
\end{equation}
Consider that the state%
\begin{equation}
\sigma_{RSE^{n}}\equiv(\operatorname{id}_{R}\otimes\mathcal{D}^{n}\circ
\lbrack\mathcal{U}^{\mathcal{N}}]^{\otimes n}\circ\mathcal{E}^{n})(\Phi_{RS})
\end{equation}
extends the state output from the actual protocol. By Uhlmann's theorem (see
\eqref{eq:Uhlmann-4-exts}), there exists an extension of $\Phi_{RS}$ such that
the fidelity between this extension and the state $\sigma_{RSE^{n}}$ is equal
to the fidelity in \eqref{eq:fid-crit-q-to-priv}. However, the maximally
entangled state $\Phi_{RS}$ is \textquotedblleft
unextendible\textquotedblright\ in the sense that the only possible extension
is a tensor-product state $\Phi_{RS}\otimes\omega_{E^{n}}$ for some state
$\omega_{E^{n}}$. So, putting these statements together, we find that%
\begin{equation}
F(\Phi_{RS}\otimes\omega_{E^{n}},(\operatorname{id}_{R}\otimes\mathcal{D}%
^{n}\circ\lbrack\mathcal{U}^{\mathcal{N}}]^{\otimes n}\circ\mathcal{E}%
^{n})(\Phi_{RS}))\geq1-\varepsilon.
\end{equation}
Furthermore, measuring the $R$ and $S$ systems locally in the Schmidt basis of
$\Phi_{RS}$ only increases the fidelity, so that%
\begin{equation}
F(\overline{\Phi}_{RS}\otimes\omega_{E^{n}},(\operatorname{id}_{R}%
\otimes\overline{\mathcal{D}}^{n}\circ\lbrack\mathcal{U}^{\mathcal{N}%
}]^{\otimes n}\circ\mathcal{E}^{n})(\overline{\Phi}_{RS}))\geq1-\varepsilon,
\end{equation}
where $\overline{\mathcal{D}}^{n}$ denotes the concatenation of the original
decoder $\mathcal{D}^{n}$ followed by the local measurement:%
\begin{align}
\overline{\mathcal{D}}^{n}(\cdot)  &  \equiv\sum_{m}|m\rangle\langle
m|\mathcal{D}^{n}(\cdot)|m\rangle\langle m|\\
&  =\sum_{m}\operatorname{Tr}\{\mathcal{D}^{n\dag}[|m\rangle\langle
m|](\cdot)\}|m\rangle\langle m|.
\end{align}
Observe that $\{\mathcal{D}^{n\dag}[|m\rangle\langle m|]\}_{m}$ is a valid
POVM. Employing the inequalities in \eqref{eq:F-vd-G}, we can conclude that%
\begin{equation}
\frac{1}{2}\left\Vert \overline{\Phi}_{RS}\otimes\omega_{E^{n}}%
-(\operatorname{id}_{R}\otimes\overline{\mathcal{D}}^{n}\circ\lbrack
\mathcal{U}^{\mathcal{N}}]^{\otimes n}\circ\mathcal{E}^{n})(\overline{\Phi
}_{RS})\right\Vert _{1}\leq\sqrt{\varepsilon}.
\end{equation}
Using the direct sum property of the trace distance from
\eqref{eq:direct-sum-TD} and defining $\rho_{A^{n}}^{m}\equiv\mathcal{E}%
^{n}(|m\rangle\langle m|_{S})$, we can then rewrite this as%
\begin{equation}
\frac{1}{2M}\sum_{m=1}^{M}\left\Vert |m\rangle\langle m|_{S}\otimes
\omega_{E^{n}}-(\overline{\mathcal{D}}^{n}\circ\lbrack\mathcal{U}%
^{\mathcal{N}}]^{\otimes n})(\rho_{A^{n}}^{m})\right\Vert _{1}\leq
\sqrt{\varepsilon}.
\end{equation}
Markov's inequality then guarantees that there exists a subset $\mathcal{M}%
^{\prime}$\ of $\left[  M\right]  $ of size $\left\lfloor M/2\right\rfloor $
such that the following condition holds for all $m\in\mathcal{M}^{\prime}$:%
\begin{equation}
\frac{1}{2}\left\Vert |m\rangle\langle m|_{S}\otimes\omega_{E^{n}}%
-(\overline{\mathcal{D}}^{n}\circ\lbrack\mathcal{U}^{\mathcal{N}}]^{\otimes
n})(\rho_{A^{n}}^{m})\right\Vert _{1}\leq2\sqrt{\varepsilon}.
\label{eq:q-to-p-good-condition}%
\end{equation}
We now define the private communication code to consist of codewords
$\{\rho_{A^{n}}^{m}\equiv\mathcal{E}^{n}(|m\rangle\langle m|_{S}%
)\}_{m\in\mathcal{M}^{\prime}}$ and the decoding POVM\ to be%
\begin{multline}
\{\Lambda_{B^{n}}^{m}\equiv\mathcal{D}^{n\dag}(|m\rangle\langle m|)\}_{m\in
\mathcal{M}^{\prime}}\\
\cup\left\{  \Lambda_{B^{n}}^{0}\equiv\mathcal{D}^{n\dag}\!\left(
\sum_{m\not \in \mathcal{M}^{\prime}}|m\rangle\langle m|\right)  \right\}  .
\end{multline}
Note that the energy constraint holds for all codewords%
\begin{equation}
\operatorname{Tr}\{\overline{G}_{n}\rho_{A^{n}}^{m}\}\leq P,
\end{equation}
due to the assumption that we start from a quantum communication code with
uniform energy constraint as given in \eqref{eq:q-code-unif-energy-const}.
Applying monotonicity of partial trace to \eqref{eq:q-to-p-good-condition}
with respect to system~$S$, we find that the following condition holds for all
$m\in\mathcal{M}^{\prime}$:%
\begin{equation}
\frac{1}{2}\left\Vert \omega_{E^{n}}-\mathcal{\hat{N}}^{\otimes n}(\rho
_{A^{n}}^{m})\right\Vert _{1}\leq2\sqrt{\varepsilon},
\end{equation}
which gives the desired security condition in \eqref{eq:security-cond}.
Applying monotonicity of partial trace to
\eqref{eq:q-to-p-good-condition}\ with respect to system~$E^{n}$ gives that%
\begin{equation}
\frac{1}{2}\left\Vert |m\rangle\langle m|_{S}-(\overline{\mathcal{D}}^{n}%
\circ\mathcal{N}^{\otimes n})(\rho_{A^{n}}^{m})\right\Vert _{1}\leq
2\sqrt{\varepsilon}, \label{eq:q-to-p-good-decoding-1}%
\end{equation}
for all $m\in\mathcal{M}^{\prime}$. Abbreviating $\Gamma^{m^{\prime}}_{B^{n}}
\equiv\mathcal{D}^{n\dag}(|m^{\prime}\rangle\langle m^{\prime}|)$, consider
then that for all $m \in\mathcal{M}^{\prime}$
\begin{align}
&  \frac{1}{2}\left\Vert |m\rangle\langle m|_{S}-(\overline{\mathcal{D}}%
^{n}\circ\mathcal{N}^{\otimes n})(\rho_{A^{n}}^{m})\right\Vert _{1}\nonumber\\
&  =\frac{1}{2}\left\Vert |m\rangle\langle m|_{S}-\sum_{m^{\prime}=1}^{M
}\operatorname{Tr}\{\Gamma^{m^{\prime}}_{B^{n}}\mathcal{N}^{\otimes n}%
(\rho_{A^{n}}^{m})\}|m^{\prime}\rangle\langle m^{\prime}|\right\Vert
_{1}\nonumber\\
&  =\frac{1}{2}\left\Vert p_{e}|m\rangle\langle m|_{S}-\sum_{m^{\prime}\neq
m}\operatorname{Tr}\{\Gamma^{m^{\prime}}_{B^{n}}\mathcal{N}^{\otimes n}%
(\rho_{A^{n}}^{m})\}|m^{\prime}\rangle\langle m^{\prime}|\right\Vert
_{1}\nonumber\\
&  =\frac{1}{2}\left(  p_{e}+\sum_{m^{\prime}\neq m}\operatorname{Tr}%
\{\Gamma^{m^{\prime}}_{B^{n}}\mathcal{N}^{\otimes n}(\rho_{A^{n}}%
^{m})\}\right) \nonumber\\
&  =1-\operatorname{Tr}\{\Lambda_{B^{n}}^{m}\mathcal{N}^{\otimes n}%
(\rho_{A^{n}}^{m})\},
\end{align}
where $p_{e}\equiv1-\operatorname{Tr}\{\Lambda_{B^{n}}^{m}\mathcal{N}^{\otimes
n}(\rho_{A^{n}}^{m})\}$. Combining this equality with
\eqref{eq:q-to-p-good-decoding-1} gives the desired reliable decoding
condition in \eqref{eq:private-good-comm} for all $m \in\mathcal{M}^{\prime}$
\begin{equation}
\operatorname{Tr}\{\Lambda_{B^{n}}^{m}\mathcal{N}^{\otimes n}(\rho_{A^{n}}%
^{m})\}\geq1-2\sqrt{\varepsilon}.
\end{equation}
Thus, we have shown that from an $(n,M,G,P,\varepsilon)$ quantum communication
code with uniform energy constraint, one can realize an $(n,\left\lfloor
M/2\right\rfloor ,G,P,2\sqrt{\varepsilon})$ code for private communication
with uniform energy constraint.
\end{proof}

\begin{remark}
That a quantum communication code can be easily converted to a private
communication code is part of the folklore of quantum information theory.
Ref.~\cite{ieee2005dev}\ proved that the unconstrained quantum capacity never
exceeds the unconstrained private capacity, but we are not aware of an
explicit code conversion statement of the form given in
Proposition~\ref{thm:QC-to-PC}.
\end{remark}

\subsection{Secret key transmission with an average energy constraint implies
private communication with a uniform energy constraint}

We finally establish that a secret key transmission code with average energy
constraint can be converted to a private communication code with uniform
energy constraint.

\begin{proposition}
\label{thm:SK-to-PC}For $\delta\in(1/M,1/3)$, the existence of an
$(n,M,G,P,\varepsilon)$ secret key transmission code with average energy
constraint implies the existence of an $(n,\left\lfloor \delta M\right\rfloor
,G,P/(1-3\delta),\min\{1,\varepsilon/[\delta-1/M]\})$ private communication code
with uniform energy constraint.
\end{proposition}

\begin{proof}
To begin with, suppose that $\delta M$ is an integer. The existence of an
$(n,M,G,P,\varepsilon)$ secret key transmission code with average energy
constraint implies that the following three conditions hold:
\begin{align}
\frac{1}{M}\sum_{m=1}^{M}E_{m}  &  \leq P~,\quad\frac{1}{M}\sum_{m=1}^{M}%
T_{m}\geq1-\varepsilon~,\label{eq:average-E-SK-code}\\
\frac{1}{M}\sum_{m=1}^{M}D_{m}  &  \leq\varepsilon~,
\end{align}
where%
\begin{align}
E_{m}  &  \equiv\operatorname{Tr}\{\overline{G}_{n}\rho_{A^{n}}^{m}\}~,\\
T_{m}  &  \equiv\operatorname{Tr}\{\Lambda_{B^{n}}^{m}\mathcal{N}^{\otimes
n}(\rho_{A^{n}}^{m})\}~,\\
D_{m}  &  \equiv\frac{1}{2}\left\Vert \mathcal{\hat{N}}^{\otimes n}%
(\rho_{A^{n}}^{m})-\omega_{E^{n}}\right\Vert _{1}~.
\end{align}
Now taking $\hat{M}$ as a uniform random variable with realizations
$m\in\left\{  {1,\ldots,M}\right\}  $ and applying Markov's inequality, we
have for $\delta\in(0,1/3)$ that%
\begin{equation}
\Pr_{\hat{M}}\{1-T_{\hat{M}}\geq\varepsilon/\delta\}\leq\frac{\mathbb{E}%
_{\hat{M}}\{1-T_{\hat{M}}\}}{\varepsilon/\delta}\leq\frac{\varepsilon
}{\varepsilon/\delta}~.
\end{equation}
This implies that $(1-\delta)M$ of the $T_{m}$ values are such that $T_{m}%
\geq1-\varepsilon/\delta$. We then rearrange the order of $T_{m}$, $D_{m}$,
and $E_{m}$ using a label $m^{\prime}$ such that the first $(1-\delta)M$ of
the $T_{m^{\prime}}$ variables satisfy the condition $T_{m^{\prime}}%
\geq1-\varepsilon/\delta$. Now from \eqref{eq:average-E-SK-code}, we have that%
\begin{equation}
\varepsilon\geq\frac{1}{M}\sum_{m^{\prime}=1}^{M}D_{m^{\prime}}\geq
\frac{1-\delta}{(1-\delta)M}\sum_{m^{\prime}=1}^{(1-\delta)M}D_{m^{\prime}}~,
\end{equation}
which can be rewritten as%
\begin{equation}
\frac{1}{(1-\delta)M}\sum_{m=1}^{(1-\delta)M}D_{m^{\prime}}\leq\frac
{\varepsilon}{1-\delta}.
\end{equation}
Now taking $\hat{M}^{\prime}$ as a uniform random variable with realizations
$m^{\prime}\in\{1,\ldots,(1-\delta)M\}$ and applying Markov's inequality, we
find that%
\begin{align}
\Pr_{\hat{M}^{\prime}}\left\{  D_{\hat{M}^{\prime}}\geq\varepsilon
/\delta\right\}   &  \leq\frac{\mathbb{E}_{\hat{M}^{\prime}}\{D_{\hat
{M}^{\prime}}\}}{\varepsilon/\delta}\\
&  \leq\frac{\varepsilon/(1-\delta)}{\varepsilon/\delta}\\
&  =\frac{\delta}{1-\delta}~.
\end{align}
Thus a fraction $1-\left[  \delta/(1-\delta)\right]  =(1-2\delta)/(1-\delta)$
of the first $(1-\delta)M$ variables $D_{m^{\prime}}$ satisfy $D_{\hat
{M}^{\prime}}\leq\varepsilon/\delta$. Now rearrange the order of
$T_{m^{\prime}}$, $D_{m^{\prime}}$, and $E_{m^{\prime}}$ with label
$m^{\prime\prime}$ such that the first $(1-2\delta)M$ of them satisfy
\begin{align}
T_{m^{\prime\prime}}  &  \geq1-\varepsilon/\delta~,\\
D_{m^{\prime\prime}}  &  \leq\varepsilon/\delta~.
\end{align}
From \eqref{eq:average-E-SK-code}, we get that%
\begin{equation}
P\geq\frac{1}{M}\sum_{m^{\prime\prime}=1}^{M}E_{m^{\prime\prime}}\geq
\frac{1-2\delta}{(1-2\delta)M}\sum_{m^{\prime\prime}=1}^{(1-2\delta
)M}E_{m^{\prime\prime}}~,
\end{equation}
which can be rewritten as%
\begin{equation}
\frac{1}{(1-2\delta)M}\sum_{m^{\prime\prime}=1}^{(1-2\delta)M}E_{m^{\prime
\prime}}\leq\frac{P}{1-2\delta}~.
\end{equation}
Taking $\hat{M}^{\prime\prime}$ as a uniform random variable with realizations
$m^{\prime\prime}\in\{1,...,(1-2\delta)M\}$ and applying Markov's inequality,
we find that%
\begin{align}
\Pr_{\hat{M}^{\prime\prime}}\left\{  E_{\hat{M}^{\prime\prime}}\geq
P/(1-3\delta)\right\}   &  \leq\frac{\mathbb{E}_{\hat{M}^{\prime\prime}%
}\{E_{\hat{M}^{\prime\prime}}\}}{P/(1-\delta)}\\
&  \leq\frac{P/(1-2\delta)}{P/(1-3\delta)}\\
&  =\frac{1-3\delta}{1-2\delta}~.
\end{align}
Thus a fraction $1-(1-3\delta)/(1-2\delta)=\delta/(1-2\delta)$ of the first
$(1-2\delta)M$ variables $E_{m^{\prime\prime}}$ satisfy the condition
$E_{\hat{M}^{\prime\prime}}\leq P/(1-3\delta)$. We can finally relabel
$T_{m^{\prime\prime}}$, $D_{m^{\prime\prime}}$, and $E_{m^{\prime\prime}}$
with a label $m^{\prime\prime\prime}$ such that the first $\delta M$ of them
satisfy
\begin{align}
E_{m^{\prime\prime\prime}}  &  \leq P/(1-3\delta
)~,\label{eq:resulting-code-energy-2}\\
T_{m^{\prime\prime\prime}}  &  \geq1-\varepsilon/\delta~,\\
D_{m^{\prime\prime\prime}}  &  \leq\varepsilon/\delta~.
\label{eq:resulting-code-secrecy-2}%
\end{align}
The corresponding codewords then constitute an $(n,\delta M,G,P/(1-3\delta
),\varepsilon/\delta)$ private communication code with uniform energy constraint.

To finish off the proof, suppose that $\delta M$ is not an integer. Then there
exists a $\delta^{\prime}<\delta$ such that $\delta^{\prime}M=\left\lfloor
\delta M\right\rfloor $ is a positive integer. By the above reasoning, there
exists a code with parameters as given in
\eqref{eq:resulting-code-energy-2}--\eqref{eq:resulting-code-secrecy-2},
except with $\delta$ replaced by $\delta^{\prime}$. Then the code size is
equal to $\left\lfloor \delta M\right\rfloor $. Using that $\delta^{\prime
}M=\left\lfloor \delta M\right\rfloor >\delta M-1$, we find that
$\delta^{\prime}>\delta-1/M$, which implies that $1-\varepsilon/\delta
^{\prime}>1-\varepsilon/[\delta-1/M]$ and $\varepsilon/\delta^{\prime
}<\varepsilon/\left[  \delta-1/M\right]  $. We also have that $P/\left(
1-3\delta^{\prime}\right)  <P/\left(  1-3\delta\right)  $. This concludes the proof.
\end{proof}

\section{Implications of code conversions for capacities}

\label{sec:cap-imps}In this brief section, we show how the various code
conversions from Section~\ref{sec:code-conversions}\ have implications for the
capacities defined in Section~\ref{sec:energy-constrained-caps}. The main
result is the following theorem:

\begin{theorem}
\label{thm:cap-relations}Let $\mathcal{N}:\mathcal{T}(\mathcal{H}%
_{A})\rightarrow\mathcal{T}(\mathcal{H}_{B})$ be a quantum channel,
$G\in\mathcal{P}(\mathcal{H}_{A})$ an energy observable, and $P\in
\lbrack0,\infty)$. Then the following relations hold for the capacities
defined in Section~\ref{sec:energy-constrained-caps}:%
\begin{align}
Q(\mathcal{N},G,P)  &  =E(\mathcal{N},G,P)\nonumber\\
&  \leq P(\mathcal{N},G,P)=K(\mathcal{N},G,P).
\end{align}

\end{theorem}

\begin{proof}
As a consequence of the definitions of these capacities and as remarked in
\eqref{eq:Q-less-than-E} and \eqref{eq:P-less-than-K}, we have that%
\begin{align}
Q(\mathcal{N},G,P) &  \leq E(\mathcal{N},G,P),\\
P(\mathcal{N},G,P) &  \leq K(\mathcal{N},G,P).
\end{align}
So it suffices to prove the following three inequalities:%
\begin{align}
Q(\mathcal{N},G,P) &  \geq E(\mathcal{N},G,P),\label{eq:Q>=E}\\
Q(\mathcal{N},G,P) &  \leq P(\mathcal{N},G,P),\label{eq:Q<=P}\\
P(\mathcal{N},G,P) &  \geq K(\mathcal{N},G,P).\label{eq:K<=P}%
\end{align}
These follow from Propositions~\ref{thm:EG-2-QC}, \ref{thm:QC-to-PC}, and
\ref{thm:SK-to-PC}, respectively. Let us establish \eqref{eq:Q>=E}. Fix
a constant $\delta\in(0,1/2)$. Suppose
that $R$ is an achievable rate for entanglement transmission with an average
energy constraint $P(1-2\delta)$. This implies the existence of a sequence of $(n,M_{n}%
,G,P(1-2\delta),\varepsilon_{n})$\ codes such that%
\begin{align}
\liminf_{n\rightarrow\infty}\frac{1}{n}\log M_{n} &  =R,\\
\lim_{n\rightarrow\infty}\varepsilon_{n} &  =0.
\end{align}
Suppose that the sequence is such that $M_{n}$ is non-decreasing with $n$ (if
it is not the case, then pick out a subsequence for which it is the case).  Now pick $n$ large enough such that $\delta
\geq1/M_{n}$. Invoking Proposition~\ref{thm:EG-2-QC}, there exists an
$(n,\left\lfloor \delta M_{n}\right\rfloor ,G,P,\min
[1,2\sqrt{\varepsilon_{n}/\left[  \delta-1/M_{n}\right]  }])$ quantum
communication code with uniform energy constraint. From the facts that%
\begin{align}
\liminf_{n\rightarrow\infty}\frac{1}{n}\log\left(  \left\lfloor \delta
M_{n}\right\rfloor \right)   &  =\liminf_{n\rightarrow\infty}\frac{1}{n}\log
M_{n}\\
&  =R,\\
\limsup_{n\rightarrow\infty}2\sqrt{\varepsilon_{n}/\left[  \delta
-1/M_{n}\right]  } &  =0,
\end{align}
we can conclude that $R$ is an achievable rate for quantum communication with
uniform energy constraint $P$. So this implies that
$Q(\mathcal{N},G,P)   \geq E(\mathcal{N},G,P(1-2\delta))$. However, since we have shown this
inequality to be true for all $\delta\in(0,1/2)$, we can then take a supremum over $\delta\in(0,1/2)$ to conclude that $Q(\mathcal{N},G,P)   \geq \sup_{\delta\in(0,1/2)}E(\mathcal{N},G,P(1-2\delta)) = E(\mathcal{N},G,P)$. So we conclude \eqref{eq:Q>=E}. We can
argue the other inequalities in \eqref{eq:Q<=P} and
\eqref{eq:K<=P}\ similarly, by applying Propositions~\ref{thm:QC-to-PC} and~\ref{thm:SK-to-PC}, respectively.
\end{proof}

\section{Achievability of regularized, energy-constrained coherent information
for energy-constrained quantum communication\label{sec:coh-info-ach}}

The main result of this section is Theorem~\ref{thm:coh-info-ach}, which shows
that the regularized energy-constrained coherent information is achievable for
energy-constrained quantum communication. In order to do so, we need to
restrict the energy observables and channels that we consider. We impose two
arguably natural constraints: that the energy observable be a Gibbs observable
as given in Definition~\ref{def:Gibbs-obs}\ and that the channel have finite
output entropy as given in Condition~\ref{cond:finite-out-entropy}. Gibbs
observables have been considered in several prior works
\cite{H03,H04,HS06,Holevo2010,H12,Winter15}\ as well as finite output-entropy
channels \cite{H03,H04,H12}.

When defining a Gibbs observable, we follow \cite[Lemma~11.8]{H12} and
\cite[Section~IV]{Winter15}:

\begin{definition}
[Gibbs observable]\label{def:Gibbs-obs}Let $G$ be an energy observable as
given in Definition~\ref{def:energy-obs}. Such an operator $G$ is a Gibbs
observable if for all $\beta>0$, the following holds%
\begin{equation}
\operatorname{Tr}\{\exp(-\beta G)\}<\infty. \label{eq:thermal-well-defined}%
\end{equation}

\end{definition}

The above condition implies that a Gibbs observable$~G$ always has a finite
value of the partition function $\operatorname{Tr}\{\exp(-\beta G)\}$ for all
$\beta>0$ and thus a well defined thermal state for all $\beta>0$, given by
$e^{-\beta G}/\operatorname{Tr}\{e^{-\beta G}\}$.

\begin{condition}
[Finite output entropy]\label{cond:finite-out-entropy}Let $G$ be a Gibbs
observable and $P\in\lbrack0,\infty)$. A quantum channel $\mathcal{N}$
satisfies the finite-output entropy condition with respect to $G$ and $P$ if%
\begin{equation}
\sup_{\rho:\operatorname{Tr}\{G\rho\}\leq P}H(\mathcal{N}(\rho))<\infty,
\label{eq:finiteness-cap}%
\end{equation}

\end{condition}

\begin{lemma}
\label{lem:env-out-ent}Let $\mathcal{N}$ denote a quantum channel satisfying
Condition~\ref{cond:finite-out-entropy}, $G$ a Gibbs observable, and
$P\in\lbrack0,\infty)$. Then any complementary channel $\mathcal{\hat{N}}$ of
$\mathcal{N}$ satisfies the finite-entropy condition%
\begin{equation}
\sup_{\rho:\operatorname{Tr}\{G\rho\}\leq P}H(\mathcal{\hat{N}}(\rho))<\infty.
\label{eq:finite-entropy-comp-ch}%
\end{equation}

\end{lemma}

\begin{proof}
Let $\rho$ be a density operator satisfying $\operatorname{Tr}\{G\rho\}\leq
P$, and let $\sum_{i}p_{i}|i\rangle\langle i|$ be a spectral decomposition of
$\rho$. Let%
\begin{equation}
\theta_{\beta}\equiv e^{-\beta G}/\operatorname{Tr}\{e^{-\beta G}\}
\end{equation}
denote a thermal state of $G$ with inverse temperature $\beta>0$. Consider
that $H(\rho)$ is finite because a rewriting of $D(\rho\Vert\theta_{\beta
})\geq0$ implies that%
\begin{align}
H(\rho)  &  \leq\beta\operatorname{Tr}\{G\rho\}+\log\operatorname{Tr}%
\{e^{-\beta G}\}\\
&  \leq\beta P+\log\operatorname{Tr}\{e^{-\beta G}\}<\infty,
\label{eq:input-finite-entropy}%
\end{align}
where the last inequality follows from \eqref{eq:thermal-well-defined} and
from the assumption that $P<\infty$. Consider that $|\psi^{\rho}\rangle
=\sum_{i}\sqrt{p_{i}}|i\rangle\otimes|i\rangle$ is a purification of $\rho$
and satisfies%
\begin{align}
H(\mathcal{\hat{N}}(\rho))  &  =H((\operatorname{id}\otimes\mathcal{N}%
)(|\psi^{\rho}\rangle\langle\psi^{\rho}|))\\
&  \leq H(\rho)+H(\mathcal{N}(\rho))<\infty.
\end{align}
The equality follows because the marginals of a pure bipartite state have the
same entropy. The first inequality follows from subadditivity of entropy, and
the last from \eqref{eq:input-finite-entropy} and the assumption that
Condition~\ref{cond:finite-out-entropy}\ holds. We have shown that the entropy
$H(\mathcal{\hat{N}}(\rho))$ is finite for all states satisfying
$\operatorname{Tr}\{G\rho\}\leq P$, and so \eqref{eq:finite-entropy-comp-ch} holds.
\end{proof}

\begin{theorem}
\label{thm:coh-info-ach}Let $\mathcal{N}:\mathcal{T}(\mathcal{H}%
_{A})\rightarrow\mathcal{T}(\mathcal{H}_{B})$ denote a quantum channel
satisfying Condition~\ref{cond:finite-out-entropy}, $G$ a Gibbs observable,
and $P\in\lbrack0,\infty)$. Then the energy-constrained entanglement
transmission capacity $E(\mathcal{N},G,P)$\ is bounded from below by the
regularized energy-constrained coherent information of the channel
$\mathcal{N}$:%
\begin{equation}
E(\mathcal{N},G,P)\geq\lim_{k\rightarrow\infty}\frac{1}{k}I_{c}(\mathcal{N}%
^{\otimes k},\overline{G}_{k},P),\nonumber
\end{equation}
where the energy-constrained coherent information of $\mathcal{N}$ is defined
as%
\begin{equation}
I_{c}(\mathcal{N},G,P)\equiv\sup_{\rho:\operatorname{Tr}\{G\rho\}\leq
P}H(\mathcal{N}(\rho))-H(\mathcal{\hat{N}}(\rho)),
\end{equation}
and $\mathcal{\hat{N}}$ denotes a complementary channel of $\mathcal{N}$.
\end{theorem}

\begin{proof}
The main challenge in proving this theorem is to have codes achieving the
coherent information while meeting the average energy constraint. We prove the
theorem by combining Klesse's technique for constructing entanglement
transmission codes \cite{K07,qcap2008second}\ with an adaptation of Holevo's
technique of approximation and constructing codes meeting an energy constraint
\cite{H03,H04}. We follow their arguments very closely and show how to combine
the techniques to achieve the desired result.

First, we recall what Klesse accomplished in \cite{K07} (see also the
companion paper \cite{qcap2008second}). Let $\mathcal{M}:\mathcal{T}%
(\mathcal{H}_{A})\rightarrow\mathcal{T}(\mathcal{H}_{B})$ denote a quantum
channel satisfying Condition~\ref{cond:finite-out-entropy} for some Gibbs
observable and energy constraint, so that the receiver entropy is finite, as
well as the environment entropy by Lemma~\ref{lem:env-out-ent}. This implies
that entropy-typical subspaces and sequences corresponding to these entropies
are well defined and finite, a fact of which we make use. Let $V$ denote a
finite-dimensional linear subspace of $\mathcal{H}_{A}$. Set $L\equiv\dim(V)$,
and let $\mathcal{L}$ denote a channel defined to be the restriction of
$\mathcal{M}$ to states with support contained in $V$. Let $\{K_{y}\}_{y}$ be
a set of Kraus operators for $\mathcal{M}$ and define the probability
$p_{Y}(y)$ by%
\begin{equation}
p_{Y}(y)\equiv\frac{1}{L}\operatorname{Tr}\{\Pi_{V}K_{y}^{\dag}K_{y}\Pi_{V}\},
\end{equation}
where $\Pi_{V}$ is a projection onto $V$. As discussed in \cite{K07}, there is
unitary freedom in the choice of the Kraus operators, and they can be chosen
\textquotedblleft diagonal,\textquotedblright\ so that $\operatorname{Tr}%
\{\Pi_{V}K_{y}^{\dag}K_{x}\Pi_{V}\}=0$ for $x\neq y$. Let $T_{Y}^{n,\delta}$
denote the $\delta$-entropy-typical set for $p_{Y}$, defined as%
\begin{equation}
T_{Y}^{n,\delta}\equiv\left\{  y^{n}:\left\vert -\left[  \log p_{Y^{n}}%
(y^{n})\right]  /n-H(Y)\right\vert \leq\delta\right\}  ,
\end{equation}
for integer $n\geq1$ and real $\delta>0$, where $p_{Y^{n}}(y^{n})\equiv
p_{Y}(y_{1})p_{Y}(y_{2})\cdots p_{Y}(y_{n})$. Let $K_{y^{n}}\equiv K_{y_{1}%
}\otimes K_{y_{2}}\otimes\cdots\otimes K_{y_{n}}$. Now define the
(trace-non-increasing) quantum operation $\mathcal{L}^{n,\delta}$ to be a map
consisting of only the entropy-typical Kraus operators $K_{y^{n}}$ such that
$y^{n}\in T_{Y}^{n,\delta}$. The number of such Kraus operators is no larger
than $2^{n\left[  H(Y)+\delta\right]  }$, and one can show that
$H(Y)=H(\mathcal{\hat{M}}(\pi_{V}))$, where $\mathcal{\hat{M}}$ is a channel
complementary to $\mathcal{M}$ and $\pi_{V}\equiv\Pi_{V}/L$ denotes the
maximally mixed state on $V$ \cite{K07}.

One can then further reduce the quantum operation $\mathcal{L}^{n,\delta}$ to
another one $\widetilde{\mathcal{L}}^{n,\delta}$ defined by projecting the
output of $\mathcal{L}^{n,\delta}$ to the entropy-typical subspace of the
density operator $\mathcal{L}(\pi_{V})=\mathcal{M}(\pi_{V})$. The
entropy-typical subspace of a density operator $\sigma$ with spectral
decomposition $\sigma=\sum_{z}p_{Z}(z)|z\rangle\langle z|$ is defined as%
\begin{equation}
T_{\sigma}^{n,\delta}\equiv\operatorname{span}\{|z^{n}\rangle:\left\vert
-\left[  \log p_{Z^{n}}(z^{n})\right]  /n-H(\sigma)\right\vert \leq\delta\},
\end{equation}
for integer $n\geq1$ and real $\delta>0$. The resulting quantum operation
$\widetilde{\mathcal{L}}^{n,\delta}$ is thus finite-dimensional and has a
finite number of Kraus operators. We then have the following bounds argued in
\cite{K07}:%
\begin{align}
\widetilde{L}^{n,\delta}  &  \leq2^{n\left[  H(\mathcal{\hat{M}}(\pi
_{V}))+\delta\right]  },\label{eq:klesse-1}\\
\operatorname{Tr}\{\widetilde{\mathcal{L}}^{n,\delta}(\pi_{V^{\otimes n}})\}
&  \geq1-\varepsilon_{1},\\
\left\Vert \widetilde{\mathcal{L}}^{n,\delta}(\pi_{V^{\otimes n}})\right\Vert
_{2}^{2}  &  \leq2^{-n\left[  H(\mathcal{M}(\pi_{V}))-3\delta\right]  },\\
F_{e}(C_{n},\mathcal{L}^{\otimes n})  &  \geq F_{e}(C_{n},\widetilde
{\mathcal{L}}^{n,\delta}), \label{eq:klesse-4}%
\end{align}
where $\widetilde{L}^{n,\delta}$ denotes the number of Kraus operators for
$\widetilde{\mathcal{L}}^{n,\delta}$ and the second inequality inequality
holds for all $\varepsilon_{1}\in(0,1)$ and sufficiently large $n$. Note that
for this latter estimate, we require the law of large numbers to hold when we
only know that the entropy is finite (this can be accomplished using the
technique discussed in \cite{T12}). In the last line, we have written the
entanglement fidelity of a code $C_{n}$ (some subspace of $V^{\otimes n}$),
which is defined as%
\begin{equation}
F_{e}(C_{n},\mathcal{L}^{\otimes n})\equiv\sup_{\mathcal{R}^{n}}\langle
\Phi_{C_{n}}|(\operatorname{id}\otimes\lbrack\mathcal{R}^{n}\circ
\mathcal{L}^{\otimes n}])(\Phi_{C_{n}})|\Phi_{C_{n}}\rangle,
\end{equation}
where $|\Phi_{C_{n}}\rangle$ denotes a maximally entangled state built from an
orthonormal basis of $C_{n}$ and the optimization is with respect to recovery
channels $\mathcal{R}^{n}$. Let $K_{n}\equiv\dim C_{n}$. From the developments
in \cite{K07}, the following bound holds%
\begin{multline}
\mathbb{E}_{U_{K_{n}}(V^{\otimes n})}\{F_{e}(U_{K_{n}}C_{n},\widetilde
{\mathcal{L}}^{n,\delta})\}\\
\geq\operatorname{Tr}\{\widetilde{\mathcal{L}}^{n,\delta}(\pi_{V^{\otimes n}%
})\}-\sqrt{K\widetilde{L}^{n,\delta}\left\Vert \widetilde{\mathcal{L}%
}^{n,\delta}(\pi_{V^{\otimes n}})\right\Vert _{2}^{2}},
\end{multline}
where $\mathbb{E}_{U_{K_{n}}(V^{\otimes n})}$ denotes the expected
entanglement fidelity when we apply a randomly selected unitary $U_{K_{n}}$ to
the codespace $C_{n}$, taking it to some different subspace of $V^{\otimes n}%
$. The unitary $U_{K}$ is selected according to the unitarily invariant
measure on the group $\mathbf{U}(V^{\otimes n})$ of unitaries acting on the
subspace $V^{\otimes n}$. Combining with the inequalities in
\eqref{eq:klesse-1}--\eqref{eq:klesse-4}, we find that%
\begin{multline}
\mathbb{E}_{U_{K_{n}}(V^{\otimes n})}\{F_{e}(U_{K_{n}}C_{n},\mathcal{L}%
^{\otimes n})\}\\
\geq1-\varepsilon_{1}-\left[  2^{-n\left[  H(\mathcal{M}(\pi_{V}%
))-\mathcal{\hat{M}}(\pi_{V}))-R-4\delta\right]  }\right]  ^{\frac{1}{2}},
\end{multline}
where the rate $R$ of entanglement transmission is defined as $R\equiv\left[
\log K_{n}\right]  /n$. Thus, if we choose%
\begin{equation}
R=H(\mathcal{M}(\pi_{V}))-\mathcal{\hat{M}}(\pi_{V}))-5\delta,
\end{equation}
then we find that%
\begin{equation}
\mathbb{E}_{U_{K_{n}}(V^{\otimes n})}\{F_{e}(U_{K_{n}}C_{n},\widetilde
{\mathcal{L}}^{n,\delta})\}\geq1-\varepsilon_{1}-2^{-n\delta/2},
\label{eq:good-q-klesse-code}%
\end{equation}
and we see that the RHS\ can be made arbitrarily close to one by taking $n$
large enough. We can then conclude that there exists a unitary $U_{K_{n}}$,
such that the codespace defined by $U_{K_{n}}C_{n}$ achieves the same
entanglement fidelity given above, implying that the rate $H(\mathcal{M}%
(\pi_{V}))-\mathcal{\hat{M}}(\pi_{V}))$ is achievable for entanglement
transmission over $\mathcal{M}$.

Now we apply the methods of Holevo \cite{H04}\ and further arguments of Klesse
\cite{K07}\ to see how to achieve the rate given in the statement of the
theorem for the channel$~\mathcal{N}$ while meeting the desired energy
constraint. We follow the reasoning in \cite{H04} very closely. Consider that
$G$ is a non-constant operator. Thus, the image of the convex set of all
density operators under the map $\rho\rightarrow\operatorname{Tr}\{G\rho\}$ is
an interval. Suppose first that $P$ is not equal to the minimum eigenvalue of
$G$. Then there exists a real number $P^{\prime}$ and a density operator
$\rho$ in $\mathcal{D}(\mathcal{H}_{A})$ such that%
\begin{equation}
\operatorname{Tr}\{G\rho\}\leq P^{\prime}<P.
\end{equation}
Let $\rho=\sum_{j=1}^{\infty}\lambda_{j}|j\rangle\langle j|$ be a spectral
decomposition of $\rho$, and define%
\begin{align}
\rho_{d}  &  \equiv\sum_{j=1}^{d}\tilde{\lambda}_{j}|j\rangle\langle
j|,\ \ \text{where}\\
\tilde{\lambda}_{j}  &  \equiv\lambda_{j}\left(  \sum_{j=1}^{d}\lambda
_{j}\right)  ^{-1}.
\end{align}
Then $\left\Vert \rho-\rho_{d}\right\Vert _{1}\rightarrow0$ as $d\rightarrow
\infty$. Let $g(j)\equiv\langle j|G|j\rangle$, so that%
\begin{equation}
\operatorname{Tr}\{G\rho_{d}\}=\sum_{j=1}^{d}\tilde{\lambda}_{j}%
g(j)=P^{\prime}+\varepsilon_{d},
\end{equation}
where $\varepsilon_{d}\rightarrow0$ as $d\rightarrow\infty$. Consider the
density operator $\rho_{d}^{\otimes m}$, and let $\Pi_{d}^{m,\delta}$ denote
its strongly typical projector, defined as the projection onto the strongly
typical subspace%
\begin{equation}
\operatorname{span}\{|j^{m}\rangle:\left\vert N(j|j^{m})/m-\tilde{\lambda}%
_{j}\right\vert \leq\delta\},
\end{equation}
where $|j^{m}\rangle\equiv|j_{1}\rangle\otimes\cdots\otimes|j_{m}\rangle$ and
$N(j|j^{m})$ denotes the number of appearances of the symbol $j$ in the
sequence $j^{m}$. Let%
\begin{equation}
\pi_{d}^{m,\delta}\equiv\Pi_{d}^{m,\delta}/\operatorname{Tr}\{\Pi
_{d}^{m,\delta}\}
\end{equation}
denote the maximally mixed state on the strongly typical subspace. We then
find that for positive integers$~m$ and$~n$,%
\begin{align}
&  \operatorname{Tr}\left\{  \overline{G}_{mn}\left(  \left[  \pi
_{d}^{m,\delta}\right]  ^{\otimes n}-\rho_{d}^{\otimes mn}\right)  \right\}
\nonumber\\
&  =\operatorname{Tr}\left\{  \overline{\left(  \overline{G}_{m}\right)  }%
_{n}\left(  \left[  \pi_{d}^{m,\delta}\right]  ^{\otimes n}-\rho_{d}^{\otimes
mn}\right)  \right\} \\
&  =\operatorname{Tr}\left\{  \overline{G}_{m}\left(  \pi_{d}^{m,\delta}%
-\rho_{d}^{\otimes m}\right)  \right\}  \leq\delta\max_{j\in\left[  d\right]
}g(j),
\end{align}
where $\left[  d\right]  \equiv\{1,\ldots,d\}$ and the inequality follows from
applying a bound from \cite{Hol01a}\ (also called \textquotedblleft typical
average lemma\textquotedblright\ in \cite{el2010lecture}). Now we can apply
the above inequality to find that%
\begin{align}
&  \operatorname{Tr}\left\{  \overline{G}_{mn}\left[  \pi_{d}^{m,\delta
}\right]  ^{\otimes n}\right\} \nonumber\\
&  \leq\operatorname{Tr}\{\overline{G}_{m}\rho_{d}^{\otimes m}\}+\delta
\max_{j\in\left[  d\right]  }g(j)\\
&  =\operatorname{Tr}\{G\rho_{d}\}+\delta\max_{j\in\left[  d\right]  }g(j)\\
&  =P^{\prime}+\varepsilon_{d}+\delta\max_{j\in\left[  d\right]  }g(j).
\end{align}
For all $d$ large enough, we can then find $\delta_{0}$ such that the last
line\ above is $\leq P/(1+\delta_{1})$ for $\delta,\delta_{1}\in(0,\delta
_{0}]$.

The quantum coding scheme we use is that of Klesse \cite{K07} discussed
previously, now setting $\mathcal{M}=\mathcal{N}^{\otimes m}$ and the subspace
$V$ to be the frequency-typical subspace of $\rho_{d}^{\otimes m}$, so that
$\Pi_{V}=\Pi_{d}^{m,\delta}$. Letting $\pi_{C_{n}}$ denote the maximally mixed
projector onto the codespace $C_{n}\subset V^{\otimes n}$, we find that
\cite[Section~5.3]{K07}%
\begin{equation}
\mathbb{E}_{U_{K_{n}}(V^{\otimes n})}\{U_{K_{n}}\pi_{C_{n}}U_{K_{n}}^{\dag
}\}=\pi_{V^{\otimes n}}=\left[  \pi_{d}^{m,\delta}\right]  ^{\otimes n}.
\end{equation}
So this and the reasoning directly above imply that%
\begin{equation}
\mathbb{E}_{U_{K_{n}}(V^{\otimes n})}\{\operatorname{Tr}\{\overline{G_{mn}%
}U_{K_{n}}\pi_{C_{n}}U_{K_{n}}^{\dag}\}\}\leq P/(1+\delta_{1}),
\end{equation}
for $\delta,\delta_{1}\leq\delta_{0}$. Furthermore, from
\eqref{eq:good-q-klesse-code}, for arbitrary $\varepsilon\in(0,1)$ and
sufficiently large $n$, we find that%
\begin{equation}
\mathbb{E}_{U_{K_{n}}(V^{\otimes n})}\{1-F_{e}(U_{K_{n}}C_{n},\mathcal{N}%
^{\otimes mn})\}\leq\varepsilon,
\end{equation}
as long as the rate%
\begin{equation}
R=[H(\mathcal{N}^{\otimes m}(\pi_{d}^{m,\delta}))-H(\mathcal{\hat{N}}^{\otimes
m}(\pi_{d}^{m,\delta}))]/m-\delta^{\prime}%
\end{equation}
for $\delta^{\prime}>0$. At this point, we would like to argue the existence
of a code that has arbitrarily small error and meets the energy constraint.
Let $E_{0}$ denote the event $1-F_{e}(U_{K_{n}}C_{n},\mathcal{N}^{\otimes
mn})\leq\sqrt{\varepsilon}$ and let $E_{1}$ denote the event
$\operatorname{Tr}\{\overline{G_{mn}}U_{K_{n}}\pi_{C_{n}}U_{K_{n}}^{\dag
}\}\leq P$. We can apply the union bound and Markov's inequality to find that%
\begin{align}
&  \Pr_{U_{K_{n}}(V^{\otimes n})}\{\overline{E_{0}\cap E_{1}}\}\nonumber\\
&  =\Pr_{U_{K_{n}}(V^{\otimes n})}\{E_{0}^{c}\cup E_{1}^{c}\}\\
&  \leq\Pr_{U_{K_{n}}(V^{\otimes n})}\{1-F_{e}(U_{K_{n}}C_{n},\mathcal{N}%
^{\otimes mn})\geq\sqrt{\varepsilon}\}\nonumber\\
&  \qquad+\Pr_{U_{K_{n}}(V^{\otimes n})}\left\{  \operatorname{Tr}%
\{\overline{G_{mn}}U_{K_{n}}\pi_{C_{n}}U_{K_{n}}^{\dag}\}\geq P\right\} \\
&  \leq\frac{1}{\sqrt{\varepsilon}}\mathbb{E}_{U_{K_{n}}(V^{\otimes n}%
)}\{1-F_{e}(U_{K_{n}}C_{n},\mathcal{N}^{\otimes mn})\}\nonumber\\
&  \qquad+\frac{1}{P}\mathbb{E}_{U_{K_{n}}(V^{\otimes n})}\{\operatorname{Tr}%
\{\overline{G_{mn}}U_{K_{n}}\pi_{C_{n}}U_{K_{n}}^{\dag}\}\}\\
&  \leq\sqrt{\varepsilon}+1/(1+\delta_{1}).
\end{align}
Since we can choose $n$ large enough to have $\varepsilon$ arbitrarily small,
there exists such an $n$ such that the last line is strictly less than one.
This then implies the existence of a code $C_{n}$ such that $F_{e}%
(C_{n},\mathcal{N}^{\otimes mn})\geq1-\sqrt{\varepsilon}$ and
$\operatorname{Tr}\{\overline{G_{mn}}\pi_{C_{n}}\}\leq P$ (i.e., it has
arbitrarily good entanglement fidelity and meets the average energy
constraint). Furthermore, the rate achievable using this code is equal to
$[H(\mathcal{N}^{\otimes m}(\pi_{d}^{m,\delta}))-H(\mathcal{\hat{N}}^{\otimes
m}(\pi_{d}^{m,\delta}))]/m$. We have shown that this rate is achievable for
all $\delta>0$ and all integer $m\geq1$. By applying the limiting argument
from \cite{Hol01a} (see also \cite{ieee2002bennett}), we thus have that the
following is an achievable rate as well:%
\begin{multline}
\lim_{\delta\rightarrow0}\lim_{m\rightarrow\infty}\frac{1}{m}[H(\mathcal{N}%
^{\otimes m}(\pi_{d}^{m,\delta}))-H(\mathcal{\hat{N}}^{\otimes m}(\pi
_{d}^{m,\delta}))]\\
=H(\mathcal{N}(\rho_{d}))-H(\mathcal{\hat{N}}(\rho_{d})),
\end{multline}
where $\operatorname{Tr}\{G\rho_{d}\}\leq P^{\prime}+\varepsilon_{d}\leq P$.
Given that both $H(\mathcal{N}(\rho_{d}))$ and $H(\mathcal{\hat{N}}(\rho
_{d}))$ are finite, we can apply
\eqref{eq:coh-info-def}--\eqref{eq:coh-info-ch-def} and rewrite
\begin{equation}
H(\mathcal{N}(\rho_{d}))-H(\mathcal{\hat{N}}(\rho_{d}))=I_{c}(\rho
_{d},\mathcal{N}).
\end{equation}
Finally, we take the limit $d\rightarrow\infty$ and find that%
\begin{equation}
\liminf_{d\rightarrow\infty}I_{c}(\rho_{d},\mathcal{N})\geq I_{c}%
(\rho,\mathcal{N}),
\end{equation}
where we have used the representation%
\begin{equation}
I_{c}(\rho_{d},\mathcal{N})=I(\rho_{d},\mathcal{N})-H(\rho_{d}),
\end{equation}
applied that the mutual information is lower semicontinuous
\cite[Proposition~1]{HS10}, the entropy $H$ is continuous for all states
$\sigma$ such that $\operatorname{Tr}\{G\sigma\}<P$ (following from a
variation of \cite[Lemma~11.8]{H12}), and the fact that a purification
$|\psi_{d}^{\rho}\rangle\equiv\sum_{j=1}^{d}\tilde{\lambda}_{j}^{1/2}%
|j\rangle\otimes|j\rangle$ has the convergence $\left\Vert |\psi_{d}^{\rho
}\rangle\langle\psi_{d}^{\rho}|-|\psi^{\rho}\rangle\langle\psi^{\rho
}|\right\Vert _{1}\rightarrow0$ as $d\rightarrow\infty$. Now since
$H(\mathcal{N}(\rho))$ and $H(\mathcal{\hat{N}}(\rho))$ are each finite, we
can rewrite%
\begin{equation}
I_{c}(\rho,\mathcal{N})=H(\mathcal{N}(\rho))-H(\mathcal{\hat{N}}(\rho)).
\end{equation}
We have thus proven that the rate $H(\mathcal{N}(\rho))-H(\mathcal{\hat{N}%
}(\rho))$ is achievable for entanglement transmission with average energy
constraint for all $\rho$ satisfying $\operatorname{Tr}\{G\rho\}<P$.

We can extend this argument to operators $\rho$ such that $\operatorname{Tr}%
\{G\rho\}=P$ by approximating them with operators $\rho_{\xi}=(1-\xi)\rho
+\xi|e\rangle\langle e|$, where $|e\rangle$ is chosen such that $\langle
e|G|e\rangle<P$. Suppose now that $P$ is the minimum eigenvalue of $G$. In
this case, the condition $\operatorname{Tr}\{G\rho\}\leq P$ reduces to the
support of $\rho$ being contained in the spectral projection of $G$
corresponding to this minimum eigenvalue. The condition in
Definition~\ref{def:Gibbs-obs}\ implies that the eigenvalues of $G$ have
finite multiplicity, and so the support of $\rho$ is a fixed
finite-dimensional subspace. Thus we can take $\rho_{d}=\rho$, and we can
repeat the above argument with the equality $\operatorname{Tr}\{G\rho\}=P$
holding at each step.

As a consequence, we can conclude that%
\begin{equation}
\sup_{\operatorname{Tr}\{G\rho\}\leq P}H(\mathcal{N}(\rho))-H(\mathcal{\hat
{N}}(\rho))
\end{equation}
is achievable as well. Finally, we can repeat the whole argument for all
$\rho^{(k)}\in\mathcal{D}(\mathcal{H}_{A}^{\otimes k})$ satisfying
$\operatorname{Tr}\{\overline{G}_{k}\rho^{(k)}\}\leq P$, take the channel as
$\mathcal{N}^{\otimes k}$, and conclude that the following rate is achievable:%
\begin{equation}
\frac{1}{k}\sup_{\operatorname{Tr}\{\overline{G}_{k}\rho^{(k)}\}\leq
P}H(\mathcal{N}^{\otimes k}(\rho^{(k)}))-H(\mathcal{\hat{N}}^{\otimes k}%
(\rho^{(k)})).
\end{equation}
Taking the limit as $k\rightarrow\infty$ gives the statement of the theorem.
\end{proof}

\section{Energy-constrained quantum and private capacity of degradable
channels}

\label{sec:degradable-channels}It is unknown how to compute the quantum and
private capacities of general channels, but if they are degradable, the task
simplifies considerably. That is, it is known from \cite{cmp2005dev}\ and
\cite{S08}, respectively, that both the unconstrained quantum and private
capacities of a degradable channel $\mathcal{N}$ are given by the following
formula:%
\begin{equation}
Q(\mathcal{N})=P(\mathcal{N})=\sup_{\rho}I_{c}(\rho,\mathcal{N}).
\end{equation}

Here we prove the following theorem, which holds for the energy-constrained
quantum and private capacities of a degradable channel $\mathcal{N}$:

\begin{theorem}
\label{thm:-energy-constr-q-p-cap}Let $G$ be a Gibbs observable and
$P\in\lbrack0,\infty)$. Let a quantum channel $\mathcal{N}$ be degradable and
satisfy Condition~\ref{cond:finite-out-entropy}. Then the energy-constrained
capacities $Q(\mathcal{N},G,P)$, $E(\mathcal{N},G,P)$, $P(\mathcal{N},G,P)$,
and $K(\mathcal{N},G,P)$\ are finite, equal, and given by the following
formula:%
\begin{equation}
\sup_{\rho:\operatorname{Tr}\{G\rho\}\leq P}H(\mathcal{N}(\rho
))-H(\mathcal{\hat{N}}(\rho)), \label{eq:coh-info-1-letter}%
\end{equation}
where $\mathcal{\hat{N}}$ denotes a complementary channel of $\mathcal{N}$.
\end{theorem}

\begin{proof}
That the quantity in \eqref{eq:coh-info-1-letter} is finite follows directly
from the assumption in Condition~\ref{cond:finite-out-entropy} and
Lemma~\ref{lem:env-out-ent}. From Theorem~\ref{thm:cap-relations}, we have
that%
\begin{align}
Q(\mathcal{N},G,P)  &  =E(\mathcal{N},G,P)\nonumber\\
&  \leq P(\mathcal{N},G,P)=K(\mathcal{N},G,P).
\end{align}
Theorem~\ref{thm:coh-info-ach}\ implies that the rate in
\eqref{eq:coh-info-1-letter} is achievable. So this gives that%
\begin{multline}
\sup_{\rho:\operatorname{Tr}\{G\rho\}\leq P}H(\mathcal{N}(\rho
))-H(\mathcal{\hat{N}}(\rho))\\
\leq Q(\mathcal{N},G,P)=E(\mathcal{N},G,P).
\end{multline}

To establish the theorem, it thus suffices to prove the following converse
inequality%
\begin{equation}
K(\mathcal{N},G,P)\leq\sup_{\rho:\operatorname{Tr}\{G\rho\}\leq P}%
H(\mathcal{N}(\rho))-H(\mathcal{\hat{N}}(\rho)).
\label{eq:key-less-than-coh-info}%
\end{equation}
To do so, we make use of several ideas from
\cite{ieee2005dev,cmp2005dev,S08,YHD05MQAC}. Consider an $(n,M,G,P,\varepsilon
)$ code for secret key transmission with an average energy constraint, as
described in Section~\ref{sec:SKT-AVG-code}. Using such a code, we take a
uniform distribution over the codewords, and the state resulting from an
isometric extension of the channel is as follows:%
\begin{equation}
\sigma_{\hat{M}B^{n}E^{n}}\equiv\frac{1}{M}\sum_{m=1}^{M}|m\rangle\langle
m|_{\hat{M}}\otimes\lbrack\mathcal{U}^{\mathcal{N}}]^{\otimes n}(\rho_{A^{n}%
}^{m}).
\end{equation}
Now consider that each codeword in such a code has a spectral decomposition as
follows:%
\begin{equation}
\rho_{A^{n}}^{m}\equiv\sum_{l=1}^{\infty}p_{L|\hat{M}}(l|m)|\psi^{l,m}%
\rangle\langle\psi^{l,m}|_{A^{n}},
\end{equation}
for a probability distribution $p_{L|\hat{M}}$ and some orthonormal basis
$\{|\psi^{l,m}\rangle_{A^{n}}\}_{l}$ for $\mathcal{H}_{A^{n}}$. Then the state
$\sigma_{\hat{M}B^{n}E^{n}}$ has the following extension:%
\begin{multline}
\sigma_{L\hat{M}B^{n}E^{n}}\equiv\frac{1}{M}\sum_{m=1}^{M}\sum_{l=1}^{\infty
}p_{L|\hat{M}}(l|m)|l\rangle\langle l|_{L}\otimes|m\rangle\langle m|_{\hat{M}%
}\\
\otimes\lbrack\mathcal{U}^{\mathcal{N}}]^{\otimes n}(|\psi^{l,m}\rangle
\langle\psi^{l,m}|_{A^{n}}).
\end{multline}
We can also define the state after the decoding measurement acts as%
\begin{multline}
\sigma_{L\hat{M}M^{\prime}E^{n}}\equiv\frac{1}{M}\sum_{m,m^{\prime}=1}^{M}%
\sum_{l=1}^{\infty}p_{L|\hat{M}}(l|m)|l\rangle\langle l|_{L}\otimes
|m\rangle\langle m|_{\hat{M}}\\
\otimes\operatorname{Tr}_{B^{n}}\{\Lambda_{B^{n}}^{m^{\prime}}[\mathcal{U}%
^{\mathcal{N}}]^{\otimes n}(|\psi^{l,m}\rangle\langle\psi^{l,m}|_{A^{n}%
})\}\otimes|m^{\prime}\rangle\langle m^{\prime}|_{M^{\prime}}.
\end{multline}

Let $\overline{\rho}_{A}$ denote the average single-channel input state,
defined as%
\begin{equation}
\overline{\rho}_{A}\equiv\frac{1}{Mn}\sum_{m=1}^{M}\sum_{i=1}^{n}%
\operatorname{Tr}_{A^{n}\backslash A_{i}}\{\rho_{A^{n}}^{m}\}.
\label{eq:avg-input-state-conv-prf}%
\end{equation}
Applying the partial trace and the assumption in
\eqref{eq:SKT-energy-constraint}, it follows that%
\begin{equation}
\operatorname{Tr}\{G\overline{\rho}_{A}\}=\frac{1}{M}\sum_{m=1}^{M}%
\operatorname{Tr}\{\overline{G}_{n}\rho_{A^{n}}^{m}\}\leq P.
\label{eq:conv-pf-energy-constr}%
\end{equation}
Let $\overline{\sigma}_{B}$ denote the average single-channel output state:%
\begin{equation}
\overline{\sigma}_{B}\equiv\mathcal{N}(\overline{\rho}_{A})=\frac{1}{n}%
\sum_{i=1}^{n}\operatorname{Tr}_{B^{n}\backslash B_{i}}\{\sigma_{B^{n}}\},
\end{equation}
and let $\overline{\sigma}_{E}$ denote the average single-channel environment
state:%
\begin{equation}
\overline{\sigma}_{E}\equiv\mathcal{\hat{N}}(\overline{\rho}_{A})=\frac{1}%
{n}\sum_{i=1}^{n}\operatorname{Tr}_{E^{n}\backslash E_{i}}\{\sigma_{E^{n}}\}.
\end{equation}
It follows from non-negativity, subadditivity of entropy, concavity of
entropy, \eqref{eq:conv-pf-energy-constr}, and the assumption that $G$ is a
Gibbs observable that%
\begin{multline}
0\leq H\left(  \frac{1}{M}\sum_{m=1}^{M}\rho_{A^{n}}^{m}\right) \\
\leq\sum_{i=1}^{n}H\left(  \frac{1}{M}\sum_{m=1}^{M}\operatorname{Tr}%
_{A^{n}\backslash A_{i}}\{\rho_{A^{n}}^{m}\}\right) \\
\leq nH(\overline{\rho}_{A})<\infty.
\end{multline}
Similar reasoning but applying Condition~\ref{cond:finite-out-entropy} implies
that%
\begin{equation}
0\leq H(B^{n})_{\sigma}\leq\sum_{i=1}^{n}H(B_{i})_{\sigma}\leq
nH(B)_{\overline{\sigma}}<\infty.
\end{equation}
Similar reasoning but applying Lemma~\ref{lem:env-out-ent}\ implies that%
\begin{equation}
0\leq H(E^{n})_{\sigma}\leq\sum_{i=1}^{n}H(E_{i})_{\sigma}\leq
nH(E)_{\overline{\sigma}}<\infty.
\end{equation}
Furthermore, the entropy $H(\hat{M})_{\sigma}=\log_{2}M$ because the reduced
state $\sigma_{M}$ is maximally mixed with dimension equal to $M$.

Our analysis makes use of several other entropic quantities, each of which we
need to argue is finitely bounded from above and below and thus can be added
or subtracted at will in our analysis. The quantities involved are as follows,
along with bounds for them \cite{Lindblad1973,K11,S15}:%
\begin{align}
0  &  \leq I(\hat{M};B^{n})_{\sigma}\leq\min\{\log_{2}M,nH(B)_{\overline
{\sigma}}\},\\
0  &  \leq I(\hat{M};E^{n})_{\sigma}\leq\min\{\log_{2}M,nH(E)_{\overline
{\sigma}}\},\\
0  &  \leq H(\hat{M}|E^{n})_{\sigma}\leq\log_{2}M,
\end{align}
as well as%
\begin{multline}
0\leq I(\hat{M}L;B^{n})_{\sigma},\ I(L;B^{n}|\hat{M})_{\sigma},\\
H(B^{n}|L\hat{M})_{\sigma}\leq nH(B)_{\overline{\sigma}},
\end{multline}
and%
\begin{multline}
0\leq I(\hat{M}L;E^{n})_{\sigma},\ I(L;E^{n}|\hat{M})_{\sigma},\\
H(E^{n}|L\hat{M})_{\sigma}\leq nH(E)_{\overline{\sigma}}.
\end{multline}

We now proceed with the converse proof:%
\begin{align}
\log_{2}M  &  =H(\hat{M})_{\sigma}\\
&  =I(\hat{M};M^{\prime})_{\sigma}+H(\hat{M}|M^{\prime})_{\sigma}\\
&  \leq I(\hat{M};M^{\prime})_{\sigma}+h_{2}(\varepsilon)+\varepsilon\log
_{2}(M-1)\\
&  \leq I(\hat{M};B^{n})_{\sigma}+h_{2}(\varepsilon)+\varepsilon\log_{2}M.
\label{eq:1st-block-last-line}%
\end{align}
The first equality follows because the entropy of a uniform distribution is
equal to the logarithm of its cardinality. The second equality is an identity.
The first inequality follows from applying Fano's inequality in
\eqref{eq:fano}\ to the condition in \eqref{eq:private-good-comm-SKT}. The
second inequality follows from applying the Holevo bound
\cite{Holevo73,PhysRevLett.70.363}. The direct sum property of the trace
distance and the security condition in \eqref{eq:security-cond-SKT} imply that%
\begin{multline}
\frac{1}{2}\left\Vert \sigma_{\hat{M}E^{n}}-\pi_{\hat{M}}\otimes\omega_{E^{n}%
}\right\Vert _{1}\\
=\frac{1}{M}\sum_{m=1}^{M}\frac{1}{2}\left\Vert \mathcal{\hat{N}}^{\otimes
n}(\rho_{A^{n}}^{m})-\omega_{E^{n}}\right\Vert _{1}\leq\varepsilon,
\end{multline}
which, by the AFW inequality in \eqref{eq:AFW-cq}\ for classical--quantum
states, means that%
\begin{equation}
\left\vert H(\hat{M}|E^{n})_{\pi\otimes\omega}-H(\hat{M}|E^{n})_{\sigma
}\right\vert \leq\varepsilon\log_{2}(M)+g(\varepsilon).
\end{equation}
But%
\begin{align}
&  H(\hat{M}|E^{n})_{\pi\otimes\omega}-H(\hat{M}|E^{n})_{\sigma}\nonumber\\
&  =H(\hat{M})_{\pi}-H(\hat{M}|E^{n})_{\sigma}\\
&  =H(\hat{M})_{\sigma}-H(\hat{M}|E^{n})_{\sigma}\\
&  =I(\hat{M};E^{n})_{\sigma},
\end{align}
so then%
\begin{equation}
I(\hat{M};E^{n})_{\sigma}\leq\varepsilon\log_{2}(M)+g(\varepsilon).
\label{eq:eve-holevo-upper}%
\end{equation}
Returning to \eqref{eq:1st-block-last-line}\ and inserting
\eqref{eq:eve-holevo-upper}, we find that%
\begin{multline}
\log_{2}M\leq I(\hat{M};B^{n})_{\sigma}-I(\hat{M};E^{n})_{\sigma
}\label{eq:priv-cap-conv-regul}\\
+2\varepsilon\log_{2}M+h_{2}(\varepsilon)+g(\varepsilon).
\end{multline}
We now focus on bounding the term $I(\hat{M};B^{n})_{\sigma}-I(\hat{M}%
;E^{n})_{\sigma}$:%
\begin{align}
&  I(\hat{M};B^{n})_{\sigma}-I(\hat{M};E^{n})_{\sigma}\nonumber\\
&  =I(\hat{M}L;B^{n})_{\sigma}-I(L;B^{n}|\hat{M})_{\sigma}\nonumber\\
&  \qquad-\left[  I(\hat{M}L;E^{n})_{\sigma}-I(L;E^{n}|\hat{M})_{\sigma
}\right] \\
&  =I(\hat{M}L;B^{n})_{\sigma}-I(\hat{M}L;E^{n})_{\sigma}\nonumber\\
&  \qquad-\left[  I(L;B^{n}|\hat{M})_{\sigma}-I(L;E^{n}|\hat{M})_{\sigma
}\right] \\
&  \leq I(\hat{M}L;B^{n})_{\sigma}-I(\hat{M}L;E^{n})_{\sigma}\\
&  =H(B^{n})_{\sigma}-H(B^{n}|L\hat{M})_{\sigma}\nonumber\\
&  \qquad-\left[  H(E^{n})_{\sigma}-H(E^{n}|L\hat{M})_{\sigma}\right] \\
&  =H(B^{n})_{\sigma}-H(B^{n}|L\hat{M})_{\sigma}\nonumber\\
&  \qquad-\left[  H(E^{n})_{\sigma}-H(B^{n}|L\hat{M})_{\sigma}\right] \\
&  =H(B^{n})_{\sigma}-H(E^{n})_{\sigma}. \label{eq:second-block-last-line}%
\end{align}
The first equality follows from the chain rule for mutual information. The
second equality follows from a rearrangement. The first inequality follows
from the assumption of degradability of the channel, which implies that Bob's
mutual information is never smaller than Eve's: $I(L;B^{n}|\hat{M})_{\sigma
}\geq I(L;E^{n}|\hat{M})_{\sigma}$. The third equality follows from
definitions. The fourth equality follows because the marginal entropies of a
pure state are equal, i.e.,%
\begin{align}
&  H(B^{n}|L\hat{M})_{\sigma}\nonumber\\
&  =\frac{1}{M}\sum_{l,m}p_{L|\hat{M}}(l|m)H(\operatorname{Tr}_{E^{n}%
}\{[\mathcal{U}^{\mathcal{N}}]^{\otimes n}(|\psi^{l,m}\rangle\langle\psi
^{l,m}|_{A^{n}})\})\nonumber\\
&  =\frac{1}{M}\sum_{l,m}p_{L|\hat{M}}(l|m)H(\operatorname{Tr}_{B^{n}%
}\{[\mathcal{U}^{\mathcal{N}}]^{\otimes n}(|\psi^{l,m}\rangle\langle\psi
^{l,m}|_{A^{n}})\})\nonumber\\
&  =H(E^{n}|L\hat{M})_{\sigma}.
\end{align}
Continuing, we have that%
\begin{align}
\eqref{eq:second-block-last-line}  &  =H(B_{1})_{\sigma}-H(E_{1})_{\sigma
}+H(B_{2}\cdots B_{n})_{\sigma}\nonumber\\
&  \qquad-H(E_{1}\cdots E_{n})_{\sigma}\nonumber\\
&  \qquad-\left[  I(B_{1};B_{2}\cdots B_{n})_{\sigma}-I(E_{1};E_{2}\cdots
E_{n})_{\sigma}\right] \\
&  \leq H(B_{1})_{\sigma}-H(E_{1})_{\sigma}\nonumber\\
&  \qquad+H(B_{2}\cdots B_{n})_{\sigma}-H(E_{1}\cdots E_{n})_{\sigma}\\
&  \leq\sum_{i=1}^{n}H(B_{i})_{\sigma}-H(E_{i})_{\sigma}\\
&  \leq n\left[  H(B)_{\mathcal{U}(\overline{\rho})}-H(E)_{\mathcal{U}%
(\overline{\rho})}\right] \\
&  \leq n\left[  \sup_{\rho:\operatorname{Tr}\{G\rho\}\leq P}H(\mathcal{N}%
(\rho))-H(\mathcal{\hat{N}}(\rho))\right]  .
\end{align}
The first equality follows by exploiting the definition of mutual information.
The first inequality follows from the assumption of degradability, which
implies that $I(B_{1};B_{2}\cdots B_{n})_{\sigma}\geq I(E_{1};E_{2}\cdots
E_{n})_{\sigma}$. The second inequality follows by iterating the argument. The
third inequality follows from the concavity of the coherent information for
degradable channels (Proposition~\ref{prop:concave-degrad}), with
$\overline{\rho}_{A}$ defined as in \eqref{eq:avg-input-state-conv-prf} and
satisfying \eqref{eq:conv-pf-energy-constr}. Thus, the final inequality
follows because we can optimize the coherent information with respect all
density operators satisfying the energy constraint.

Putting everything together and assuming that $\varepsilon<1/2$, we find the
following bound for all $\left(  n,M,G,P,\varepsilon\right)  $ private
communication codes:%
\begin{multline}
\left(  1-2\varepsilon\right)  \frac{1}{n}\log_{2}M-\frac{1}{n}\left[
h_{2}(\varepsilon)+g(\varepsilon)\right] \\
\leq\sup_{\rho:\operatorname{Tr}\{G\rho\}\leq P}H(\mathcal{N}(\rho
))-H(\mathcal{\hat{N}}(\rho)).
\end{multline}
Now taking the limit as $n\rightarrow\infty$ and then as $\varepsilon
\rightarrow0$, we can conclude the inequality in
\eqref{eq:key-less-than-coh-info}. This concludes the proof.
\end{proof}

\section{Regularized converses for energy-constrained quantum and private
capacity of general channels}

\label{sec:reg-converses}

In this section, we establish regularized converses for the energy-constrained
quantum and private capacities of general channels. We start with private
capacity, but before doing so, we should give some further background
(available in \cite{HS06,H12,S16cont}) and recall the definition of the energy-constrained
private information of a channel \cite{S16cont}. A generalized (continuous) ensemble
corresponds to a Borel probability measure on the set of quantum states. Let
$\mathcal{M}(\mathcal{H})$ denote the set of all Borel probability measures on
$\mathcal{D}(\mathcal{H})$ having the topology of weak convergence. The
average state $\overline{\rho}(\mu)$\ of a generalized ensemble $\mu
\in\mathcal{M}(\mathcal{H})$ is the barycenter of the measure $\mu$ defined by
the following Bochner integral:%
\begin{equation}
\overline{\rho}(\mu)\equiv \int_{\mathcal{D}(\mathcal{H})}\mu(d\rho)\ \rho.
\end{equation}
(The notation $\mu(d\rho)$ indicates that $\mu$ is a measure over all mixed states.)
We let $\mathcal{N}(\mu)$ denote the generalized ensemble resulting from
applying the channel to the states in the generalized ensemble specified by
$\mu$. The Holevo quantity for a generalized ensemble is defined as%
\begin{equation}
\chi(\mu)\equiv\int_{\mathcal{D}(\mathcal{H})}\mu(d\rho)\ D(\rho\Vert\overline
{\rho}(\mu)).
\end{equation}
The energy-constrained private information of a channel $\mathcal{N}$ is then defined as \cite{S16cont}
\begin{equation}
C_{p}(\mathcal{N},G,P)\equiv\sup_{\mu\in\mathcal{M}(\mathcal{H}%
):\operatorname{Tr}\{G\overline{\rho}(\mu)\}\leq P}\chi(\mathcal{N}(\mu
))-\chi(\mathcal{\hat{N}}(\mu)),
\label{eq:en-const-priv-info}
\end{equation}
where $\mathcal{\hat{N}}$ denotes a complementary channel of $\mathcal{N}$.
 We can now state our first result for general channels:

\begin{theorem}
\label{thm:-energy-constr-p-cap-regularized}Let $G$ be a Gibbs observable and
$P\in\lbrack0,\infty)$. Let a quantum channel $\mathcal{N}$ satisfy
Condition~\ref{cond:finite-out-entropy}. Then the energy-constrained
capacities $P(\mathcal{N},G,P)$ and $K(\mathcal{N},G,P)$ are finite, equal,
and bounded from above by the regularized energy-constrained private
information:%
\begin{equation}
P(\mathcal{N},G,P)=K(\mathcal{N},G,P)\leq\lim_{k\rightarrow\infty}\frac{1}%
{k}C_{p}(\mathcal{N}^{\otimes k},\overline{G}_{k},P).
\end{equation}

\end{theorem}

\begin{proof}
Theorem~\ref{thm:cap-relations} implies that%
\begin{equation}
P(\mathcal{N},G,P)=K(\mathcal{N},G,P).
\end{equation}
To establish the theorem stated above, it thus suffices to prove the following converse
inequality%
\begin{equation}
K(\mathcal{N},G,P)\leq\lim_{k\rightarrow\infty}\frac{1}{k}C_{p}(\mathcal{N}%
^{\otimes k},\overline{G}_{k},P).\label{eq:key-bounded-by-reg-priv}%
\end{equation}
To do so, we follow all of the steps of
Theorem~\ref{thm:-energy-constr-q-p-cap}\ until
\eqref{eq:priv-cap-conv-regul}. Now let $\mu_{0}\in\mathcal{M}(\mathcal{H}%
^{\otimes n})$ denote the discrete measure induced by the
$(n,M,G,P,\varepsilon)$ secret-key transmission code. For this measure, the
condition $\operatorname{Tr}\{\overline{G}_{n}\overline{\rho}(\mu_{0})\}\leq
P$ holds by definition, being the same as \eqref{eq:SKT-energy-constraint}.
Thus, picking up from \eqref{eq:priv-cap-conv-regul}, we obtain the following:%
\begin{align}
&  \!\!\!\!\!I(\hat{M};B^{n})_{\sigma}-I(\hat{M};E^{n})_{\sigma}\nonumber\\
&  =\chi(\mathcal{N}(\mu_{0}))-\chi(\mathcal{\hat{N}}(\mu_{0}))\\
&  \leq C_{p}(\mathcal{N}^{\otimes n},\overline{G}_{n},P),
\end{align}
with the inequality holding for the simple reason that we can never achieve a
smaller value by optimizing over all generalized ensembles satisfying the
energy constraint.\ We then conclude 
that%
\begin{multline}
(1-2\varepsilon)\frac{1}{n}\log_{2}M\leq\frac{1}{n}C_{p}(\mathcal{N}^{\otimes
n},\overline{G}_{n},P)\\
+\frac{1}{n}\left[  h_{2}(\varepsilon)+g(\varepsilon)\right]  .
\end{multline}
Now taking the limit as $n\rightarrow\infty$ and then as $\varepsilon
\rightarrow0$, we can conclude the inequality in
\eqref{eq:key-bounded-by-reg-priv}.
\end{proof}

\bigskip

We now turn to the quantum capacity:

\begin{theorem}
\label{thm:-energy-constr-q-cap-regularized}Let $G$ be a Gibbs observable and
$P\in\lbrack0,\infty)$. Let a quantum channel $\mathcal{N}$ satisfy
Condition~\ref{cond:finite-out-entropy}. Then the energy-constrained
capacities $Q(\mathcal{N},G,P)$ and $E(\mathcal{N},G,P)$\ are finite and equal
to the regularized energy-constrained coherent information:%
\begin{equation}
Q(\mathcal{N},G,P)=E(\mathcal{N},G,P)=\lim_{k\rightarrow\infty}\frac{1}%
{k}I_{c}(\mathcal{N}^{\otimes k},\overline{G}_{k},P).
\end{equation}

\end{theorem}

\begin{proof}
Theorem~\ref{thm:coh-info-ach} establishes the following lower bound:%
\begin{equation}
Q(\mathcal{N},G,P)\geq\lim_{k\rightarrow\infty}\frac{1}{k}I_{c}(\mathcal{N}%
^{\otimes k},\overline{G}_{k},P),
\end{equation}
and Theorem~\ref{thm:cap-relations}\ the following equality:%
\begin{equation}
Q(\mathcal{N},G,P)=E(\mathcal{N},G,P).
\end{equation}
We now establish the upper bound%
\begin{equation}
E(\mathcal{N},G,P)\leq\lim_{k\rightarrow\infty}\frac{1}{k}I_{c}(\mathcal{N}%
^{\otimes k},\overline{G}_{k},P).\label{eq:ent-trans-bnded-by-reg-coh-info}%
\end{equation}
Fix $\delta \in (0,1)$.
Consider an $(n,M,G,P(1-\delta),\varepsilon)$ code for entanglement transmission with an
average energy constraint, as described in
Section~\ref{sec:ent-trans-avg-code}. Let%
\begin{align}
\omega_{RS} &  \equiv(\operatorname{id}_{R}\otimes\lbrack\mathcal{D}^{n}%
\circ\mathcal{N}^{\otimes n}\circ\mathcal{E}^{n}])(\Phi_{RS}),\\
\kappa_{RB^{n}} &  \equiv(\operatorname{id}_{R}\otimes\lbrack\mathcal{N}%
^{\otimes n}\circ\mathcal{E}^{n}])(\Phi_{RS}),
\end{align}
where the symbols on the right-hand side are described in
Section~\ref{sec:ent-trans-avg-code}. Note that $M=\dim(\mathcal{H}_{R})$, by
definition. Let $\sum_{l}p(l)|\phi^{l}\rangle\langle\phi^{l}|_{RA^{n}}$ be a
spectral decomposition of the state $(\operatorname{id}_{R}\otimes$
$\mathcal{E}^{n})(\Phi_{RS})$, and define%
\begin{align}
\omega_{RS}^{l}  & \equiv(\operatorname{id}_{R}\otimes\lbrack\mathcal{D}%
^{n}\circ\mathcal{N}^{\otimes n}])(|\phi^{l}\rangle\langle\phi^{l}|_{RA^{n}%
}),\\
\kappa_{RB^{n}}^{l}  & \equiv(\operatorname{id}_{R}\otimes\mathcal{N}^{\otimes
n})(|\phi^{l}\rangle\langle\phi^{l}|_{RA^{n}}),
\end{align}
so that%
\begin{equation}
\omega_{RS}=\sum_{l}p(l)\omega_{RS}^{l},\quad\kappa_{RB^{n}}=\sum
_{l}p(l)\kappa_{RB^{n}}^{l}.
\end{equation}
By the condition in \eqref{eq:EG-good-reliable-code}, we have that%
\begin{align}
\varepsilon & \geq1-\langle\Phi|_{RS}\omega_{RS}|\Phi\rangle_{RS}\\
& =\sum_{l}p(l)\left[  1-\langle\Phi|_{RS}\omega_{RS}^{l}|\Phi\rangle
_{RS}\right]  \\
& \equiv\sum_{l}p(l)\hat{F}_{l}.
\end{align}
Also, the energy constraint in \eqref{eq:EG-avg-energy-constraint}\ implies
that%
\begin{align}
P(1-\delta)  & \geq\operatorname{Tr}\{\overline{G}_{n}\mathcal{E}^{n}(\pi_{S})\}\\
& =\sum_{l}p(l)\operatorname{Tr}\{\overline{G}_{n}\phi_{A^{n}}^{l}\}\\
& \equiv\sum_{l}p(l)E_{l}.
\end{align}
We would like to conclude that there exists at least one value of $l$ for
which the state $\phi_{RA^{n}}^{l}$ realizes a good entanglement generation
code, in the sense of \cite{ieee2005dev}, while at the same time meeting the
energy constraint. Let $L$ be a random variable with probability distribution
$p(l)$. By the union bound and Markov's inequality, for constant $\delta
\in(0,1)$, we have that%
\begin{align}
& \Pr_{L}\left\{  \left[  E_{L}\leq P\cap\hat{F}_{L}\leq
2\varepsilon/\delta\right]  ^{c}\right\}  \notag \\
& =\Pr_{L}\left\{  E_{L}>P\cup\hat{F}_{L}>2\varepsilon
/\delta\right\}  \\
& \leq\Pr_{L}\left\{  E_{L}>P\right\}  +\Pr_{L}\left\{  \hat{F}%
_{L}>2\varepsilon/\delta\right\}  \\
& \leq\frac{\mathbb{E}_{L}\{E_{L}\}}{P}+\frac{\mathbb{E}_{L}%
\{\hat{F}_{L}\}}{2\varepsilon/\delta}\\
& \leq\frac{P(1-\delta)}{P}+\frac{\mathbb{\varepsilon}}{2\varepsilon/\delta
}\\
& =1-\delta+\delta/2=1-\delta/2.
\end{align}
Thus, $\Pr_{L}\{E_{L}\leq P\cap\hat{F}_{L}\leq2\varepsilon
/\delta\}>\delta/2>0$, and we can conclude that there exists at least one
realization $l$ of $L$ for which the conditions $E_{l}\leq P$ and
$\hat{F}_{l}\leq2\varepsilon/\delta$ hold. For this value, we have by
\eqref{eq:F-vd-G}\ that%
\begin{equation}
\frac{1}{2}\left\Vert \omega_{RS}^{l}-\Phi_{RS}\right\Vert _{1}\leq
\sqrt{2\varepsilon/\delta}.\label{eq:EG-code-close-max-ent}%
\end{equation}
Now consider that%
\begin{align}
\log\dim(\mathcal{H}_{R}) &  =I(R\rangle S)_{\Phi}\\
&  \leq I(R\rangle S)_{\omega^{l}}+2\sqrt{2\varepsilon/\delta}\log
\dim(\mathcal{H}_{R})\nonumber\\
&  \qquad+g(\sqrt{2\varepsilon/\delta}).
\end{align}
The equality follows from a direct calculation and the inequality from
\eqref{eq:EG-code-close-max-ent} and the continuity bound in
\eqref{eq:AFW-ineq-CE}. Continuing, we have that%
\begin{align}
I(R\rangle S)_{\omega^{l}} &  \leq I(R\rangle B^{n})_{\kappa^{l}}\\
&  =H(\mathcal{N}^{\otimes n}(\phi_{A^{n}}^{l}))-H(\mathcal{\hat{N}}^{\otimes
n}(\phi_{A^{n}}^{l}))\\
&  \leq I_{c}(\mathcal{N}^{\otimes n},\overline{G}_{n},P).
\end{align}
The first inequality follows from data processing of coherent information
recalled in \eqref{eq:coh-info-DP}. The equality follows by rewriting the
coherent information, given that the various entropies involved are finite.
The final inequality follows because the definition of $I_{c}(\mathcal{N}%
^{\otimes n},\overline{G}_{n},P)$ involves an optimization with
respect to all input states $\rho^{(n)}$ satisfying $\operatorname{Tr}%
\{\overline{G}_{n}\rho^{(n)}\}\leq P$ and $\phi_{A^{n}}^{l}$ is one
such state. Putting everything together, we find that%
\begin{multline}
(1-2\sqrt{2\varepsilon/\delta})\frac{1}{n}\log\dim(\mathcal{H}_{R})\leq
\frac{1}{n}I_{c}(\mathcal{N}^{\otimes n},\overline{G}_{n},P)\\
+\frac{1}{n}g(\sqrt{2\varepsilon/\delta}).
\end{multline}
Now taking the limit as $n\rightarrow\infty$ and then as $\varepsilon
\rightarrow0$, we conclude that
\begin{equation}
E(\mathcal{N},G,P(1-\delta))\leq\lim_{k\rightarrow\infty}\frac{1}{k}I_{c}(\mathcal{N}%
^{\otimes k},\overline{G}_{k},P).
\end{equation}
However, we have proved that the above inequality holds for all $\delta\in(0,1)$,
and so we can take a supremum over $\delta\in(0,1)$ and arrive at the conclusion
that%
\begin{align}
\sup_{\delta\in(0,1)}E(\mathcal{N},G,P(1-\delta))  & = 
E(\mathcal{N},G,P) \\
& \leq \lim_{k\rightarrow\infty}\frac{1}{k}I_{c}(\mathcal{N}^{\otimes k}%
,\overline{G}_{k},P),
\end{align}
which is the inequality in \eqref{eq:ent-trans-bnded-by-reg-coh-info}. This
concludes the proof.
\end{proof}

\section{Thermal state as the optimizer}

In this section, we prove that the function%
\begin{equation}
\sup_{\operatorname{Tr}\{G\rho\}=P}H(\mathcal{N}(\rho))-H(\mathcal{\hat{N}%
}(\rho))
\end{equation}
is optimized by a thermal state input if the channel $\mathcal{N}$ is
degradable and satisfies certain other properties. In what follows, for a
Gibbs observable $G$, we define the thermal state $\theta_{\beta}$\ of inverse
temperature $\beta>0$ as%
\begin{equation}
\theta_{\beta}\equiv\frac{e^{-\beta G}}{\operatorname{Tr}\{e^{-\beta G}\}}.
\label{eq:thermal-state-beta-G}%
\end{equation}

\begin{theorem}
\label{thm:thermal-optimal-degrad}Let $G$ be a Gibbs observable and
$P\in\lbrack0,\infty)$. Let $\mathcal{N}:\mathcal{T}(\mathcal{H}%
_{A})\rightarrow\mathcal{T}(\mathcal{H}_{B})$ be a degradable quantum channel
satisfying Condition~\ref{cond:finite-out-entropy}. Let $\theta_{\beta}$
denote the thermal state of $G$, as in \eqref{eq:thermal-state-beta-G},
satisfying $\operatorname{Tr}\{G\theta_{\beta}\}=P$ for some $\beta>0$.
Suppose that $\mathcal{N}$ and a complementary channel $\mathcal{\hat{N}%
}:\mathcal{T}(\mathcal{H}_{A})\rightarrow\mathcal{T}(\mathcal{H}_{E})$ are
Gibbs preserving, in the sense that there exist $\beta_{1},\beta_{2}>0$ such
that%
\begin{equation}
\mathcal{N}(\theta_{\beta})=\theta_{\beta_{1}},\qquad\mathcal{\hat{N}}%
(\theta_{\beta})=\theta_{\beta_{2}}.
\end{equation}
Set%
\begin{equation}
P_{1}\equiv\operatorname{Tr}\{G\mathcal{N}(\theta_{\beta})\},\qquad
P_{2}\equiv\operatorname{Tr}\{G\mathcal{\hat{N}}(\theta_{\beta})\}.
\end{equation}
Suppose further that $\mathcal{N}$ and $\mathcal{\hat{N}}$\ are such that, for
all input states $\rho$ such that $\operatorname{Tr}\{G\rho\}=P$, the output
energies satisfy%
\begin{equation}
\operatorname{Tr}\{G\mathcal{N}(\rho)\}\leq P_{1},\qquad\operatorname{Tr}%
\{G\mathcal{\hat{N}}(\rho)\}\geq P_{2}.
\end{equation}
Then the function%
\begin{equation}
\sup_{\operatorname{Tr}\{G\rho\}=P}H(\mathcal{N}(\rho))-H(\mathcal{\hat{N}%
}(\rho)),
\end{equation}
is optimized by the thermal state $\theta_{\beta}$.
\end{theorem}

\begin{proof}
Let $\mathcal{D}:\mathcal{T}(\mathcal{H}_{B})\rightarrow\mathcal{T}%
(\mathcal{H}_{E})$ be a degrading channel such that $\mathcal{D}%
\circ\mathcal{N}=\mathcal{\hat{N}}$. Consider a state $\rho$ such that
$\operatorname{Tr}\{G\rho\}=P$. The monotonicity of quantum relative entropy
with respect to quantum channels (see \eqref{eq:mono-rel-ent})\ implies that%
\begin{align}
D(\mathcal{N}(\rho)\Vert\mathcal{N}(\theta_{\beta}))  &  \geq D((\mathcal{D}%
\circ\mathcal{N})(\rho)\Vert(\mathcal{D}\circ\mathcal{N})(\theta_{\beta}))\\
&  =D(\mathcal{\hat{N}}(\rho)\Vert\mathcal{\hat{N}}(\theta_{\beta})).
\end{align}
By the assumption of the theorem, this means that%
\begin{equation}
D(\mathcal{N}(\rho)\Vert\theta_{\beta_{1}})\geq D(\mathcal{\hat{N}}(\rho
)\Vert\theta_{\beta_{2}}),
\end{equation}
where $\beta_{1}$ and $\beta_{2}$ are such that $\operatorname{Tr}%
\{G\theta_{\beta_{1}}\}=P_{1}$ and $\operatorname{Tr}\{G\theta_{\beta_{2}%
}\}=P_{2}$. After a rewriting using definitions and the fact that all terms
below are finite, the inequality above becomes%
\begin{multline}
\operatorname{Tr}\{\mathcal{\hat{N}}(\rho)\log\theta_{\beta_{2}}%
\}-\operatorname{Tr}\{\mathcal{N}(\rho)\log\theta_{\beta_{1}}\}\\
\geq H(\mathcal{N}(\rho))-H(\mathcal{\hat{N}}(\rho)).
\end{multline}
Set $Z_{1}\equiv\operatorname{Tr}\{e^{-\beta_{1}G}\}$ and $Z_{2}%
\equiv\operatorname{Tr}\{e^{-\beta_{2}G}\}$. We can then rewrite the upper
bound as%
\begin{align}
&  \operatorname{Tr}\{\mathcal{\hat{N}}(\rho)\log\theta_{\beta_{2}%
}\}-\operatorname{Tr}\{\mathcal{N}(\rho)\log\theta_{\beta_{1}}\}\nonumber\\
&  =\operatorname{Tr}\{\mathcal{\hat{N}}(\rho)\log\left[  e^{-\beta_{2}%
G}/Z_{2}\right]  \}\nonumber\\
&  \qquad-\operatorname{Tr}\{\mathcal{N}(\rho)\log\left[  e^{-\beta_{1}%
G}/Z_{1}\right]  \}\\
&  =\log\left[  Z_{1}/Z_{2}\right]  -\beta_{2}\operatorname{Tr}%
\{G\mathcal{\hat{N}}(\rho)\}+\beta_{1}\operatorname{Tr}\{G\mathcal{N}%
(\rho)\}\\
&  \leq\log\left[  Z_{1}/Z_{2}\right]  -\beta_{2}P_{2}+\beta_{1}P_{1}.
\end{align}
Thus, we have established a uniform upper bound on the coherent information of
states subject to the constraints given in the theorem:%
\begin{equation}
H(\mathcal{N}(\rho))-H(\mathcal{\hat{N}}(\rho))\leq\log\left[  Z_{1}%
/Z_{2}\right]  -\beta_{2}P_{2}+\beta_{1}P_{1}.
\end{equation}
This bound is saturated when we choose the input $\rho=\theta_{\beta}$, where
$\beta$ is such that $\operatorname{Tr}\{G\theta_{\beta}\}=P$, because%
\begin{equation}
\log\left[  Z_{1}/Z_{2}\right]  -\beta_{2}P_{2}+\beta_{1}P_{1}=H(\mathcal{N}%
(\theta_{\beta}))-H(\mathcal{\hat{N}}(\theta_{\beta})).
\end{equation}
This concludes the proof.
\end{proof}

\begin{remark}
Note that we can also conclude that $P_{1}\geq P_{2}$ for channels satisfying
the hypotheses of the above theorem because the channel is degradable,
implying that $H(\theta_{\beta_{1}})\geq H(\theta_{\beta_{2}})$, and the
entropy of a thermal state is a strictly increasing function of the energy
(and thus invertible) \cite[Proposition~10]{Winter15}.
\end{remark}

\begin{remark}
The assumptions in Theorem~\ref{thm:thermal-optimal-degrad} might seem
somewhat artificial, but the next section demonstrates several natural
examples of channels that satisfy the assumptions.
\end{remark}

\section{Application to Gaussian quantum channels\label{sec:Gaussian-results}}

We can now apply all of the results from previous sections to the particular
case of quantum bosonic Gaussian channels \cite{CEGH08,S17}. These
channels model natural physical processes such as photon loss, photon
amplification, thermalizing noise, or random kicks in phase space. They
satisfy Condition~\ref{cond:finite-out-entropy} when the Gibbs observable for
$m$ modes is taken to be%
\begin{equation}
\hat{E}_{m}\equiv\sum_{j=1}^{m}\omega_{j}\hat{a}_{j}^{\dag}\hat{a}_{j},
\label{eq:photon-num-op-freqs}%
\end{equation}
where $\omega_{j}>0$ is the frequency of the $j$th mode and $\hat{a}_{j}$ is
the photon annihilation operator for the $j$th mode, so that $\hat{a}%
_{j}^{\dag}\hat{a}_{j}$ is the photon number operator for the $j$th mode.

We start with a brief review of Gaussian states and channels (see
\cite{CEGH08,adesso14,S17} for more comprehensive reviews, but note that
here we mostly follow the conventions of \cite{CEGH08}). Let%
\begin{equation}
\hat{R}\equiv\left[  \hat{q}_{1},\ldots,\hat{q}_{m},\hat{p}_{1},\ldots,\hat
{p}_{m}\right]  \equiv\left[  \hat{x}_{1},\ldots,\hat{x}_{2m}\right]
\end{equation}
denote a row vector of position- and momentum-quadrature operators, satisfying
the canonical commutation relations:%
\begin{equation}
\left[  \hat{R}_{j},\hat{R}_{k}\right]  =i\Omega_{j,k},\quad\text{where}%
\quad\Omega\equiv%
\begin{bmatrix}
0 & 1\\
-1 & 0
\end{bmatrix}
\otimes I_{m},
\end{equation}
and $I_{m}$ denotes the $m\times m$ identity matrix. We take the annihilation
operator for the $j$th mode as $\hat{a}_{j}=(\hat{q}_{j}+i\hat{p}_{j}%
)/\sqrt{2}$. For $z$ a column vector in $\mathbb{R}^{2m}$, we define the
unitary displacement operator $D(z)=D^{\dagger}(-z)\equiv\exp(i\hat{R}z)$.
Displacement operators satisfy the following relation:
\begin{equation}
D(z)D(z^{\prime})=D(z+z^{\prime})\exp\!\left(  -\frac{i}{2}z^{T} \Omega
z^{\prime}\right)  .
\end{equation}
Every state $\rho\in\mathcal{D}(\mathcal{H})$ has a corresponding Wigner
characteristic function, defined as%
\begin{equation}
\chi_{\rho}(z)\equiv\operatorname{Tr}\{D(z)\rho\},
\end{equation}
and from which we can obtain the state $\rho$ as%
\begin{equation}
\rho=\int\frac{d^{2m}z}{\left(  2\pi\right)  ^{m}} \ \chi_{\rho}(z)\ D^{\dag
}(z).
\end{equation}
A quantum state $\rho$ is Gaussian if its Wigner characteristic function has a
Gaussian form as%
\begin{equation}
\chi_{\rho}(\xi)=\exp\left(  -\frac{1}{4}z^{T}V^{\rho}z+i\left[  \mu^{\rho
}\right]  ^{T}z\right)  ,
\end{equation}
where $\mu^{\rho}$ is the $2m\times1$ mean vector of $\rho$, whose entries are
defined by $\mu_{j}^{\rho}\equiv\langle\hat{R}_{j}\rangle_{\rho}$ and
$V^{\rho}$ is the $2m\times2m$ covariance matrix of $\rho$, whose entries are
defined as%
\begin{equation}
V_{j,k}^{\rho}\equiv\langle\{\hat{R}_{j}-\mu_{j}^{\rho},\hat{R}_{k}-\mu
_{k}^{\rho}\}\rangle_{\rho}.
\end{equation}
The following condition holds for a valid covariance matrix: $V\geq i\Omega$,
which is a manifestation of the uncertainty principle.

A thermal Gaussian state $\theta_{\beta}$ of $m$ modes with respect to
$\hat{E}_{m}$ from \eqref{eq:photon-num-op-freqs}\ and having inverse
temperature $\beta>0$ thus has the following form:%
\begin{equation}
\theta_{\beta}=e^{-\beta\hat{E}_{m}}/\operatorname{Tr}\{e^{-\beta\hat{E}_{m}%
}\}, \label{eq:thermal-E-m-op}%
\end{equation}
and has a mean vector equal to zero and a diagonal $2m\times2m$ covariance
matrix. One can calculate that the photon number in this state is equal to%
\begin{equation}
\sum_{j}\frac{1}{e^{\beta\omega_{j}}-1}.
\end{equation}
It is also well known that thermal states can be written as a Gaussian mixture
of displacement operators acting on the vacuum state:%
\begin{equation}
\theta_{\beta}=\int d^{2m}\xi\ p(\xi)\ D(\xi)\left[  |0\rangle\langle
0|\right]  ^{\otimes m}D^{\dag}(\xi),
\end{equation}
where $p(\xi)$ is a zero-mean, circularly symmetric Gaussian distribution.
From this, it also follows that randomly displacing a thermal state in such a
way leads to another thermal state of higher temperature:%
\begin{equation}
\theta_{\beta}=\int d^{2m}\xi\ q(\xi)\ D(\xi)\theta_{\beta^{\prime}}D^{\dag
}(\xi), \label{eq:displaced-thermal-is-thermal}%
\end{equation}
where $\beta^{\prime}\geq\beta$ and $q(\xi)$ is a particular circularly
symmetric Gaussian distribution.

A $2m\times2m$ matrix $S$ is symplectic if it preserves the symplectic form:
$S\Omega S^{T}=\Omega$. According to Williamson's theorem \cite{W36}, there is
a diagonalization of the covariance matrix $V^{\rho}$ of the form,
\begin{equation}
V^{\rho}=S^{\rho}\left(  D^{\rho}\oplus D^{\rho}\right)  \left(  S^{\rho
}\right)  ^{T},
\end{equation}
where $S^{\rho}$ is a symplectic matrix and $D^{\rho}\equiv\operatorname{diag}%
(\nu_{1},\ldots,\nu_{m})$ is a diagonal matrix of symplectic eigenvalues such
that $\nu_{i}\geq1$ for all $i\in\left\{  1,\ldots,m\right\}  $. Computing
this decomposition is equivalent to diagonalizing the matrix $iV^{\rho}\Omega$
\cite[Appendix~A]{WTLB16}.

The entropy $H(\rho)$\ of a quantum Gaussian state $\rho$\ is a direct
function of the symplectic eigenvalues of its covariance matrix $V^{\rho}$
\cite{CEGH08}:%
\begin{equation}
H(\rho)=\sum_{j=1}^{m}g((\nu_{j}-1)/2)\equiv g(V^{\rho}),
\end{equation}
where $g(\cdot)$ is defined in \eqref{eq:g-function}\ and we have indicated a
shorthand for this entropy as $g(V^{\rho})$.

The Hilbert--Schmidt adjoint of a Gaussian quantum channel $\mathcal{N}_{X,Y}$\ from $m$ modes to $m$ modes
has the following effect on a displacement operator $D(z)$ \cite{CEGH08}:%
\begin{equation}
D(z)\longmapsto D(Xz)\exp\left(  -\frac{1}{4}z^{T}Yz+iz^{T}d\right)  ,
\end{equation}
where $X$ is a real $2m\times2m$ matrix, $Y$ is a real $2m\times2m$ positive
semi-definite matrix, and $d\in\mathbb{R}^{2m}$, such that they satisfy%
\begin{equation}
Y-i\Omega+iX^{T}\Omega X\geq0.
\end{equation}
The effect of the channel on the mean vector $\mu^{\rho}$ and the covariance
matrix $V^{\rho}$\ is thus as follows:%
\begin{align}
\mu^{\rho}  &  \longmapsto X^{T}\mu^{\rho}+d,\\
V^{\rho}  &  \longmapsto X^{T}V^{\rho}X+Y.
\end{align}
All Gaussian channels are covariant with respect to displacement operators.
That is, the following relation holds%
\begin{equation}
\mathcal{N}_{X,Y}(D(z)\rho D^{\dag}(z))=D(X^{T}z)\mathcal{N}_{X,Y}%
(\rho)D^{\dag}(X^{T}z). \label{eq:covariance-gaussian}%
\end{equation}

Just as every quantum channel can be implemented as a unitary transformation
on a larger space followed by a partial trace, so can Gaussian channels be
implemented as a Gaussian unitary on a larger space with some extra modes
prepared in the vacuum state, followed by a partial trace \cite{CEGH08}. Given
a Gaussian channel $\mathcal{N}_{X,Y}$\ with $Z$ such that $Y=ZZ^{T}$ we can
find two other matrices $X_{E}$ and $Z_{E}$ such that there is a symplectic
matrix
\begin{equation}
S=%
\begin{bmatrix}
X^{T} & Z\\
X_{E}^{T} & Z_{E}%
\end{bmatrix}
, \label{eq:gaussian-dilation}%
\end{equation}
which corresponds to the Gaussian unitary transformation on a larger space.
The complementary channel $\mathcal{\hat{N}}_{X_{E},Y_{E}}$\ from input to the
environment then effects the following transformation on mean vectors and
covariance matrices:%
\begin{align}
\mu^{\rho}  &  \longmapsto X_{E}^{T}\mu^{\rho},\\
V^{\rho}  &  \longmapsto X_{E}^{T}V^{\rho}X_{E}+Y_{E},
\end{align}
where $Y_{E}\equiv Z_{E}Z_{E}^{T}$.

A quantum Gaussian channel for which $X=X^{\prime}\oplus X^{\prime}$,
$Y=Y^{\prime}\oplus Y^{\prime}$, and $d=d^{\prime}\oplus d^{\prime}$ is known
as a phase-insensitive Gaussian channel, because it does not have a bias to
either quadrature when applying noise to the input state.

The main result of this section is the following theorem, which gives an
explicit expression for the energy-constrained capacities of all
phase-insensitive degradable Gaussian channels that satisfy the conditions of
Theorem~\ref{thm:thermal-optimal-degrad} for all $\beta>0$:

\begin{theorem}
\label{thm:PI-Gauss-degrad-caps}Let $\mathcal{N}_{X,Y}$ be a phase-insensitive
degradable Gaussian channel, having a dilation of the form in
\eqref{eq:gaussian-dilation}. Suppose that $\mathcal{N}_{X,Y}$\ satisfies the
conditions of Theorem~\ref{thm:thermal-optimal-degrad} for all $\beta>0$. Then
its energy-constrained capacities $Q(\mathcal{N}_{X,Y},\hat{E}_{m},P)$,
$E(\mathcal{N}_{X,Y},\hat{E}_{m},P)$, $P(\mathcal{N}_{X,Y},\hat{E}_{m},P)$,
and $K(\mathcal{N}_{X,Y},\hat{E}_{m},P)$ are equal and given by the following
formula:%
\begin{equation}
g(X^{T}V^{\theta_{\beta}}X+Y)-g(X_{E}^{T}V^{\theta_{\beta}}X_{E}+Y_{E}),
\end{equation}
where $\theta_{\beta}$ is a thermal state of mean photon number $P$.
\end{theorem}

\begin{proof}
Since the channel is degradable, satisfies
Condition~\ref{cond:finite-out-entropy}, and $\hat{E}_{m}$ is a Gibbs
observable, Theorem~\ref{thm:-energy-constr-q-p-cap}\ applies and these
capacities are given by the following formula:%
\begin{equation}
\sup_{\rho:\operatorname{Tr}\{\hat{E}_{m}\rho\}\leq P}H(\mathcal{N}_{X,Y}%
(\rho))-H(\mathcal{\hat{N}}_{X_{E},Y_{E}}(\rho)).
\end{equation}
By assumption, the channel satisfies the conditions of
Theorem~\ref{thm:thermal-optimal-degrad}\ as well for all $\beta>0$, so that
the following function is optimized by a thermal state $\theta_{\beta}$ of
mean photon number $P$:%
\begin{multline}
\sup_{\rho:\operatorname{Tr}\{\hat{E}_{m}\rho\}=P}H(\mathcal{N}_{X,Y}%
(\rho))-H(\mathcal{\hat{N}}_{X_{E},Y_{E}}(\rho))\\
=H(\mathcal{N}_{X,Y}(\theta_{\beta}))-H(\mathcal{\hat{N}}_{X_{E},Y_{E}}%
(\theta_{\beta})).
\end{multline}
It thus remains to prove that $H(\mathcal{N}_{X,Y}(\theta_{\beta
}))-H(\mathcal{\hat{N}}_{X_{E},Y_{E}}(\theta_{\beta}))$ is increasing with
decreasing $\beta$. This follows from the covariance property in
\eqref{eq:covariance-gaussian}, the concavity of coherent information in the
input for degradable channels (Proposition~\ref{prop:concave-degrad}), and the
fact that thermal states can be realized by random Gaussian displacements of
thermal states with lower temperature. Consider that%
\begin{align}
&  H(\mathcal{N}_{X,Y}(\theta_{\beta^{\prime}}))-H(\mathcal{\hat{N}}%
_{X_{E},Y_{E}}(\theta_{\beta^{\prime}}))\nonumber\\
&  =\int d^{2m}\xi\ q(\xi)\ \left[  H(\mathcal{N}_{X,Y}(\theta_{\beta^{\prime
}}))-H(\mathcal{\hat{N}}_{X_{E},Y_{E}}(\theta_{\beta^{\prime}}))\right] \\
&  =\int d^{2m}\xi\ q(\xi)\ \Big[H(D(X\xi)\mathcal{N}_{X,Y}(\theta
_{\beta^{\prime}})D^{\dag}(X\xi))\nonumber\\
&  \qquad-H(D(X_{E}\xi)\mathcal{\hat{N}}_{X_{E},Y_{E}}(\theta_{\beta^{\prime}%
})D^{\dag}(X_{E}\xi))\Big]\\
&  =\int d^{2m}\xi\ q(\xi)\ \Big[H(\mathcal{N}_{X,Y}(D(\xi)\theta
_{\beta^{\prime}}D^{\dag}(\xi)))\nonumber\\
&  \qquad-H(\mathcal{\hat{N}}_{X_{E},Y_{E}}(D(\xi)\theta_{\beta^{\prime}%
}D^{\dag}(\xi)))\Big]\\
&  \leq H(\mathcal{N}_{X,Y}(\theta_{\beta}))-H(\mathcal{\hat{N}}_{X_{E},Y_{E}%
}(\theta_{\beta})).
\end{align}
The first equality follows by placing a probability distribution in front, and
the second follows from the unitary invariance of quantum entropy. The third
equality follows from the covariance property of quantum Gaussian channels,
given in \eqref{eq:covariance-gaussian}. The inequality follows because the
coherent information of degradable channels is concave in the input state
(Proposition~\ref{prop:concave-degrad}) and from \eqref{eq:displaced-thermal-is-thermal}.
\end{proof}

\subsection{Special cases:\ Single-mode pure-loss and quantum-limited amplifier channels}

We can now discuss some special cases of the above result, some of which have
already been known in the literature. Suppose that the channel is a
single-mode pure-loss channel $\mathcal{L}_{\eta}$, where $\eta\in\left[
1/2,1\right]  $ characterizes the average fraction of photons that make it
through the channel from sender to receiver \footnote{We do not consider
transmissivities $\eta\in\left[  0,1/2\right]  $ because the quantum capacity
vanishes in this range since the channel becomes antidegradable.}. In this
case, the channel has $X=\sqrt{\eta}I_{2}$ and $Y=(1-\eta)I_{2}$. We take the
Gibbs observable to be the photon-number operator $\hat{a}^{\dag}\hat{a}$ and
the energy constraint to be $N_{S}\in\lbrack0,\infty)$. Such a channel is
degradable \cite{CG06} and was conjectured \cite{GSE08}\ to have
energy-constrained quantum and private capacities equal to%
\begin{equation}
g(\eta N_{S})-g((1-\eta)N_{S}). \label{eq:q-cap-loss}%
\end{equation}
This conjecture was proven for the quantum capacity in \cite[Theorem~8]%
{PhysRevA.86.062306}, and the present paper establishes the statement for
private capacity. This was argued by exploiting particular properties of the
$g$ function (established in great detail in \cite{G08thesis}) to show that
the thermal state input is optimal for any fixed energy constraint. Here we
can see this latter result as a consequence of the more general statements in
Theorems~\ref{thm:thermal-optimal-degrad} and \ref{thm:PI-Gauss-degrad-caps},
which are based on the monotonicity of relative entropy and other properties
of this channel, such as covariance and degradability. Taking the limit
$N_{S}\rightarrow\infty$, the formula in \eqref{eq:q-cap-loss}\ converges to%
\begin{equation}
\log_{2}(\eta/[1-\eta]), \label{eq:loss-unconstrained}%
\end{equation}
which is consistent with the formula stated in \cite{WPG07}.

Suppose that the channel is a single-mode quantum-limited amplifier channel
$\mathcal{A}_{\kappa}$ of gain $\kappa\geq1$. In this case, the channel has
$X=\sqrt{\kappa}I_{2}$ and $Y=(\kappa-1)I_{2}$. Again we take the energy
operator and constraint as above. This channel is degradable \cite{CG06} and
was recently proven \cite{QW16}\ to have energy-constrained quantum and
private capacity equal to%
\begin{equation}
g(\kappa N_{S}+\kappa-1)-g(\left[  \kappa-1\right]  \left[  N_{S}+1\right]  ).
\label{eq:amp-q-cap-constrained}
\end{equation}
The result was established by exploiting particular properties of the $g$
function in addition to other arguments. However, we can again see this result
as a consequence of the more general statements given in
Theorems~\ref{thm:thermal-optimal-degrad} and \ref{thm:PI-Gauss-degrad-caps}.
Taking the limit $N_{S}\rightarrow\infty$, the formula converges to%
\begin{equation}
\log_{2}(\kappa/\left[  \kappa-1\right]  ), \label{eq:amp-unconstrained}%
\end{equation}
which is consistent with the formula stated in \cite{WPG07} and recently
proven in \cite{PLOB15,WTB16}.

\begin{remark}
\label{rem:WPG07}Ref.~\cite{WPG07}\ has been widely accepted to have provided
a complete proof of the unconstrained quantum capacity formulas given in
\eqref{eq:loss-unconstrained} and \eqref{eq:amp-unconstrained}. The important
developments of \cite{WPG07} were to identify that it suffices to optimize
coherent information of these channels with respect to a single channel use
and Gaussian input states. The issue is that \cite{WPG07}\ relied on an
\textquotedblleft optimization procedure carried out in\textquotedblright%
\ \cite{HW01} in order to establish the infinite-energy quantum capacity
formula given there (see just before \cite[Eq.~(12)]{WPG07}). However, a
careful inspection of \cite[Section~V-B]{HW01}\ reveals that no explicit
optimization procedure is given there. The contentious point is that it is
necessary to show that, among all Gaussian states, the thermal state is the
input state optimizing the coherent information of the quantum-limited
attenuator and amplifier channels. This point is not argued or in any way
justified in \cite[Section~V-B]{HW01} or in any subsequent work or review on
the topic \cite{Holevo2007,CGH06,HG12,H12}. As a consequence, we have been
left to conclude that the proof from \cite{WPG07} features a gap which was
subsequently closed in \cite[Section III-G-1]{PhysRevA.86.062306} and
\cite{QW16}. The result in \cite{PLOB15,WTB16} gives a completely different
approach for establishing the unconstrained quantum and private capacities of
the quantum-limited amplifier channel, which preceded the development in
\cite{QW16}.
\end{remark}

\subsection{Special cases:\ Multi-mode pure-loss and quantum-limited amplifier channels}

Our results from Theorems~\ref{thm:thermal-optimal-degrad} and
\ref{thm:PI-Gauss-degrad-caps} allow for making more general statements,
applicable to broadband scenarios considered in prior works for other
capacities \cite{GLMS03,GGLMSY04,Guha04}. Let the Gibbs observable be $\hat
{E}_{m}$, as given in \eqref{eq:photon-num-op-freqs}, and suppose that the
energy constraint is $P\in\lbrack0,\infty)$. Suppose that the channel is an
$m$-mode channel consisting of $m$ parallel pure-loss channels $\mathcal{L}%
_{\eta}$, each\ with the same transmissivity $\eta\in\left[  1/2,1\right]  $.
Then for $\hat{E}_{m}$\ and such an $m$-mode channel, the conditions of
Theorems~\ref{thm:thermal-optimal-degrad} and \ref{thm:PI-Gauss-degrad-caps}%
\ are satisfied, so that the energy-constrained quantum and private capacities
are given by%
\begin{equation}
\sum_{j=1}^{m}g(\eta N_{j}(\beta))-g((1-\eta)N_{j}(\beta)),
\end{equation}
where%
\begin{equation}
N_{s}(\beta)\equiv1/(e^{\beta\omega_{s}}-1),
\end{equation}
and $\beta$ is chosen such that $P=\sum_{j=1}^{m}N_{j}(\beta)$, so that the
energy constraint is satisfied. A similar statement applies to $m$ parallel
quantum-limited amplifier channels each having the same gain $\kappa\geq1$. In
this case, the conditions of Theorems~\ref{thm:thermal-optimal-degrad} and
\ref{thm:PI-Gauss-degrad-caps}\ are satisfied, so that the energy-constrained
quantum and private capacities are given by%
\begin{equation}
\sum_{j=1}^{m}g(\kappa N_{j}(\beta)+\kappa-1)-g(\left[  \kappa-1\right]
\left[  N_{j}(\beta)+1\right]  ),
\end{equation}
where $N_{j}(\beta)$ is as defined above and $\beta$ is chosen to satisfy
$P=\sum_{j=1}^{m}N_{j}(\beta)$.

Theorems~\ref{thm:thermal-optimal-degrad} and \ref{thm:PI-Gauss-degrad-caps}%
\ can be applied indirectly to a more general scenario. Let $m=k+l$, where $k$
and $l$ are positive integers. Suppose that the channel consists of $k$
pure-loss channels $\mathcal{L}_{\eta_{i}}$, each of transmissivity $\eta
_{i}\in\lbrack1/2,1]$, and $l$ quantum-limited amplifier channels
$\mathcal{A}_{\kappa_{j}}$, each of gain $\kappa_{j}$ for $j\in\left\{
1,\ldots,l\right\}  $. In this scenario,
Theorems~\ref{thm:thermal-optimal-degrad} and \ref{thm:PI-Gauss-degrad-caps}
apply to the individual channels, so that we know that a thermal state is the
optimal input to each of them for a fixed input energy. The task is then to
determine how to allocate the energy such that the resulting capacity is
optimal. Let $P$ denote the total energy budget, and suppose that a particular
allocation $\{\{N_{i}\}_{i=1}^{k},\{M_{j}\}_{j=1}^{l}\}$ is made such that%
\begin{equation}
P=\sum_{i=1}^{k}\omega_{i}N_{i}+\sum_{j=1}^{m}\omega_{j}M_{j}.
\end{equation}
Then Theorems~\ref{thm:thermal-optimal-degrad} and
\ref{thm:PI-Gauss-degrad-caps} apply to the scenario when the allocation is
fixed and imply that the resulting quantum and private capacities are equal
and given by%
\begin{multline}
\sum_{i=1}^{k}g(\eta_{i}N_{i})-g((1-\eta_{i})N_{i})\\
+\sum_{j=1}^{l}g(\kappa M_{j}+\kappa-1)-g(\left[  \kappa-1\right]  \left[
M_{j}+1\right]  ).
\end{multline}
However, we can then optimize this expression with respect to the energy
allocation, leading to the following constrained optimization problem:%
\begin{multline}
\max_{\{\{N_{i}\}_{i=1}^{k},\{N_{j}\}_{j=1}^{l}\}}\sum_{i=1}^{k}g(\eta
_{i}N_{i})-g((1-\eta_{i})N_{i})\\
+\sum_{j=1}^{l}g(\kappa M_{j}+\kappa-1)-g(\left[  \kappa-1\right]  \left[
M_{j}+1\right]  ),
\end{multline}
such that%
\begin{equation}
P=\sum_{i=1}^{k}\omega_{i}N_{i}+\sum_{j=1}^{m}\omega_{j}M_{j}.
\end{equation}
This problem can be approached using Lagrange multiplier methods, and in some
cases handled analytically, while others need to be handled numerically. Many
different scenarios were considered already in \cite{GLMS03}, to which we
point the interested reader. However, we should note that \cite{GLMS03} was
developed when the formulas above were only conjectured to be equal to the
capacity and not proven to be so.

\section{Discussion of non-Gaussian channels and other energy constraints}

\label{sec:non-Gaussian}

We stress again here that the framework for energy-constrained quantum and private capacity given in this paper applies to more general situations beyond bosonic Gaussian channels with photon number constraints, just as the frameworks from \cite{H03,H04,HS06,HS12} do for other kinds of communication capacities. All we require for our theorems to apply is that the energy observable be a Gibbs observable (Definition~\ref{def:Gibbs-obs}) and the channel satisfy the finite output-entropy condition (Condition~\ref{cond:finite-out-entropy}).

There are some interesting cases to consider. For example, it could be the case that the initial state of the environment in a thermal channel has not reached its equilibrium state and is in a non-Gaussian state different from a thermal state. This kind of channel is related to those presented and analyzed recently in \cite{SW17}. If the initial environment state is an approximate thermal state (has trace distance close to a thermal state of a certain photon number), then the tools of the present paper, as well as those detailed in the recent work \cite{SWAT17}, could be used to estimate the quantum and private capacity of this non-equilibrium thermal channel.

Even in the bosonic setting, one could also consider other energy observables besides photon number observables. For example, one could consider the square or higher powers of the photon number observables, which might be relevant in situations in which the transmitter is highly sensitive to higher photon numbers. Using the square of the photon number would penalize higher photon numbers more severely than the typical photon number constraint.

For the case of the pure-loss  and quantum-limited amplifier channels, we can give concrete bounds for energy-constrained quantum and private capacity, using $\hat{n}^2$ as the energy observable, by employing an idea put forward recently in \cite[Remark~21]{SWAT17}, as well as other arguments. Suppose that the Gibbs observable is now $\hat{n}^2$ (the square of the photon number operator). We first discuss how to obtain an upper bound on the capacities. Due to these channels being degradable, Theorem~\ref{thm:-energy-constr-q-p-cap} applies, and it suffices to consider optimizing the single-copy energy-constrained coherent information in \eqref{eq:coh-info-1-letter}, subject to the constraint $\operatorname{Tr}\{\hat{n}^2 \rho\} \leq P$ on the input state $\rho$. By concavity of the square-root function, and due to the fact that $\operatorname{Tr}\{\hat{n}^2 \rho\} = \sum_{n=0}^\infty p(n) n^2$ for some probability distribution $p(n)$, it follows that every state satisfying $\operatorname{Tr}\{\hat{n}^2 \rho\} \leq P$ also satisfies $\operatorname{Tr}\{\hat{n} \rho\} \leq \sqrt{P}$. Setting $N_S = \sqrt{P}$, we then find that the formulas in \eqref{eq:q-cap-loss} and \eqref{eq:amp-q-cap-constrained} with this value of $N_S$ give an upper bound on the capacities.

To find a lower bound on the capacities, we can optimize the single-copy energy-constrained coherent information in \eqref{eq:coh-info-1-letter} with respect to all Gaussian state inputs.
The coherent information can also be rewritten in this case as a particular conditional entropy (see \cite{cmp2005dev,WPG07}) that  is a function of the input state $\rho$. 
Now we apply an argument from \cite[Remark~21]{SWAT17}. The pure-loss and quantum-limited amplifier channels and their complementary channels are phase-covariant, meaning that a unitary phase operator $e^{i \hat{n} \phi}$ acting on the input state commutes with the channels and acts as a unitary phase operator on the output. Since for any state $\rho$, the value of $\operatorname{Tr}\{\hat{n}^2 \rho\}$ is unchanged by applying a random phase to $\rho$, but the conditional entropy does not decrease under this operation and the phase-randomized state becomes number diagonal, it suffices to perform the optimization over all states that are both Gaussian (by assumption) and number diagonal. For a single-mode state, the only such possibility is a thermal state. Finally, since the functions in \eqref{eq:q-cap-loss} and \eqref{eq:amp-q-cap-constrained} are equal to the coherent informations of these channels when sending in a thermal state of mean photon number $N_S$, and these functions are monotone increasing with respect to $N_S$, it suffices to pick a thermal state $\theta(N_S)$ of mean photon number $N_S$ that meets the energy constraint $P$ with equality. Since $\operatorname{Tr}\{\hat{n}^2 \theta(N_S)\} = N_S(2N_S + 1)$, by solving the equation $N_S(2N_S + 1) = P$, we find that $N_S = \tfrac{1}{4}(\sqrt{1+8P} - 1)$ and then lower bounds on the energy-constrained quantum and private capacities of these channels are given by \eqref{eq:q-cap-loss} and \eqref{eq:amp-q-cap-constrained} with this value of~$N_S$.

In summary, the energy-constrained quantum and private capacities of the pure-loss channel, with Gibbs observable set to $\hat{n}^2$, are bounded from above by the function in \eqref{eq:q-cap-loss} evaluated at $N_S = \sqrt{P}$ and from below by the formula in \eqref{eq:q-cap-loss} evaluated at $N_S = \tfrac{1}{4}(\sqrt{1+8P} - 1)$. One obtains related bounds for the energy-constrained quantum and private capacities of the quantum-limited amplifier channel, with Gibbs observable set to $\hat{n}^2$, by evaluating the formula in \eqref{eq:amp-q-cap-constrained} at the same values of~$N_S$.

We note that similar arguments can be employed for any power of the photon number operator $\hat{n}$, and one would find bounds for the energy-constrained capacities in a similar way.

Going beyond the bounds given above, it is an intriguing open question to identify the actual capacities with these modified Gibbs observables. In this scenario, it is not clear that the extremality of Gaussian states \cite{WGC06} applies, because the constraint is not on the covariance matrix, but rather on the expectation of a four-point correlator.

\section{Conclusion\label{sec:conclusion}}

This paper has provided a general theory of energy-constrained quantum and
private communication over quantum channels. We defined several communication
tasks (Section~\ref{sec:energy-constrained-caps}), and then established ways
of converting a code for one task to that of another task
(Section~\ref{sec:code-conversions}). These code conversions have implications
for capacities, establishing non-trivial relations between them
(Section~\ref{sec:cap-imps}). We showed that the regularized,
energy-constrained coherent information is achievable for entanglement
transmission with an average energy constraint, under the assumption that the
energy observable is of the Gibbs form (Definition~\ref{def:Gibbs-obs}) and
the channel satisfies the finite-output entropy condition
(Condition~\ref{cond:finite-out-entropy}). We then proved that the various
quantum and private capacities of degradable channels are equal and
characterized by the single-letter, energy-constrained coherent information
(Section~\ref{sec:degradable-channels}). We finally applied our results to
Gaussian channels and recovered some results already known in the literature
in addition to establishing new ones.

We have left open the question of proving that the regularized, energy-constrained private information, defined in \eqref{eq:en-const-priv-info}, is an achievable rate for private communication. We think that this should certainly be possible. One particular method for doing so would be to extend the results of \cite{Hay15} such that they apply to coding with energy constraints and over infinite-dimensional channels. Other approaches, like that along the lines of \cite{H03,H04} for public classical communication, in conjunction with the method from \cite{ieee2005dev,1050633}, could also be employed. The first approach mentioned above could potentially lead to a simpler proof of Theorem~\ref{thm:coh-info-ach} (regarding quantum communication instead of private communication), but the details remain to be worked out. 

Going forward from here, a great challenge is to establish a general theory of
energy-constrained private and quantum communication with a limited number of
channel uses. Recent progress in these scenarios without energy constraints
\cite{TBR15,WTB16}\ suggests that this might be amenable to analysis. Another
question is to identify and explore other physical systems, beyond bosonic
channels, to which the general framework could apply. It could be interesting
to explore generalizations of the results and settings from
\cite{B05,PhysRevLett.95.260503,PhysRevA.75.022331,G13,GG16} regarding fermionic Gaussian channels.
 A more
particular question we would like to see answered is whether concavity of
coherent information of degradable channels could hold in settings beyond that
considered in Proposition~\ref{prop:concave-degrad}. We suspect that an
approximation argument along the lines of that given in the proof of
\cite[Proposition~1]{HS10} should make this possible. We also think it should be possible to establish an equality in Theorem~\ref{thm:-energy-constr-p-cap-regularized}, but we leave this for
future endeavors.

\begin{acknowledgments}
We are grateful to Saikat Guha, Alexander Holevo, Anna Kuznetsova, and Maksim Shirokov for discussions related to
this paper. We thank the anonymous referees for several comments that helped to improve the paper. MMW acknowledges the NSF under Award No.~CCF-1350397, as well as the Office of Naval Research. HQ is
supported by the Air Force Office of Scientific Research, the Army Research
Office, and the National Science Foundation.
\end{acknowledgments}

\appendix

\section{Minimum fidelity and minimum entanglement fidelity}

The following proposition states that a quantum code with good minimum
fidelity implies that it has good minimum entanglement fidelity with
negligible loss in parameters. This was first established in \cite{BKN98}\ and
reviewed in \cite{KW04}. Here we follow the proof available in \cite{Wat16},
which therein established a relation between trace distance and diamond
distance between an arbitrary channel and the identity channel.

\begin{proposition}
\label{prop:min-fid-to-min-ent-fid}Let $\mathcal{C}:\mathcal{T}(\mathcal{H}%
)\rightarrow\mathcal{T}(\mathcal{H})$ be a quantum channel with
finite-dimensional input and output. Let $\mathcal{H}^{\prime}$ be a Hilbert
space isomorphic to $\mathcal{H}$. If%
\begin{equation}
\min_{|\phi\rangle\in\mathcal{H}}\langle\phi|\mathcal{C}(|\phi\rangle
\langle\phi|)|\phi\rangle\geq1-\varepsilon, \label{eq:min-fid-starting-pnt}%
\end{equation}
then%
\begin{equation}
\min_{|\psi\rangle\in\mathcal{H}^{\prime}\otimes\mathcal{H}}\langle
\psi|(\operatorname{id}_{\mathcal{H}^{\prime}}\otimes\mathcal{C})(|\psi
\rangle\langle\psi|)|\psi\rangle\geq1-2\sqrt{\varepsilon},
\end{equation}
where the optimizations are with respect to state vectors.
\end{proposition}

\begin{proof}
The inequality in \eqref{eq:min-fid-starting-pnt} implies that the following
inequality holds for all state vectors $|\phi\rangle\in\mathcal{H}$:%
\begin{equation}
\langle\phi|\left[  |\phi\rangle\langle\phi|-\mathcal{C}(|\phi\rangle
\langle\phi|)\right]  |\phi\rangle\leq\varepsilon.
\end{equation}
By the inequalities in \eqref{eq:F-vd-G}, this implies that%
\begin{equation}
\left\Vert |\phi\rangle\langle\phi|-\mathcal{C}(|\phi\rangle\langle
\phi|)\right\Vert _{1}\leq2\sqrt{\varepsilon},
\label{eq:min-to-ent-t-norm-bnd}%
\end{equation}
for all state vectors $|\phi\rangle\in\mathcal{H}$. We will show that%
\begin{equation}
\left\vert \langle\phi|\left[  |\phi\rangle\langle\phi^{\bot}|-\mathcal{C}%
(|\phi\rangle\langle\phi^{\bot}|)\right]  |\phi^{\bot}\rangle\right\vert
\leq2\sqrt{\varepsilon}, \label{eq:ortho-pairs-fid-bnd}%
\end{equation}
for every orthonormal pair $\left\{  |\phi\rangle,|\phi^{\bot}\rangle\right\}
$ of state vectors in$~\mathcal{H}$. Set%
\begin{equation}
|w_{k}\rangle\equiv\frac{|\phi\rangle+i^{k}|\phi^{\bot}\rangle}{\sqrt{2}}%
\end{equation}
for $k\in\{0,1,2,3\}$. Then it follows that%
\begin{equation}
|\phi\rangle\langle\phi^{\bot}|=\frac{1}{2}\sum_{k=0}^{3}i^{k}|w_{k}%
\rangle\langle w_{k}|. \label{eq:ortho-pairs-char}%
\end{equation}
Consider now that%
\begin{align}
&  \left\vert \langle\phi|\left[  |\phi\rangle\langle\phi^{\bot}%
|-\mathcal{C}(|\phi\rangle\langle\phi^{\bot}|)\right]  |\phi^{\bot}%
\rangle\right\vert \nonumber\\
&  \leq\left\Vert |\phi\rangle\langle\phi^{\bot}|-\mathcal{C}(|\phi
\rangle\langle\phi^{\bot}|)\right\Vert _{\infty}\\
&  \leq\frac{1}{2}\sum_{k=0}^{3}\left\Vert |w_{k}\rangle\langle w_{k}%
|-\mathcal{C}(|w_{k}\rangle\langle w_{k}|)\right\Vert _{\infty}\\
&  \leq\frac{1}{4}\sum_{k=0}^{3}\left\Vert |w_{k}\rangle\langle w_{k}%
|-\mathcal{C}(|w_{k}\rangle\langle w_{k}|)\right\Vert _{1}\\
&  \leq2\sqrt{\varepsilon}.
\end{align}
The first inequality follows from the characterization of the operator norm as
$\left\Vert A\right\Vert _{\infty}=\sup_{|\phi\rangle,|\psi\rangle}\left\vert
\langle\phi|A|\psi\rangle\right\vert $, where the optimization is with respect
to state vectors $|\varphi\rangle$ and $|\psi\rangle$. The second inequality
follows from substituting \eqref{eq:ortho-pairs-char} and applying the
triangle inequality and homogeneity of the $\infty$-norm. The third inequality
follows because the $\infty$-norm of a traceless Hermitian operator is bounded
from above by half of its trace norm \cite[Lemma~4]{AE05}. The final
inequality follows from applying \eqref{eq:min-to-ent-t-norm-bnd}.

Let $|\psi\rangle\in\mathcal{H}^{\prime}\otimes\mathcal{H}$ be an arbitrary
state vector. All such state vectors have a Schmidt decomposition of the
following form:%
\begin{equation}
|\psi\rangle=\sum_{x}\sqrt{p(x)}|\zeta_{x}\rangle\otimes|\varphi_{x}\rangle,
\end{equation}
where $\{p(x)\}_{x}$ is a probability distribution and $\{|\zeta_{x}%
\rangle\}_{x}$ and $\{|\varphi_{x}\rangle\}_{x}$ are orthonormal sets,
respectively. Then consider that%
\begin{align}
&  1-\langle\psi|(\operatorname{id}_{\mathcal{H}^{\prime}}\otimes
\mathcal{C})(|\psi\rangle\langle\psi|)|\psi\rangle\nonumber\\
&  =\langle\psi|(\operatorname{id}_{\mathcal{H}^{\prime}}\otimes
\operatorname{id}_{\mathcal{H}}-\operatorname{id}_{\mathcal{H}^{\prime}%
}\otimes\mathcal{C})(|\psi\rangle\langle\psi|)|\psi\rangle\nonumber\\
&  =\langle\psi|(\operatorname{id}_{\mathcal{H}^{\prime}}\otimes\left[
\operatorname{id}_{\mathcal{H}}-\mathcal{C}\right]  )(|\psi\rangle\langle
\psi|)|\psi\rangle\nonumber\\
&  =\sum_{x,y}p(x)p(y)\langle\varphi_{x}|\left[  |\varphi_{x}\rangle
\langle\varphi_{y}|-\mathcal{C}(|\varphi_{x}\rangle\langle\varphi
_{y}|)\right]  |\varphi_{y}\rangle.
\end{align}
Now applying the triangle inequality and \eqref{eq:ortho-pairs-fid-bnd}, we
find that%
\begin{align}
&  1-\langle\psi|(\operatorname{id}_{\mathcal{H}^{\prime}}\otimes
\mathcal{C})(|\psi\rangle\langle\psi|)|\psi\rangle\nonumber\\
&  =\left\vert \sum_{x,y}p(x)p(y)\langle\varphi_{x}|\left[  |\varphi
_{x}\rangle\langle\varphi_{y}|-\mathcal{C}(|\varphi_{x}\rangle\langle
\varphi_{y}|)\right]  |\varphi_{y}\rangle\right\vert \nonumber\\
&  \leq\sum_{x,y}p(x)p(y)\left\vert \langle\varphi_{x}|\left[  |\varphi
_{x}\rangle\langle\varphi_{y}|-\mathcal{C}(|\varphi_{x}\rangle\langle
\varphi_{y}|)\right]  |\varphi_{y}\rangle\right\vert \nonumber\\
&  \leq2\sqrt{\varepsilon}.
\end{align}
This concludes the proof.
\end{proof}

\bibliographystyle{unsrt}
\bibliography{Ref}

\end{document}